\documentclass[11pt]{article}
\usepackage[utf8]{inputenc}
\usepackage{amsmath}
\usepackage{graphicx}
\usepackage[colorinlistoftodos]{todonotes}
\usepackage[colorlinks=true, allcolors=purple]{hyperref}
\usepackage{amsthm,thm-restate}
\usepackage{amssymb}
\usepackage{bbm}
\usepackage{url}
\usepackage{amsrefs}
\usepackage{amsfonts}
\usepackage{mathtools}
\usepackage{marginnote}
\usepackage[T1]{fontenc}
\usepackage{soul}
\usepackage{tikz}
\usepackage{tikz-cd}
\usepackage{caption}
\usepackage{setspace}
\usepackage{enumitem}
\usepackage{physics}
\usepackage{lmodern}
\usepackage{authblk}
\usetikzlibrary{intersections}
\usepackage{comment}
\usepackage{subcaption}

\usepackage{scalerel}
\usepackage{stackengine,wasysym}

\usetikzlibrary{patterns}
\usepackage{genyoungtabtikz}
\usetikzlibrary {patterns.meta}

\pgfdeclarepattern{
name=hatch,
parameters={\hatchsize,\hatchangle,\hatchlinewidth},
bottom left={\pgfpoint{-.1pt}{-.1pt}},
top right={\pgfpoint{\hatchsize+.1pt}{\hatchsize+.1pt}},
tile size={\pgfpoint{\hatchsize}{\hatchsize}},
tile transformation={\pgftransformrotate{\hatchangle}},
code={
\pgfsetlinewidth{\hatchlinewidth}
\pgfpathmoveto{\pgfpoint{-.1pt}{-.1pt}}
\pgfpathlineto{\pgfpoint{\hatchsize+.1pt}{\hatchsize+.1pt}}
\pgfpathmoveto{\pgfpoint{-.1pt}{\hatchsize+.1pt}}
\pgfpathlineto{\pgfpoint{\hatchsize+.1pt}{-.1pt}}
\pgfusepath{stroke}
}
}

\tikzdeclarepattern{
name=mystars,
type=uncolored,
bounding box={(-5pt,-5pt) and (5pt,5pt)},
tile size={(\tikztilesize,\tikztilesize)},
parameters={\tikzstarpoints,\tikzstarradius,\tikzstarrotate,\tikztilesize},
tile transformation={rotate=\tikzstarrotate},
defaults={
points/.store in=\tikzstarpoints,points=5,
radius/.store in=\tikzstarradius,radius=3pt,
rotate/.store in=\tikzstarrotate,rotate=0,
tile size/.store in=\tikztilesize,tile size=10pt,
},
code={
\pgfmathparse{180/\tikzstarpoints}\let\a=\pgfmathresult
\fill (90:\tikzstarradius) \foreach \i in {1,...,\tikzstarpoints}{
-- (90+2*\i*\a-\a:\tikzstarradius/2) -- (90+2*\i*\a:\tikzstarradius)
} -- cycle;
}
}

\tikzset{
hatch size/.store in=\hatchsize,
hatch angle/.store in=\hatchangle,
hatch line width/.store in=\hatchlinewidth,
hatch size=5pt,
hatch angle=0pt,
hatch line width=.5pt,
}

\usepackage{geometry}
\usetikzlibrary{arrows.meta}
\usetikzlibrary{decorations.markings}

\newtheorem{thm}{Theorem}[section]
\newtheorem{lem}[thm]{Lemma}
\newtheorem{prop}[thm]{Proposition}

\theoremstyle{definition}

\newtheorem{rmk}[thm]{Remark}
\newtheorem{cor}[thm]{Corollary}

\DeclareMathOperator{\artanh}{artanh}

\newcommand\reallywidetilde[1]{\ThisStyle{%
  \setbox0=\hbox{$\SavedStyle#1$}%
  \stackengine{-.1\LMpt}{$\SavedStyle#1$}{%
    \stretchto{\scaleto{\SavedStyle\mkern.2mu\AC}{.44\wd0}}{.4\ht0}%
  }{O}{c}{F}{T}{S}%
}}

\begin{document}
\reversemarginpar

\title{On a class of orthogonal-invariant quantum spin systems on the complete graph}
\author{Kieran Ryan}
\affil{\href{mailto:kieran.ryan@univie.ac.at}{kieran.ryan@univie.ac.at}\\
Fakultät für Mathematik, University of Vienna,\\ Oskar-Morgenstern-Platz 1, 1090 Wien}
\date{}
\maketitle
\begin{abstract}
    We study a two-parameter family of quantum spin systems on the complete graph, which is the most general model invariant under the complex orthogonal group. In spin $S=\frac{1}{2}$ it is equivalent to the XXZ model, and in spin $S=1$ to the bilinear-biquadratic Heisenberg model. The paper is motivated by the work of Björnberg, whose model is invariant under the (larger) complex general linear group. In spin $S=\frac{1}{2}$ and $S=1$ we give an explicit formula for the free energy for all values of the two parameters, and for spin $S>1$ for when one of the parameters is non-negative. This allows us to draw phase diagrams, and determine critical temperatures. For spin $S=\frac{1}{2}$ and $S=1$, we give the left and right derivatives as the strength parameter of a certain magnetisation term tends to zero, and we give a formula for a certain total spin observable, and heuristics for the set of extremal Gibbs states in several regions of the phase diagrams, in the style of a recent paper of Björnberg, Fröhlich and Ueltschi.
    
    The key technical tool is expressing the partition function in terms of the irreducible characters of the symmetric group and the Brauer algebra. The parameters considered include, and go beyond, those for which the systems have probabilistic representations as interchange processes.
\end{abstract}

\section{Introduction}
Quantum spin systems and their phase transitions have been studied widely. Mermin and Wagner showed that the quantum Heisenberg model \cite{merminwagner1966}, the classical Heisenberg model \cite{Mermin2} and many other models with continuous symmetry cannot have a phase transition in dimensions 1 and 2. Dyson, Lieb and Simon \cite{Dyson:1976zw} showed a transition for a large class of models on $\mathbb{Z}^d$, $d \ge 3$, including the quantum Heisenberg antiferromagnet with spin $\ge 1$. A phase transition on $\mathbb{Z}^d$, $d \ge 3$ for the ferromagnet remains unproved.
Tóth \cite{Toth_1993} and Aizenman-Nachtergale \cite{aizenman1994} showed that the spin $\frac{1}{2}$ Heisenberg ferro- and anti-ferromagnet (respectively) have probabilistic representations as weighted interchange processes. Other spin systems have been studied with probabilistic representations, and interchange processes have been studied widely in their own right; see, for example, \cite{Angel2003RandomIP}, \cite{hammond2013sharp}, \cite{Koteck__2016}, \cite{Schramm2004CompositionsOR}.  

The free energy of the spin $\frac{1}{2}$ Heisenberg ferromagnet on the complete graph was determined by Tóth \cite{Toth1990} and Penrose \cite{Penrose1991}. This was extended by Björnberg \cite{Bj_rnberg_2016} to a class of spin $S \in \frac{1}{2} \mathbb{N}$ models, with Hamiltonian equal to the sum of transposition operators. The model's probabilistic representation is that of the interchange process, where Tóth's weighting of $2^{\# \mathrm{cycles}}$ is replaced by $(2S+1)^{\# \mathrm{cycles}}$. 

Motivated by \cite{Bj_rnberg_2016}, we give in Theorem \ref{thm:main-1.1} the free energy, on the complete graph and in spins $S=\frac{1}{2}$ and $1$, of a model with Hamiltonian (\ref{eqn:fullhamiltonian}) given by linear combinations of the sum of transposition operators, and the sum of certain projection operators. For spins $S>1$ we can apply a similar strategy to give in Theorem \ref{thm:main-higher-spins} the free energy in the case that one of the parameters of the Hamiltonian is non-negative. In spin $\frac{1}{2}$ the model is equivalent to the Heisenberg XXZ model (Hamiltonian (\ref{eqn:spin1/2hamiltonian})). In spin $1$ it is equivalent to the bilinear-biquadratic Heisenberg model (Hamiltonian (\ref{eqn:spin1hamiltonian})), which is also known as the most general $SU_2(\mathbb{C})$-invariant spin $1$ model (here $SU_2(\mathbb{C})$-invariance means invariance under the action of $SU_2(\mathbb{C})$ generated by the spin-operators). We give a full phase diagram in the two parameters of the Hamiltonian in the $S=\frac{1}{2}$ and $1$ cases, and half of the diagram for $S>1$ (the region where we have the free energy), giving the points of phase transitions in finite temperature, and ground state behaviour. These phase diagrams differ notably in shape from those on $\mathbb{Z}^d$, since the complete graph is not bipartite. Indeed, no phase transition is observed for the spin $\frac{1}{2}$ Heisenberg antiferromagnet, in contrast to $\mathbb{Z}^d$,  \cite{Dyson:1976zw}, and the expected phase diagram for the spin $1$ model in $Z^d$ differs from ours on the complete graph - see Ueltschi's work \cite{Ueltschi_2013} and \cite{Ueltschi_2015}. In spins $\frac{1}{2}$ and $1$ we give in Theorems \ref{thm:mag} and \ref{thm:total-spin} respectively expressions for a magnetisation and a total spin observable. These are motivated by corresponding results of \cite{Bj_rnberg_2016} and \cite{Bj_rnberg_Ueltschi_Frolich_2019} respectively. In the style of \cite{Bj_rnberg_Ueltschi_Frolich_2019}, we also give a heuristic argument for the structure of the set of extremal Gibbs states, for several regions of the phase diagrams. See the discussions in Sections \ref{section:PD-1/2} and \ref{section:PD-1}.

The Hamiltonian in \cite{Bj_rnberg_2016} is $GL(\theta)$-invariant, where $\theta=2S+1$, which allows it to be studied using the representation theory of the symmetric group (here and for the rest of the paper, by $G$-invariance, we mean that $G$ acts on tensor space by $G \ni g \mapsto g^{\otimes n}$, and the Hamiltonian in question commutes with this action). Björnberg's key technical step is to express the partition function of the model in terms of the irreducible characters of the symmetric group. Our Hamiltonian is only $O(\theta)$-invariant, which requires us to look for more tools, as the symmetric group is not sufficient. (In fact, any $O(\theta)$-invariant pair-interaction Hamiltonian must be of the form (\ref{eqn:fullhamiltonian})). The key representation-theoretic step in finding the free energy is to express the partition function in terms of the irreducible characters of both the symmetric group and the Brauer algebra. Indeed, the Brauer algebra was introduced by Brauer \cite{brauer}, as the algebra of invariants of the action of the orthogonal group on tensor space. A key technical step in our proofs is solving the problem of finding when the Brauer algebra - symmetric group branching coefficients are non-zero; we have a general solution for this problem in spins $S=\frac{1}{2}, 1$ in Propositions \ref{prop:P2} and \ref{prop:P3}. For higher spins more work is needed to answer the problem fully, and handle the remaining parts of the phase diagram.

This paper is a continuation of several papers which analyse quantum spin systems and their interchange processes using representation theory (including \cite{Bj_rnberg_2016}). Alon and Kozma \cite{Alon_2013} estimate the number of cycles of length $k$ in the unweighted interchange process, on any graph. Berestycki and Kozma \cite{berestycki2012cycle} give an exact formula for the same on the complete graph, and study the phase transitions present. In \cite{alon2018meanfield} Alon and Kozma give a formula for the magnetisation of the $2^{\# cycles}$ weighted process (equivalent to the spin $\frac{1}{2}$ ferromagnet) on any graph, which simplifies greatly in the mean-field. 

The model we study was introduced by Ueltschi \cite{Ueltschi_2013}, generalising Tóth \cite{Toth_1993} and Aizenmann-Nachtergaele \cite{aizenman1994}. Ueltschi showed, for certain values of the parameters, equivalence with a weighted interchange process with ``reversals''. For these parameters, the model and interchange process have been studied on $\mathbb{Z}^d$ \cite{Bj_rnberg_Ueltschi_2015}, \cite{crawford2016emptiness}, trees \cite{betzetal}, \cite{Bj_rnberg_Ueltschi_trees_2018}, \cite{hammond2018critical}, graphs of bounded degree \cite{muhlbacher2019critical}, and the complete graph  \cite{Bj_rnberg_Milos_Lees_Kotovski_2019}, \cite{Bj_rnberg_Ueltschi_Frolich_2019}, the latter of which computes many observables. Our methods allow us to deal with all values of the parameters, not just those for which the probabilistic representation holds. The implications of our results for this interchange process seem to be limited to the following. In \cite{Bj_rnberg_Milos_Lees_Kotovski_2019}, the authors show that the transition time is independent of the parameter giving the ratio of ``crosses'' and ``reversals''; our results indicate the same is most probably true for the weighted process.

In Section \ref{section:modelnresults}, we describe our model and precisely state our results. In Section \ref{section:reptheory} we give an introduction to the Brauer algebra. In Section \ref{section:proofofmainthm} we prove our main result, Theorem \ref{thm:main-1.1}, modulo the key ingredients Propositions \ref{prop:P2} and \ref{prop:P3} which are proved in Section \ref{section:restrictiontheorems}. In Section \ref{section:other-thms} we give the free energy in higher spins, and prove our magnetisation and total spin results. In Section \ref{section:phasediagrams} we prove certain results on the analyticity of the free energy, which follow from Theorem \ref{thm:main-1.1}, and give calculations which back up our interpretation of the phase diagrams.

\subsection{Models and results}\label{section:modelnresults}

Let $S^1, S^2, S^3$ denote the usual spin-operators, satisfying the relations:
\begin{equation*}
    \begin{split}
        &[S^1, S^2] = iS^3, \ [S^2, S^3] = iS^1, \ [S^3, S^1] = iS^2,
        \\ 
        &(S^1)^2 + (S^2)^2 + (S^3)^2 = S(S+1)\textbf{id},
    \end{split}
\end{equation*}
with $i = \sqrt{-1}$. For each $S \in \frac{1}{2}\mathbb{N}$ we use the standard spin S representation, with $S^j$, $j=1,2,3$, Hermitian matrices acting on $\mathcal{H} = \mathbb{C}^\theta$, $\theta = 2S+1$. We fix an orthonormal basis of $\mathcal{H}$ given by eigenvectors $|a\rangle$ of $S^3$, with eigenvalues $a \in \{-S, \dots, S\}$.

Let $G = (V,E)$ be the complete graph on $n$ vertices. At each $x \in V$, we fix a copy $\mathcal{H}_x = \mathcal{H}$, and let the space $\mathcal{H}_V = \otimes_{x \in V}\mathcal{H}_x = \otimes_{x \in V}\mathbb{C}^\theta_x$. Now an orthonormal basis of $\mathcal{H}_V$ is given by vectors $|\textbf{a}\rangle = \otimes_{x \in V}|a_x\rangle$, where each $a_x \in \{-S, \dots, S\}$. If $A$ is an operator acting on $\mathcal{H}$, we define $A_x = A \otimes \textbf{id}_{V \setminus \{x\}}$ acting on $\mathcal{H}_V$ (ie. $A_x(v_1\otimes\cdots\otimes v_x\otimes\cdots\otimes v_n)=v_1\otimes\cdots\otimes Av_x\otimes\cdots\otimes v_n$ for all $v_1\dots v_n\in V$, extending linearly). 

We define $T_{x,y}$ to be the transposition operator, and $Q_{x,y}$ to be a certain projection operator, first on $\mathcal{H}_x \otimes \mathcal{H}_y$:
\begin{equation}\label{eqn:defnofT,Q}
    \begin{split}
        T_{x,y} |a_x, a_y \rangle &= |a_y, a_x \rangle,\\
        \langle a_x, a_y| Q_{x,y} |b_x, b_y \rangle &= \delta_{a_x, a_y} \delta_{b_x, b_y},
    \end{split}
\end{equation}
for basis vectors $a_x, b_x$ of $\mathcal{H}_x$ and $a_y,b_y$ of $\mathcal{H}_y$ as appropriate. We then identify $T_{x,y}$ with $T_{x,y} \otimes \textbf{id}_{V \setminus \{x,y\} }$, and $Q_{x,y}$ similarly. Let our Hamiltonian be defined as:
\begin{equation}\label{eqn:fullhamiltonian}
    \begin{split}
        H = H(n,\theta, L_1, L_2) 
        &= - 
        \sum_{x,y} \left( L_1 T_{x,y} + L_2 Q_{x,y}
        \right),
    \end{split}
\end{equation}
where $L_1, L_2 \in \mathbb{R}$, and the sum is over all pairs of vertices $x,y \in V$. We define the partition function as 
\begin{equation}\label{eqn:Z}
    Z_{n,\theta}(L_1, L_2) = \mathrm{tr}(e^{-\frac{1}{n} H(n,\theta, L_1, L_2) }).
\end{equation}
Note that usually we would write $e^{-\frac{\beta}{n}H}$, for inverse temperature $\beta$, but without loss of generality this $\beta$ can be incorporated into $L_1$ and $L_2$. One could think of $\beta$ as being expressed by the norm of the vector $(L_1, L_2) \in \mathbb{R}^2$. The factor $\frac{1}{n}$ compensates for the fact that on the complete graphs there are order $n^2$ interactions (as opposed to $\mathbb{Z}^d$, where the number of interactions is proportional to the volume).

We have the following results. The first and main result, Theorem \ref{thm:main-1.1}, gives the free energy when the spin $S=\frac{1}{2}$ or 1, that is, $\theta = 2,3$. Theorem \ref{thm:main-higher-spins} gives the free energy for all $\theta\ge2$, but only for $L_2\ge0$; its proof is very similar to that of \ref{thm:main-1.1}. Theorems \ref{thm:mag} and \ref{thm:total-spin} give formulae for a certain magnetisation and a certain total spin, respectively. In Theorems \ref{thm:PD-1/2}, \ref{thm:PD-1} and \ref{thm:PD-higher} we analyse the free energies of Theorems \ref{thm:main-1.1} and \ref{thm:main-higher-spins} and discuss the phase diagrams that they produce.

\begin{thm}\label{thm:main-1.1}
    For $\theta =2,3$, the free energy of the model with Hamiltonian given by (\ref{eqn:fullhamiltonian}) is:
    \begin{equation}\label{eqn:FE-main}
        \lim_{n \to \infty} \frac{1}{n} \log Z_{n,\theta}(L_1, L_2) 
        = 
        \max_{(x,y) \in \Delta^*_{\theta} } \left[ \frac{1}{2} \left(
        (L_1 + L_2) \sum_{i=1}^\theta x_i^2 - L_2 y_1^2  \right) 
        - \sum_{i=1}^\theta x_i \log(x_i) \right],
    \end{equation}
    where 
    \begin{equation}\label{eqn:Delta*}
    \begin{split}
        \Delta^*_{2} &= \{ (x,y)=(x_1,x_2,y_1)\in[0,1]^3 \ | 
                \ x_1 \ge x_2, \ x_1 + x_2 =1, \ 
                \ 0 \le y_1 \le x_1 - x_2 \},\\
        \Delta^*_{3} &= \{ (x,y)=(x_1,x_2,x_3,y_1)\in[0,1]^4 \ |
                \ x_1 \ge x_2 \ge x_3, 
                \ x_1 + x_2 + x_3 =1,
                \ 0 \le y_1 \le x_1 - x_3 \}.\\        
    \end{split}
    \end{equation}
\end{thm}

From hereon in, we label the function being maximised by: 
\begin{equation}\label{eqn:phi-main}
    \phi = \phi_{\theta, L_1, L_2}(x,y) = \frac{1}{2} \left[(L_1 + L_2) \sum_{i=1}^\theta x_i^2 - L_2 \sum_{i=1}^\theta y_i^2 \right] - \sum_{i=1}^\theta x_i \log(x_i).
\end{equation}

Our result for higher spins covers only the range $L_2\ge0$. As noted earlier in the introduction, this restriction is due to our only having a partial solution to determining when the Brauer algebra - symmetric group branching coefficients are non-zero, when the multiplicative parameter of the Brauer algebra $\theta=2S+1$ is greater than $3$ (see Section \ref{section:reptheory} for the definition of the multiplicative parameter, and (\ref{eqn:defn_of_b_lamrho}) for the branching coefficients). 

\begin{thm}\label{thm:main-higher-spins}
    Let $\theta\ge 2$, and assume $L_2 \ge 0$. Then the free energy of the model with Hamiltonian (\ref{eqn:fullhamiltonian}) is:
    \begin{equation}\label{eqn:FE-higher-spins}
        \lim_{n \to \infty} \frac{1}{n} \log Z_{n,\theta}(L_1, L_2) 
        = 
        \max_{x \in \Delta_\theta} \left[
        \frac{L_1 + L_2}{2}
        \sum_{i=1}^\theta x_i^2 
        - \sum_{i=1}^\theta x_i \log(x_i) \right],
    \end{equation}
    where 
    $\Delta_{\theta} = \{ x \in [0,1]^\theta \ | 
                \ x_i \ge x_{i+1}\ge 0, \ \sum_{i=1}^\theta x_i=1 \}.$
\end{thm}

Notice that for $\theta =2,3$ this theorem is consistent with Theorem \ref{thm:main-1.1}, since in (\ref{eqn:FE-main}), for $L_2\ge0$, we must set $y_1=0$. Note that when $L_2=0$, Theorems \ref{thm:main-1.1} and \ref{thm:main-higher-spins} recover Björnberg's result (Theorem 1.1 from \cite{Bj_rnberg_2016}), with our $L_2$ equal to the $\beta$ from that paper.

We can give two additional results, both for $\theta = 2,3$. The first gives the free energy of the model when we add a certain magnetisation term with a real strength parameter $h$, and its left and right derivatives at $h=0$. Let us modify the Hamiltonian (\ref{eqn:fullhamiltonian}):
\begin{equation}\label{eqn:fullH-mag}
    \begin{split}
        H_h = H_h(n,\theta, L_1, L_2, W) 
        &= - 
        \sum_{x,y} \left( L_1 T_{x,y} + L_2 Q_{x,y}
        \right)
        - h \sum_{x}W_x,
    \end{split}
\end{equation}
where $h$ is real, and $W$ is a $\theta\times\theta$ skew-symmetric matrix (ie. $W^T=-W$), with eigenvalues $1,-1$ for $\theta =2$, and $1,0,-1$ for $\theta =3$. In this theorem and the next, the limitation of $W$ being skew-symmetric is a technical one arising from the methodology. For $\theta=2$, we will see that the model is the Heisenberg XXZ model, and we can take $W_x=2S_x^2$ (indeed this is the only option). For $\theta=3$ the model is unitarily equivalent to the bilinear-biquadratic Heisenberg model (\ref{eqn:spin1hamiltonian}). If one adds to this bilinear-biquadratic Hamiltonian a magnetisation term in the $S^k$ direction, $k=1,2,3$, one obtains exactly (\ref{eqn:fullH-mag}) under the unitary equivalence. See Remark \ref{rmk:appendix-spin-conj} for full details. So, in our interpretation of phase diagrams, we will think of this magnetisation term as that in the $S^2$ direction (or $S^1,S^3$ when $\theta=3$). This theorem relates to Theorem \ref{thm:main-1.1} as Theorem 4.1 from \cite{Bj_rnberg_2016} does to Theorem 1.1 from that paper.
\begin{thm}\label{thm:mag}
    Let $\theta = 2,3$, and let $Z_{n,\theta}(L_1, L_2,h) = \mathrm{tr}[e^{-\frac{1}{n}H_h}]$. The free energy of the model with Hamiltonian $H_h$ (\ref{eqn:fullH-mag}) is given by:
    \begin{equation*}
        \Phi = \Phi_{\theta}(L_1, L_2, h)=
        \lim_{n \to \infty} \frac{1}{n} \log Z_{n,\theta}(L_1, L_2,h)= 
        \max_{(x,y) \in \Delta^*_{\theta} } \left[ 
        \phi_{\theta, L_1, L_2}(x,y)
        +|h|y_1 \right].
    \end{equation*}
    Further, the left and right derivatives of this free energy with respect to $h$, at $h=0$, are given by:
    \begin{equation*}
        \frac{\partial \Phi}{\partial h}\Big|_{h\searrow 0}=
        y_1^{\uparrow}, 
        \qquad
        \frac{\partial \Phi}{\partial h}\Big|_{h\nearrow 0}=
        y_1^{\downarrow}, 
    \end{equation*}
    where $(x^\uparrow, y^\uparrow)$ is the maximiser of $\phi$ which maximises $y_1$, and $(x^\downarrow, y^\downarrow)$ the one which minimises $y_1$.
\end{thm}
Note that $y_1^\uparrow$, $y_1^\downarrow$ depend on $L_1, L_2$, and we will show that both are zero when $L_1, L_2$ are small. 
We now return to the model (\ref{eqn:fullhamiltonian}) with no magnetisation term. We give a total spin observable $\langle e^{\frac{h}{n}\sum_xW_x}\rangle$, for skew-symmetric matrices $W$. As above, we can think of $W$ as $2S^2$ when $\theta=2$, and (once the Hamiltonian is transformed) as $S^k$, $k=1,2,3$ when $\theta=3$. Although only $h\in\mathbb{R}$ seems physically relevant, we will see that the theorem holds for $h\in\mathbb{C}$. This theorem is an equivalent of Theorems 2.1-2.3 from \cite{Bj_rnberg_Ueltschi_Frolich_2019}, and its proof follows similar lines of reasoning to that of Theorem 2.3 from that paper, aided by the Brauer algebra technology that we develop in this paper.
\begin{thm}\label{thm:total-spin}
    Let $\theta = 2,3$, $h \in \mathbb{C}$, and $W$ skew-symmetric with eigenvalues $1,-1$ for $\theta =2$, and $1,0,-1$ for $\theta =3$. Assume that the function $\phi_{\theta,L_1,L_2}$ has a unique maximiser $(x^*,y^*)\in\Delta_\theta^*$. Then with $H$ the Hamiltonian from (\ref{eqn:fullhamiltonian}), for $L_2\neq0$,
    \begin{equation*}
        \langle e^{\frac{h}{n}\sum_xW_x} \rangle:=
        \lim_{n\to\infty}\frac{\mathrm{tr}[e^{\frac{h}{n}\sum_xW_x} e^{-\frac{1}{n}H}]}{Z_{n,\theta}(L_1,L_2)}=
        \left\{\begin{array}{ll}
		    \cosh(hy_1^*), & \mbox{if } \theta=2, \\
		    \frac{\sinh(hy_1^*)}{hy_1^*}, & \mbox{if } \theta=3.
	    \end{array}\right.
    \end{equation*}
\end{thm}
The quantity $\cosh(hy_1^*)$ is related to Ising spin-flip symmetry, see below in Proposition \ref{prop:proofy-bits-1/2} and the discussion thereafter; and the quantity $\frac{\sinh(hy_1^*)}{hy_1^*}$ is related to $SU(2)$ (or $O(3)$) symmetry, see in Proposition \ref{prop:proofy-bits-1} and the discussion thereafter, and in \cite{Bj_rnberg_Ueltschi_Frolich_2019}. The case $L_2=0$ has $GL(3)$ symmetry, see \cite{Bj_rnberg_Ueltschi_Frolich_2019}. We remark that in the context of Theorem \ref{thm:total-spin}, in particular when $\phi_{\theta,L_1,L_2}$ has a unique maximiser $(x^*,y^*)\in\Delta_\theta^*$, we have that $y_1^*$ is the same as the magnetisation $y_1^\uparrow$ (and indeed $y_1^\downarrow$) from Theorem \ref{thm:mag}. 

We now state our results in terms of two well known models, the spin $\frac{1}{2}$ Heisenberg XXZ model, and the spin $1$ bilinear-biquadratic Heisenberg model.

\subsection{Phase diagram for spin 1/2}\label{section:PD-1/2}
Let $S=\frac{1}{2}$, so $\theta=2$. We consider the Hamiltonian of the XXZ model, which will be equivalent to (\ref{eqn:fullhamiltonian}). Let
\begin{equation}\label{eqn:spin1/2hamiltonian}
    \begin{split}
        H' = - \left( 
        \sum_{x,y} K_1 S_x^1 S_y^1 + K_2 S_x^2 S_y^2 + K_1 S_x^3 S_y^3 
        \right).
    \end{split}
\end{equation}
Our result Theorem \ref{thm:main-1.1} leads us to the following theorem, which will give information about the phase diagram of this model. See Figures \ref{fig:GS-1/2} and \ref{fig:PD-1/2}.
\begin{thm}\label{thm:PD-1/2}
    The free energy of the model with Hamiltonian (\ref{eqn:spin1/2hamiltonian}) is analytic everywhere in the $(K_1,K_2)$ plane, except the half-lines $K_1=4, K_2\le4$ and $K_2=4,K_1\le4$, where it is differentiable, but not twice-differentiable, and the half-line $K_1=K_2\ge4$, where it is continuous, but not differentiable. 
\end{thm}

Let us also formalise what we will prove about the magnetisation and finite volume ground states, which will aid our discussion of the phase diagram below.
\begin{prop}\label{prop:proofy-bits-1/2}
    Consider the spin $\frac{1}{2}$ Heisenberg XXZ model (\ref{eqn:spin1/2hamiltonian}).
    \begin{enumerate}
        \item The magnetisation $y_1^\uparrow$ of Theorem \ref{thm:mag} is positive if and only if $K_2>4$, $K_2\ge K_1$, and is zero elsewhere.
        \item 
        \begin{enumerate}
            \item For $K_2>0$, $K_2>K_1$, the finite volume ground states are spanned by the two product states $\bigotimes_{x \in V}(|\frac{1}{2}\rangle\pm i|-\frac{1}{2}\rangle)$, (where $i=\sqrt{-1}$);
            \item For $K_1>0$, $K_1>K_2$, the finite (even) volume ground state is the vector (\ref{eqn:dimer-states-1/2}).
        \end{enumerate}
    \end{enumerate}
\end{prop}

Theorem \ref{thm:PD-1/2} splits the plane into three regions of analyticity, which we identify as three phases of the model. We label the region $K_1\le4, K_2\le4$ disordered (illustrated in block pink in Figure \ref{fig:PD-1/2}); the maximiser of the function $\phi$ in (\ref{eqn:FE-main}) is constant in this region, and it maximises the entropy term (the logarithms) of $\phi$. 

We label the region $K_2>4, K_2> K_1$ the Ising phase (illustrated in dotted yellow in Figure \ref{fig:PD-1/2}). It includes the half-line $K_1=0, K_2\ge4$, where the model is the supercritical classical Ising model, and further we will show the free energy in this region is independent of $K_1$ (it is perhaps slightly surprising that $K_2$ dominates to such a complete extent). There are two finite volume ground states in this region, the product states $\bigotimes_{x \in V}(|\frac{1}{2}\rangle\pm i|-\frac{1}{2}\rangle)$. Further, for small values of $h$, adding $-h\sum_xS_x^2$ to the Hamiltonian as in (\ref{eqn:fullH-mag}) forces a unique ground state, $\bigotimes_{x\in V}(|\frac{1}{2}\rangle + i|-\frac{1}{2}\rangle)$ when $h>0$ and $\bigotimes_{x\in V}(|\frac{1}{2}\rangle - i|-\frac{1}{2}\rangle)$ when $h<0$. The magnetisation $y_1^\uparrow$ in the $S^2$ direction from Theorem \ref{thm:mag} is positive. 

The authors of \cite{Bj_rnberg_Ueltschi_Frolich_2019} give a heuristic argument that points towards an expected structure of the set of extremal Gibbs states $\Psi_\beta$ at inverse temperature $\beta$ for several models on $\mathbb{Z}^d$, $d\ge3$. The extremal Gibbs states in infinite volume are not well-defined on the complete graph, so the working is by analogy. Specifically, their heuristics indicate two expected equalities: first, that
\begin{equation}\label{eqn:BFU-heuristic}
    \lim_{\Lambda_n\to\mathbb{Z}^d}\langle e^{\frac{h}{|\Lambda|}\sum_iS^{2}_i} \rangle_{\Lambda} = \int_{\Psi_\beta}e^{h\langle S_0^{2}\rangle_\psi}d\mu(\psi),
\end{equation}
where $d\mu$ is the measure on Gibbs states corresponding to the symmetric state, $S_0^{2}$ is the spin operator at the lattice site $0$, and the left hand side is the limit of successively larger boxes $\Lambda\in\mathbb{Z}^d$; and second that 
\begin{equation}\label{eqn:BFU-heur-2}
    \lim_{n\to\infty}\langle e^{\frac{h}{n}\sum_iS^{2}_i} \rangle_{G} =
    \lim_{\Lambda_n\to\mathbb{Z}^d}\langle e^{\frac{h}{|\Lambda_n|}\sum_iS^{2}_i} \rangle_{\Lambda_n},
\end{equation}
where the left hand term is the observable on the complete graph. The left hand side of (\ref{eqn:BFU-heur-2}) is computed rigorously on the complete graph, and then, with the expected structure of $\Psi_\beta$ inserted, the right hand side of (\ref{eqn:BFU-heuristic}) is rigorously computed, and the two are shown to be the same. This working is not a proof either of the expected equalities (\ref{eqn:BFU-heuristic}), (\ref{eqn:BFU-heur-2}) or of the expected structure of $\Psi_\beta$, but it points towards all three statements holding true. We can take the same approach for our models in several of the phases, with one small difference. We expect that the equality (\ref{eqn:BFU-heur-2}) holds for the complete graph models that we study in this paper on the left hand side, and on the right hand side models on other non-bipartite graphs, for example the triangular lattice, or the models on $\mathbb{Z}^d$, $d\ge3$, with nearest and next-to-nearest neighbour interactions. For the Ising phase here, we can argue that the extremal Gibbs states in the Ising phase are expected to be indexed by $\{\pm e_2\}$, where $e_2$ is the second basis vector in $\mathbb{R}^3$. Indeed, with magnetisation $y_1^*=\langle S_0^2\rangle_{e_2}$ and $\Psi_\beta=\{\pm e_2\}$, the right hand side of this equality is $\cosh(hy_1^*)$; the left hand side is the same by Theorem \ref{thm:total-spin}.

We label the region $K_1>4, K_1>K_2$ the $XY$ phase (illustrated in hatched blue in Figure \ref{fig:PD-1/2}). We expect the $S_x^1$ and $S_x^3$ terms to dominate, and the extremal Gibbs states to be labelled by $\vec{a}\in\mathbb{S}^1$ in the $1-3$ directions. The magnetisation $y_1^\uparrow$ in the $S_x^2$ direction (from Theorem \ref{thm:mag}) is zero in this region, which is consistent with this picture. Similarly to the Ising phase, with $\Psi_\beta=\mathbb{S}^1$, the right hand side of (\ref{eqn:BFU-heuristic}) is 1, as is the left hand side of (\ref{eqn:BFU-heur-2}) by Theorem \ref{thm:total-spin}, so again we are encouraged in our labelling of the extremal states, and of (\ref{eqn:BFU-heuristic}) and (\ref{eqn:BFU-heur-2}). Equivalent calculations in the $S^1$ direction are done in \cite{Bj_rnberg_Ueltschi_Frolich_2019}. Interestingly, the ground state in finite (even) volume is the vector
\begin{equation}\label{eqn:dimer-states-1/2}
    \sum_{\underline{m},\underline{m}'}
    \bigotimes_{i=1}^{n/2}
    \sum_{a=\frac{-1}{2},\frac{1}{2}}
    |a_{m_i}, a_{m'_i}\rangle,
\end{equation}
where the sum is over all possible pairings $(\underline{m},\underline{m}')$ of the vertices of $V$ (that is, $(\underline{m},\underline{m}')=((m_1,\dots,m_{\frac{n}{2}}),(m'_1,\dots,m'_{\frac{n}{2}}))$, with $\underline{m}\cup\underline{m}'=V$, $\underline{m}\cap\underline{m}'=\varnothing$).

Note that the line $K_1=K_2\ge0$ is the ferromagnetic Heisenberg model, and the extremal Gibbs states are expected to be labelled by $\vec{a}\in\mathbb{S}^2$. Here we can prove that the magnetisation $y_1^\uparrow>0$ iff $K_1=K_2>4$. The heuristics of (\ref{eqn:BFU-heuristic}) and (\ref{eqn:BFU-heur-2}) are given in Theorem 2.1 of \cite{Bj_rnberg_Ueltschi_Frolich_2019}.

The transitions from the disordered phase to each of the Ising and $XY$ phases are second order, and the transition from Ising to $XY$ is first order. The ground state phase diagram is illustrated in Figure \ref{fig:GS-1/2}, and the finite temperature phase diagram is illustrated in Figure \ref{fig:PD-1/2}.

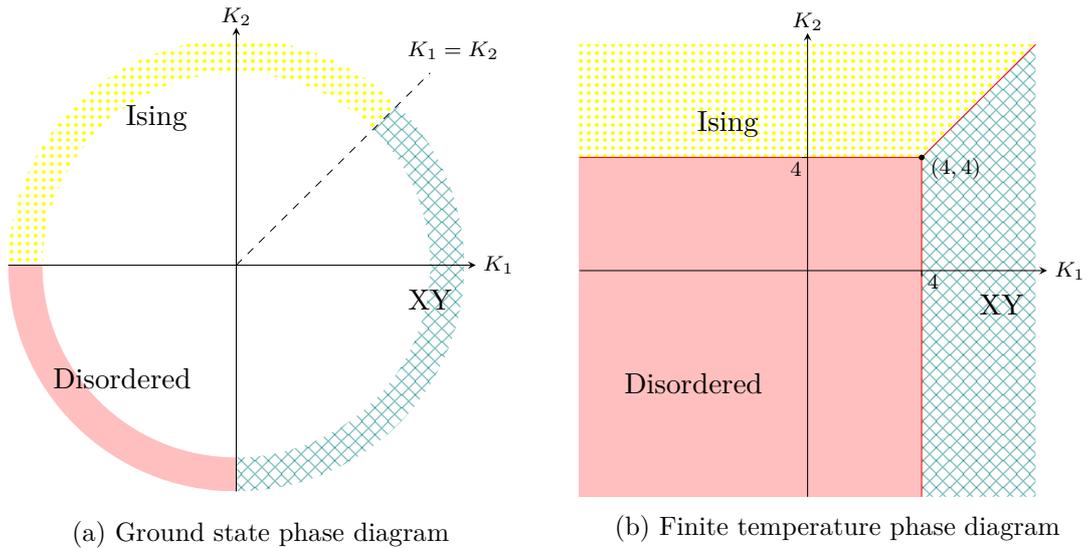
\begin{figure}[h]
    \begin{subfigure}{0.5\textwidth}
    \centering
    \begin{tikzpicture}[scale=1.5]
    
    \path[pattern = dots, pattern color = yellow, opacity = 1] (0,0) -- (-2, 0) -- (-2,2) -- (2,2) -- (0,0);
    \path[pattern = dots, pattern color = yellow, opacity = 1] (0,0) -- (-2, 0) -- (-2,2) -- (2,2) -- (0,0);
    \path[pattern = dots, pattern color = yellow, opacity = 1] (0,0) -- (-2, 0) -- (-2,2) -- (2,2) -- (0,0);
    \path[pattern = dots, pattern color = yellow, opacity = 1] (0,0) -- (-2, 0) -- (-2,2) -- (2,2) -- (0,0);
    \path[pattern = dots, pattern color = yellow, opacity = 1] (0,0) -- (-2, 0) -- (-2,2) -- (2,2) -- (0,0);
    \path[pattern = dots, pattern color = yellow, opacity = 1] (0,0) -- (-2, 0) -- (-2,2) -- (2,2) -- (0,0);
    \path[pattern = dots, pattern color = yellow, opacity = 1] (0,0) -- (-2, 0) -- (-2,2) -- (2,2) -- (0,0);
    \path[pattern = dots, pattern color = yellow, opacity = 1] (0,0) -- (-2, 0) -- (-2,2) -- (2,2) -- (0,0);
    \path[pattern = dots, pattern color = yellow, opacity = 1] (0,0) -- (-2, 0) -- (-2,2) -- (2,2) -- (0,0);
    \path[pattern = dots, pattern color = yellow, opacity = 1] (0,0) -- (-2, 0) -- (-2,2) -- (2,2) -- (0,0);
    \path[pattern = dots, pattern color = yellow, opacity = 1] (0,0) -- (-2, 0) -- (-2,2) -- (2,2) -- (0,0);
    \path[pattern = dots, pattern color = yellow, opacity = 1] (0,0) -- (-2, 0) -- (-2,2) -- (2,2) -- (0,0);
    \path[pattern color=teal, pattern=hatch, hatch size = 5pt, hatch line width = .5pt, opacity = 0.7] (0,0) -- (0,-2) -- (2,-2) -- (2,2) -- (0,0);
    \path[fill = pink, opacity = 1] (0,0) -- (-2, 0) -- (-2,-2) -- (0,-2) -- (0,0);
    \fill[white] (0,0) circle (1.7cm);
    \path[fill=white]
    (0,2.1) -- (2.1,2.1) -- (2.1,-2.1) -- (-2.1,-2.1) -- (-2.1,2.1) -- (0,2.1) -- (0,2) arc [radius=2, start angle=90, delta angle=360] -- (0,2.1) -- cycle;
    
    \draw[->, >=stealth, line width=0.3pt](-2, 0) -- (2.1,0);
    \draw[->, >=stealth, line width=0.3pt](0,-2) -- (0,2.1);
    \draw[-, >=stealth, dashed, line width=0.3pt](0,0) -- (1.7,1.7);
    
    \draw[font=\fontsize{8}{10}] (0, 2.2) node {$K_2$};
    \draw[font=\fontsize{8}{10}] (2.3, 0) node {$K_1$};
    \draw[font=\fontsize{8}{10}] (1.9, 1.9) node {$K_1=K_2$};
    
    \draw[font=\fontsize{8}{10}] (-0.7, 1.3) node {Ising};
    \draw[font=\fontsize{8}{10}] (1.7, -0.3) node {XY};
    \draw[font=\fontsize{8}{10}] (-1, -1) node {Disordered};

    \end{tikzpicture}
    \caption{Ground state phase diagram}
    \label{fig:GS-1/2}
    \end{subfigure}
    \begin{subfigure}{0.5\textwidth}
    \centering
    \begin{tikzpicture}[scale=1.5]
    
    \path[pattern = dots, pattern color = yellow, opacity = 1] (0,0) -- (-2, 0) -- (-2,2) -- (2,2) -- (0,0);
    \path[pattern = dots, pattern color = yellow, opacity = 1] (0,0) -- (-2, 0) -- (-2,2) -- (2,2) -- (0,0);
    \path[pattern = dots, pattern color = yellow, opacity = 1] (0,0) -- (-2, 0) -- (-2,2) -- (2,2) -- (0,0);
    \path[pattern = dots, pattern color = yellow, opacity = 1] (0,0) -- (-2, 0) -- (-2,2) -- (2,2) -- (0,0);
    \path[pattern = dots, pattern color = yellow, opacity = 1] (0,0) -- (-2, 0) -- (-2,2) -- (2,2) -- (0,0);
    \path[pattern = dots, pattern color = yellow, opacity = 1] (0,0) -- (-2, 0) -- (-2,2) -- (2,2) -- (0,0);
    \path[pattern = dots, pattern color = yellow, opacity = 1] (0,0) -- (-2, 0) -- (-2,2) -- (2,2) -- (0,0);
    \path[pattern = dots, pattern color = yellow, opacity = 1] (0,0) -- (-2, 0) -- (-2,2) -- (2,2) -- (0,0);
    \path[pattern = dots, pattern color = yellow, opacity = 1] (0,0) -- (-2, 0) -- (-2,2) -- (2,2) -- (0,0);
    \path[pattern = dots, pattern color = yellow, opacity = 1] (0,0) -- (-2, 0) -- (-2,2) -- (2,2) -- (0,0);
    \path[pattern = dots, pattern color = yellow, opacity = 1] (0,0) -- (-2, 0) -- (-2,2) -- (2,2) -- (0,0);
    \path[pattern = dots, pattern color = yellow, opacity = 1] (0,0) -- (-2, 0) -- (-2,2) -- (2,2) -- (0,0);
    \path[pattern color=teal, pattern=hatch, hatch size = 5pt, hatch line width = .5pt, opacity = 0.7] (0,0) -- (0,-2) -- (2,-2) -- (2,2) -- (0,0);
    \path[fill = pink, opacity = 1] (1,1) -- (-2, 1) -- (-2,-2) -- (1,-2) -- (1,1);
    
    \draw[-, >=stealth, line width=0.3pt, red] (1,1) -- (2,2);
    \draw[-, >=stealth, line width=0.3pt, red](-2,1) -- (1,1) -- (1, -2);
    \draw[->, >=stealth, line width=0.3pt](-2, 0) -- (2.1,0);
    \draw[->, >=stealth, line width=0.3pt](0,-2) -- (0,2.1);
    
    \draw[font=\fontsize{8}{10}] (-0.1, 0.9) node {$4$};
    \draw[font=\fontsize{8}{10}] (1.1, -0.1) node {$4$};
    \draw[font=\fontsize{8}{10}] (0, 2.2) node {$K_2$};
    \draw[font=\fontsize{8}{10}] (2.3, 0) node {$K_1$};

    \draw[fill=black] (1,1) circle (0.02cm);
    \draw[font=\fontsize{8}{10}] (1.3, 0.9) node {$(4,4)$};
    
    \draw[-, >=stealth, line width=0.3pt](0,1) -- (-0.05, 1);
    \draw[-, >=stealth, line width=0.3pt](1,0) -- (1,-0.05);
    
    \draw[font=\fontsize{8}{10}] (-0.7, 1.3) node {Ising};
    \draw[font=\fontsize{8}{10}] (1.7, -0.3) node {XY};
    \draw[font=\fontsize{8}{10}] (-1, -1) node {Disordered};
    
    \end{tikzpicture}
    \caption{Finite temperature phase diagram}
    \label{fig:PD-1/2}
    \end{subfigure}
    \caption{On the left, the ground state phase diagram for the Spin $\frac{1}{2}$ Heisenberg XXZ model with Hamiltonian (\ref{eqn:spin1/2hamiltonian}). The line $K_1=K_2\ge0$ gives the Heisenberg ferromagnet. On the right, the phases at finite temperature, where varying temperature is given by varying the modulus $||(K_1,K_2)||$. Transitions between phases (points of non-analyticity of the free energy) shown in red lines. The magnetisation in the $S^2$ direction ($y_1^\uparrow$ from Theorem \ref{thm:mag}) is strictly positive in the interior of the Ising phase, and on the open half-line $K_1=K_2>4$, and is zero elsewhere.}

\end{figure}

\subsection{Phase diagram for spin 1}\label{section:PD-1}
In spin $1$, we consider the bilinear-biquadratic Heisenberg model:
\begin{equation}\label{eqn:spin1hamiltonian}
        H'' =  - \left( 
        \sum_{x,y} J_1 S_x \cdot S_y + J_2 (S_x \cdot S_y)^2
        \right).
\end{equation}
Our Theorem \ref{thm:main-1.1} leads us to the following theorem, which will give information about the phase diagram of this model. See Figures \ref{fig:GS-1} and \ref{fig:PD-1}. We can rigorously analyse the free energy in the $J_2> J_1$ half of the phase diagram; on the other half we have partial results and numerical simulations to support Remark \ref{rmk:PD-1}.
\begin{thm}\label{thm:PD-1}
    Within the region $J_2> J_1$ of the $(J_1,J_2)$ plane, the free energy of the model with Hamiltonian (\ref{eqn:spin1hamiltonian}) is analytic everywhere, except the half-line $J_2=\log16,J_1\le\log16$, where it is continuous, but not differentiable. 
\end{thm}
\begin{rmk}\label{rmk:PD-1}
    We strongly suspect that Theorem \ref{thm:PD-1} extends to the following: that the free energy of the model with Hamiltonian (\ref{eqn:spin1hamiltonian}) is analytic everywhere in the $(J_1,J_2)$ plane, apart from the half-lines $J_2=\log16,J_1\le\log16$ and $J_1=J_2\ge\log16$, where it is continuous, but not differentiable, and a curve (that we label $\mathcal{C}$) made up of the half-line $J_2=2J_1-3\le 3/2$ and a curve connecting the points $(\frac{9}{4}, \frac{3}{2})$ and $(\log16,\log16)$, which (as a function of $J_1$) is convex, with gradient in $[2,3]$. It is unclear whether it is analytic on the half-line $J_1=0$, $J_2\le -3$. 
\end{rmk}
Let us make clear what we will prove towards Remark \ref{rmk:PD-1} and the following discussion of the phase diagram of the model.
\begin{prop}\label{prop:proofy-bits-1}
    Consider the bilinear-biquadratic Heisenberg model with Hamiltonian (\ref{eqn:spin1hamiltonian}).
    \begin{enumerate}
        \item The region $\mathcal{A}$ of the $J_1-J_2$ plane where the point $(x,y)=((\frac{1}{3},\frac{1}{3},\frac{1}{3}),(0,0,0))$ is a maximiser of $\phi_{3,J_1,J_2}$ (\ref{eqn:phi-J1J2}) (the equivalent of $\phi_{3,L_1,L_2}$ (\ref{eqn:phi-main}) when we change variables appropriately) is closed and convex, and its boundary is the half-line $J_1\le J_2=\log16$, and a curve $\mathcal{C}$ as described in Remark \ref{rmk:PD-1}. 
        \item The magnetisation term $y_1^\uparrow$ of Theorem \ref{thm:mag} is zero in the region $\mathcal{A}$ and the region $J_2\ge\log16$, $J_2>J_1$, and is positive in the region $J_1\ge J_2$, strictly to the right of the curve $\mathcal{C}$. 
        \item 
        \begin{enumerate}
            \item For $J_2>0$, $J_2>J_1$, the finite (even) volume ground state is the vector (\ref{eqn:dimer-states-1});
            \item For $J_1>0$, $J_1>J_2$, the finite volume ground states are those vectors invariant under $S_n$, and killed by all $P_{x,y}$, which include the product states $\bigotimes_{x\in V}|a\rangle$, where $a_{0}^2 - a_{1}a_{-1}=0$;
            \item For $\frac{1}{2}J_2<J_1<0$, the finite volume ground states are those vectors are spanned by the vectors (\ref{eqn:partial-ferro-state}).
        \end{enumerate}
    \end{enumerate}
\end{prop}
Let us now discuss the phase diagram. Theorem \ref{thm:PD-1}, Remark \ref{rmk:PD-1}, and Proposition \ref{prop:proofy-bits-1} divide the $(J_1,J_2)$ plane into four regions, which we label as phases of the model. We label the region $\mathcal{A}$ (illustrated in block pink in Figure \ref{fig:PD-1}) the disordered phase. The boundary of this region is made up of the half-line $J_2=\log16, J_1 \le \log16$ and the curve $\mathcal{C}$. The maximiser of the function $\phi$ from (\ref{eqn:FE-main}) is constant in this region, and maximises the entropy term. 

We label the region of phase space to the right of the red line in Figure \ref{fig:PD-1}, within the region $J_1> J_2$, $J_1>0$, ferromagnetic (illustrated in dotted yellow in Figure \ref{fig:PD-1}) (in fact, for large $||(J_1,J_2)||$, this region is that which is expected to be ferromagnetic in $\mathbb{Z}^3$, see \cite{Ueltschi_2013}). The finite volume ground states include the product states $\bigotimes_{x\in V}|a\rangle$, where $a_{0}^2 - a_{1}a_{-1}=0$, (eg. the ferromagnetic $|1\rangle$, $|-1\rangle$, as well as $|1\rangle + |0\rangle + |-1\rangle$). The magnetisation $y_1^\uparrow$ in the $S^k$ direction, $k=1,2,3$, (Theorem \ref{thm:mag}) is positive in this phase. We expect that the extremal Gibbs states are indexed by $\vec{a}\in\mathbb{S}^2$, in which case (with $\langle S_0^k\rangle_{e_2}=y_1^*$) the right hand side of (\ref{eqn:BFU-heuristic}) equals $\frac{\sinh(hy_1^*)}{hy_1^*}$. Numerical simulations suggest that the maximiser $y_1^*$ of $\phi$ is unique in this phase, so the left hand side of (\ref{eqn:BFU-heur-2}) should be the same, by Theorem \ref{thm:total-spin}. This encourages our expectation that the extremal states are indeed $\vec{a}\in\mathbb{S}^2$, and that (\ref{eqn:BFU-heuristic}) and (\ref{eqn:BFU-heur-2}) hold true. This extends the same analysis of the $J_2=0$ case given in Theorem 2.1 of \cite{Bj_rnberg_Ueltschi_Frolich_2019}. 

We label the region of phase space $J_2>\log16, J_2> J_1$ the nematic phase (illustrated in hatched blue in Figure \ref{fig:PD-1}); we expect the $(S_x\cdot S_y)^2$ term to dominate, and the extremal Gibbs states to be given by $a\in\mathbb{RP}^2$. The magnetisation in the $S^k$ direction, $k=1,2,3$, (Theorem \ref{thm:mag}), $y_1^\uparrow$, is zero in this phase, which matches the heuristics; we would expect to get something non-zero, for example, by replacing $S^k$ with its square. Interestingly, the finite (even) volume ground state in this phase is the vector
\begin{equation}\label{eqn:dimer-states-1}
    \sum_{\underline{m},\underline{m}'}
    \bigotimes_{i=1}^{n/2}
    \sum_{a=-1}^1
    |a_{m_i}, (-1)^{a}(-a)_{m'_i}\rangle,
\end{equation}
where the sum is over all possible pairings $(\underline{m},\underline{m}')$. This is the sum over all possible products of singlet states.

The fourth phase (illustrated in checkerboard orange in Figure \ref{fig:PD-1}) occupies the region $\frac{1}{2}J_2 + \frac{3}{2} \le J_1 \le 0$. This phase is somewhat more mysterious. While the magnetisation $y_1^\uparrow$ is positive in this phase, the finite volume ground states are complicated. They depend on the ratio $\alpha=\frac{J_2}{J_1+J_2} \in [\frac{2}{3},1]$, and are spanned by vectors of the form 
\begin{equation}\label{eqn:partial-ferro-state}
    z_{(\alpha' n,(1-\alpha')n)}
    \left[
    \left( \bigotimes_{i\in V\setminus(\underline{m}\cup\underline{m}')} |a_i'\rangle \right)
    \otimes
    \left(
    \bigotimes_{i=1}^{(1-\alpha')n/2}
            \sum_{a=-1}^1
            |a_{m_i}, (-1)^{a}(-a)_{m'_i}\rangle \right)
    \right],
\end{equation}
where $\alpha'$ is a fraction with denominator $n$ close to $\alpha$, $\underline{m}, \underline{m}'$ is a pairing of $(1-\alpha')n$ of the vertices, $a'$ satisfies $(a'_0)^2-a'_1a'_{-1}=0$, and $z_{(\alpha n,(1-\alpha)n)}$ is a Young symmetriser corresponding to the partition $(\alpha n,(1-\alpha)n)$. The vector being symmetrised can be thought of as a proportion $\alpha$ of the volume being taken up by a ferromagnetic ground state, and $1-\alpha$ being taken up by a product of singlet states. However, it is difficult to interpret the vector once the Young symmetriser is applied.

The transition from the disordered phase to the nematic phase is first order; we have not been able to prove similar statements for the other transitions. The ground state phase diagram is illustrated in Figure \ref{fig:GS-1}, and the finite temperature phase diagram is illustrated in Figure \ref{fig:PD-1}.

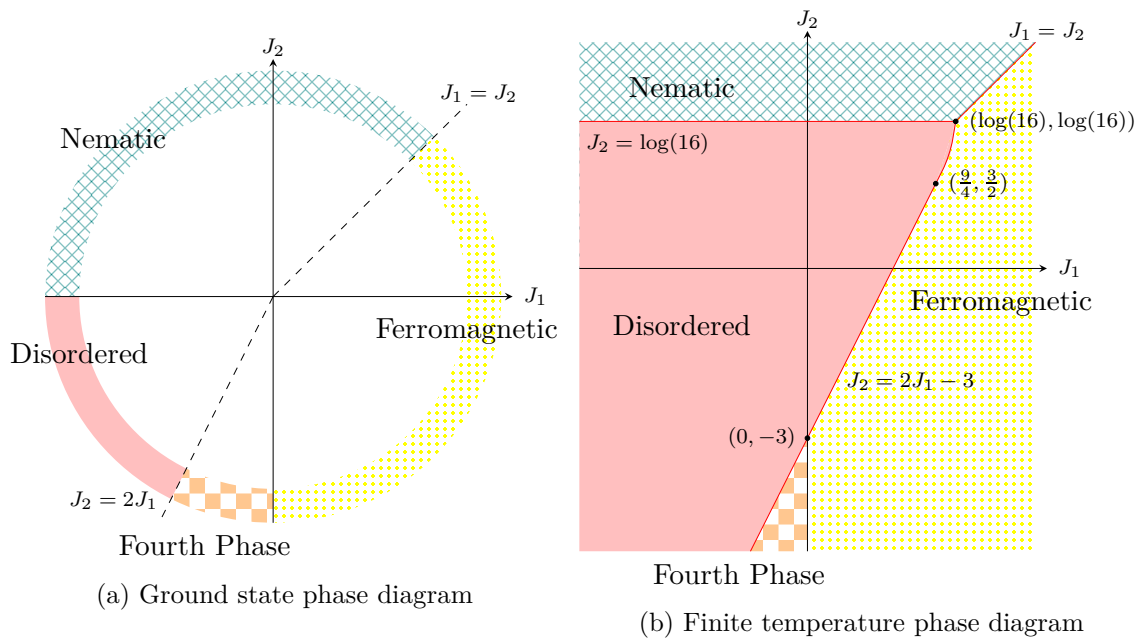
\begin{figure}[h]
    \begin{subfigure}{0.5\textwidth}
    \centering
    \begin{tikzpicture}[scale=1.5]
    
    \path[pattern color=teal, pattern=hatch, hatch size = 5pt, hatch line width = .5pt, opacity = 0.7]
    (1.1,1.1) arc [radius=1.55, start angle=45, delta angle=135]
                  -- (-2,0) arc [radius=2, start angle=180, delta angle=-135]
                  -- cycle;
    \path[pattern = dots, pattern color = yellow, opacity = 1] 
    (1.1,1.1) arc [radius=1.55, start angle=45, delta angle=-135] -- (0,-2) arc [radius=2, start angle=-90, delta angle=135] -- cycle;
    \path[pattern = dots, pattern color = yellow, opacity = 1] 
    (1.1,1.1) arc [radius=1.55, start angle=45, delta angle=-135] -- (0,-2) arc [radius=2, start angle=-90, delta angle=135] -- cycle;
    \path[pattern = dots, pattern color = yellow, opacity = 1] 
    (1.1,1.1) arc [radius=1.55, start angle=45, delta angle=-135] -- (0,-2) arc [radius=2, start angle=-90, delta angle=135] -- cycle;
    \path[pattern = dots, pattern color = yellow, opacity = 1] 
    (1.1,1.1) arc [radius=1.55, start angle=45, delta angle=-135] -- (0,-2) arc [radius=2, start angle=-90, delta angle=135] -- cycle;
    \path[pattern = dots, pattern color = yellow, opacity = 1] 
    (1.1,1.1) arc [radius=1.55, start angle=45, delta angle=-135] -- (0,-2) arc [radius=2, start angle=-90, delta angle=135] -- cycle;
    \path[pattern = dots, pattern color = yellow, opacity = 1] 
    (1.1,1.1) arc [radius=1.55, start angle=45, delta angle=-135] -- (0,-2) arc [radius=2, start angle=-90, delta angle=135] -- cycle;
    \path[pattern = dots, pattern color = yellow, opacity = 1] 
    (1.1,1.1) arc [radius=1.55, start angle=45, delta angle=-135] -- (0,-2) arc [radius=2, start angle=-90, delta angle=135] -- cycle;
    \path[pattern = dots, pattern color = yellow, opacity = 1] 
    (1.1,1.1) arc [radius=1.55, start angle=45, delta angle=-135] -- (0,-2) arc [radius=2, start angle=-90, delta angle=135] -- cycle;
    \path[pattern = dots, pattern color = yellow, opacity = 1] 
    (1.1,1.1) arc [radius=1.55, start angle=45, delta angle=-135] -- (0,-2) arc [radius=2, start angle=-90, delta angle=135] -- cycle;
    \path[pattern = dots, pattern color = yellow, opacity = 1] 
    (1.1,1.1) arc [radius=1.55, start angle=45, delta angle=-135] -- (0,-2) arc [radius=2, start angle=-90, delta angle=135] -- cycle;
    \path[pattern = dots, pattern color = yellow, opacity = 1] 
    (1.1,1.1) arc [radius=1.55, start angle=45, delta angle=-135] -- (0,-2) arc [radius=2, start angle=-90, delta angle=135] -- cycle;
    \path[pattern = dots, pattern color = yellow, opacity = 1] 
    (1.1,1.1) arc [radius=1.55, start angle=45, delta angle=-135] -- (0,-2) arc [radius=2, start angle=-90, delta angle=135] -- cycle;
    
    \path[fill=pink, opacity = 1] (0,0) -- (-2, 0) -- (-2,-2) -- (-1,-2) -- (0,0);
    \path[pattern = checkerboard, pattern color = orange, opacity = .5] (0,0) -- (0,-2) -- (-1,-2) -- (0,0);
    \path[pattern = checkerboard, pattern color = orange, opacity = .5] (0,0) -- (0,-2) -- (-1,-2) -- (0,0);
    \fill[white] (0,0) circle (1.7cm);
    \path[fill=white]
    (0,2.1) -- (2.1,2.1) -- (2.1,-2.1) -- (-2.1,-2.1) -- (-2.1,2.1) -- (0,2.1) -- (0,2) arc [radius=2, start angle=90, delta angle=360] -- (0,2.1) -- cycle;
    
    \draw[->, >=stealth, line width=0.3pt](-2, 0) -- (2.1,0);
    \draw[->, >=stealth, line width=0.3pt](0,-2) -- (0,2.1);
    \draw[-, >=stealth, dashed, line width=0.3pt](0,0) -- (1.7,1.7);
    \draw[-, >=stealth, dashed, line width=0.3pt](0,0) -- (-1, -2);
    
    \draw[font=\fontsize{8}{10}] (0, 2.2) node {$J_2$};
    \draw[font=\fontsize{8}{10}] (2.3, 0) node {$J_1$};
    \draw[font=\fontsize{8}{10}] (1.8, 1.8) node {$J_1 = J_2$};
    \draw[font=\fontsize{8}{10}] (-1.4, -1.8) node {$J_2 = 2J_1$};
    
    \draw[font=\fontsize{4}{6}] (-1.4, 1.4) node {Nematic};
    \draw[font=\fontsize{4}{6}] (1.7, -0.3) node {Ferromagnetic};
    \draw[font=\fontsize{4}{6}] (-1.7, -0.5) node {Disordered};
    \draw[font=\fontsize{4}{6}] (-0.6, -2.2) node {Fourth Phase};
    
    \end{tikzpicture}
    \caption{Ground state phase diagram}
    \label{fig:GS-1}
    \end{subfigure}
    \begin{subfigure}{0.5\textwidth}
    \centering
    \begin{tikzpicture}[scale=1.5]

    \path[pattern = dots, pattern color = yellow, opacity = 1] (0,0) -- (0,-2.5) -- (2,-2.5) -- (2,2) -- (0,0);
    \path[pattern = dots, pattern color = yellow, opacity = 1] (0,0) -- (0,-2.5) -- (2,-2.5) -- (2,2) -- (0,0);
    \path[pattern = dots, pattern color = yellow, opacity = 1] (0,0) -- (0,-2.5) -- (2,-2.5) -- (2,2) -- (0,0);
    \path[pattern = dots, pattern color = yellow, opacity = 1] (0,0) -- (0,-2.5) -- (2,-2.5) -- (2,2) -- (0,0);
    \path[pattern = dots, pattern color = yellow, opacity = 1] (0,0) -- (0,-2.5) -- (2,-2.5) -- (2,2) -- (0,0);
    \path[pattern = dots, pattern color = yellow, opacity = 1] (0,0) -- (0,-2.5) -- (2,-2.5) -- (2,2) -- (0,0);
    \path[pattern = dots, pattern color = yellow, opacity = 1] (0,0) -- (0,-2.5) -- (2,-2.5) -- (2,2) -- (0,0);
    \path[pattern = dots, pattern color = yellow, opacity = 1] (0,0) -- (0,-2.5) -- (2,-2.5) -- (2,2) -- (0,0);
    \path[pattern = dots, pattern color = yellow, opacity = 1] (0,0) -- (0,-2.5) -- (2,-2.5) -- (2,2) -- (0,0);
    \path[pattern = dots, pattern color = yellow, opacity = 1] (0,0) -- (0,-2.5) -- (2,-2.5) -- (2,2) -- (0,0);
    \path[pattern = dots, pattern color = yellow, opacity = 1] (0,0) -- (0,-2.5) -- (2,-2.5) -- (2,2) -- (0,0);
    \path[pattern = dots, pattern color = yellow, opacity = 1] (0,0) -- (0,-2.5) -- (2,-2.5) -- (2,2) -- (0,0);
    \path[pattern color=teal, pattern=hatch, hatch size = 5pt, hatch line width = .5pt, opacity = 0.7] (0,0) -- (-2, 0) -- (-2,2) -- (2,2) -- (0,0);
    \path[pattern color=teal, pattern=hatch, hatch size = 5pt, hatch line width = .5pt, opacity = 0.7] (0,0) -- (-2, 0) -- (-2,1.3) -- (1.3,1.3) -- (0,0);
    \path[pattern = checkerboard, pattern color = orange, opacity = .5] (0,-1.5) -- (0,-2.5) -- (-0.5,-2.5) -- (0,-1.5);
    \path[pattern = checkerboard, pattern color = orange, opacity = .5] (0,-1.5) -- (0,-2.5) -- (-0.5,-2.5) -- (0,-1.5);
    
    \path[fill=pink, opacity = 1](-2, 1.3) -- (-2,-2.5) -- (-0.5, -2.5) -- (1.125,0.75) -- (1.3,1.3) -- (-2,1.3); 
    \path[fill=pink, opacity = 1](1.125,0.75) --(1.25, 1) -- (1.3,1.3);
    
    \draw[->, >=stealth, line width=0.3pt](-2, 0) -- (2.1,0);
    \draw[->, >=stealth, line width=0.3pt](0,-2.5) -- (0,2.1);

    \draw[-, >=stealth, line width=0.3pt, red](2,2) -- (1.3,1.3);
    \draw[-, >=stealth, line width=0.3pt, red](-2,1.3) -- (1.3,1.3);
    \draw[-, >=stealth, line width=0.3pt, red](1.125,0.75) -- (-0.5, -2.5);
    \draw[-, >=stealth, rounded corners = 6pt, line width=0.3pt, red](1.125,0.75) --(1.25, 1) -- (1.3,1.3);
    
    \draw[fill=black] (1.3, 1.3) circle (0.02cm);
    \draw[font=\fontsize{8}{10}] (2.15, 1.3) node {$(\log(16), \log(16))$};
    \draw[fill=black] (1.125,0.75) circle (0.02cm);
    \draw[font=\fontsize{8}{10}] (1.5, 0.75) node {$(\frac{9}{4}, \frac{3}{2})$};
    \draw[fill=black] (0, -1.5) circle (0.02cm);
    
    \draw[font=\fontsize{8}{10}] (-0.4, -1.5) node {$(0,-3)$};
    \draw[font=\fontsize{8}{10}] (0, 2.2) node {$J_2$};
    \draw[font=\fontsize{8}{10}] (2.3, 0) node {$J_1$};
    \draw[font=\fontsize{8}{10}] (2.1, 2.1) node {$J_1 = J_2$};
    \draw[font=\fontsize{8}{10}] (0.9, -1) node {$J_2 = 2J_1-3$};
    \draw[font=\fontsize{8}{10}] (-1.4, 1.1) node {$J_2 = \log(16)$};
    
    \draw[font=\fontsize{4}{6}] (-1.1, 1.6) node {Nematic};
    \draw[font=\fontsize{4}{6}] (1.7, -0.3) node {Ferromagnetic};
    \draw[font=\fontsize{4}{6}] (-1.1, -0.5) node {Disordered};
    \draw[font=\fontsize{4}{6}] (-0.6, -2.7) node {Fourth Phase};
    
    \end{tikzpicture}
    \caption{Finite temperature phase diagram}
    \label{fig:PD-1}
    \end{subfigure}
    \caption{On the left, the ground state phase diagram for the Spin $1$ bilinear-biquadratic Heisenberg model with Hamiltonian (\ref{eqn:spin1hamiltonian}). On the right, the phases at finite temperature, where varying temperature is given by varying the modulus $||(J_1,J_2)||$. Transitions between phases (points of non-analyticity of the free energy) shown in red lines, and the magnetisation  in the $S^k$ direction, $k=1,2,3$, ($y_1^\uparrow$ from Theorem \ref{thm:mag}) is strictly positive in the closure of the Ferromagnetic and fourth phases, and zero elsewhere (both statements proved in the region $J_2\ge J_1$, expected for the rest of the plane).}
\end{figure}

\subsection{Higher spins}\label{section:PD-higher}
Recall that we only have the free energy of the model with Hamiltonian (\ref{eqn:fullhamiltonian}) for spins $S>1$ in the region $L_2\ge0$. We can describe this half of the phase diagram in the $(L_1,L_2)$ plane for all spins as follows. Let
\begin{equation}\label{eqn:beta_c}
    \beta_c=\beta_c(\theta)= \left\{\begin{array}{lr}
        2, & \text{for } \theta=2\\
        2\left(\frac{\theta-1}{\theta-2}\right)\log(\theta-1), & \text{for } \theta\ge3.
        \end{array}\right.
\end{equation}
\begin{thm}\label{thm:PD-higher}
    Let $S\ge\frac{1}{2}$ (so $\theta\ge2$), and let $\beta_c=\beta_c(\theta)$ from (\ref{eqn:beta_c}). Within the region $L_2>0$ of the $(L_1,L_2)$ plane, the free energy of the model with Hamiltonian (\ref{eqn:fullhamiltonian}) is analytic everywhere, except the half-line $L_1+L_2=\beta_c$, where for spin $S\ge1$ it is continuous, but not differentiable, and for spin $S=\frac{1}{2}$ it is differentiable, but not twice-differentiable.
\end{thm}
We note that this theorem is a generalisation of Theorem \ref{thm:PD-1} to spins $S\ge1$, and indeed it implies Theorem \ref{thm:PD-1}. This can be seen by a unitary transformation of the spin $1$ Hamiltonian (\ref{eqn:spin1hamiltonian}) which we describe in Section \ref{section:phasediagrams}. 

The $L_2\ge0$ part of the phase diagram can be split into three phases. The disordered phase occupies the region $L_1+L_2<\beta_c$. The maximiser of $\phi_{\theta,L_1,L_2}$ (\ref{eqn:phi-main}) maximises the entropy term (the logarithms) in this region. The region $L_1+L_2>\beta_c$, $L_2>0$ is a second phase. The finite (even) volume ground states include the vector
\begin{equation}\label{eqn:dimer-states-higher}
    \sum_{\underline{m},\underline{m}'}
    \bigotimes_{i=1}^{n/2}
    \sum_{a=-S,S}
    |a_{m_i}, a_{m'_i}\rangle,
\end{equation}
where the sum is over all possible pairings $(\underline{m},\underline{m}')$ of the vertices of $V$.
The line $L_2=0$, $L_1>0$ is the quantum interchange model of \cite{Bj_rnberg_2016}. The finite volume ground states are any vector which is invariant under the action of $S_n$.

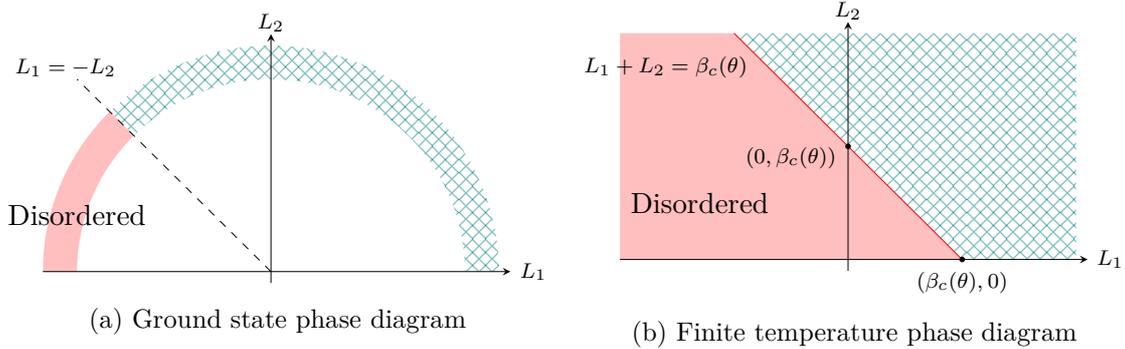
\begin{figure}[h]
    \begin{subfigure}{0.5\textwidth}
    \centering
    \begin{tikzpicture}[scale=1.5]
    
    \path[pattern color=teal, pattern=hatch, hatch size = 5pt, hatch line width = .5pt, opacity = 0.7]
    (2,0) arc [radius=2, start angle=0, delta angle=135]
                  -- (-1.1,1.1) arc [radius=1.55, start angle=135, delta angle=-135]
                  -- cycle;
    
    \path[fill=pink, opacity = 1] (0,0) -- (-2, 0) -- (-2,2) -- (0,0);
    
    \path[fill=white]
    (1.7,0) arc [radius=1.7, start angle=0, delta angle=180] -- (1.7,0) -- cycle;
    \path[fill=white]
    (0,2.1) -- (2.1,2.1) -- (2.1,0) -- (2,0) arc [radius=2, start angle=0, delta angle=180] -- (-2.1,0) -- (-2.1,2.1) -- (0,2.1) -- cycle;
    
    \draw[->, >=stealth, line width=0.3pt](-2, 0) -- (2.1,0);
    \draw[->, >=stealth, line width=0.3pt](0,-0.1) -- (0,2.1);
    \draw[-, >=stealth, dashed, line width=0.3pt](0,0) -- (-1.7,1.7);
    
    \draw[font=\fontsize{8}{10}] (0, 2.2) node {$L_2$};
    \draw[font=\fontsize{8}{10}] (2.3, 0) node {$L_1$};
    \draw[font=\fontsize{8}{10}] (-1.8, 1.8) node {$L_1 = -L_2$};
    
    \draw[font=\fontsize{4}{6}] (-1.7, 0.5) node {Disordered};
    
    \end{tikzpicture}
    \caption{Ground state phase diagram}
    \label{fig:GS-higher}
    \end{subfigure}
    \begin{subfigure}{0.5\textwidth}
    \centering
    \begin{tikzpicture}[scale=1.5]

    \path[pattern color=teal, pattern=hatch, hatch size = 5pt, hatch line width = .5pt, opacity = 0.7] (1,0) -- (-1,2) -- (2,2) -- (2,0) -- (1,0);
    
    \path[fill=pink, opacity = 1](1,0) -- (-1,2) -- (-2,2) -- (-2,0) -- (1,0);

    \draw[->, >=stealth, line width=0.3pt](-2, 0) -- (2.1,0);
    \draw[->, >=stealth, line width=0.3pt](0,-0.1) -- (0,2.1);

    \draw[-, >=stealth, line width=0.3pt, red](1,0) -- (-1,2);
    
    \draw[fill=black] (1,0) circle (0.02cm);
    \draw[font=\fontsize{8}{10}] (1, -0.2) node {$(\beta_c(\theta),0)$};
    \draw[fill=black] (0,1) circle (0.02cm);
    \draw[font=\fontsize{8}{10}] (-0.5,0.9) node {$(0,\beta_c(\theta))$};
    
    \draw[font=\fontsize{8}{10}] (0, 2.2) node {$L_2$};
    \draw[font=\fontsize{8}{10}] (2.3, 0) node {$L_1$};
    \draw[font=\fontsize{8}{10}] (-1.6, 1.7) node {$L_1+L_2 = \beta_c(\theta)$};
    
    \draw[font=\fontsize{4}{6}] (-1.3, 0.5) node {Disordered};
    
    \end{tikzpicture}
    \caption{Finite temperature phase diagram}
    \label{fig:PD-higher}
    \end{subfigure}
    \caption{On the left, the ground state phase diagram for the spin $S$ model with Hamiltonian (\ref{eqn:fullhamiltonian}), in the region $L_2\ge0$. On the right, the phases for $L_2\ge0$ at finite temperature, where varying temperature is given by varying the modulus $||(L_1,L_2)||$. Transitions between phases (points of non-analyticity of the free energy) shown in red lines.}
\end{figure}

\begin{rmk}[A remark on the Interchange process with reversals]
    As noted in the introduction, for certain values of the parameters, the model with Hamiltonian (\ref{eqn:fullhamiltonian}) has a probabilistic representation as an interchange process with reversals, re-weighed by $\theta^{\mathrm{\# \ loops}}$; see \cite{Ueltschi_2013}. To be precise, let $L_1 = 1-L_2 = u \in [0,1)$, and introduce a temperature parameter $\beta$, that is, let $Z_{n,\theta}(u,\beta):=\mathrm{tr}[e^{-\beta H(n,\theta,u,1-u)}]$. Then the corresponding interchange process is that described in Section 2A of \cite{Ueltschi_2013}, with $\beta$ translating to time in the interchange process. It is natural to ask what our results imply, if anything, about this process; we have one remark to make on this topic. In \cite{Bj_rnberg_Milos_Lees_Kotovski_2019}, the authors consider the unweighted process with reversals, and prove that above a critical time $\beta$, the rescaled loop lengths converge to a Poisson-Dirichlet distribution, as $n$, the number of particles, tends to infinity. In particular, the critical time (and indeed the limiting distribution) are independent of the parameter $u$. Our result Theorem \ref{thm:PD-higher} indicates that a similar result might hold for the re-weighed process, since the transition in the spin model occurs at $\beta=\beta_c$ (\ref{eqn:beta_c}), independent of $u$.

    For completeness, we make the following final remark. In \cite{Bj_rnberg_Ueltschi_Frolich_2019}, the authors obtain expressions for total spin observables of the form of Theorem \ref{thm:total-spin}, and note they are equal to certain observables of the corresponding interchange process, which are characteristic functions of the lengths of loops. Then they check that the limits of these observables, evaluated under the Poisson-Dirichlet distribution, are the same as the expressions obtained for the total spin observables. This supports the hypothesis that the rescaled loop lengths in the reweighed process are, in the limit, distributed is distributed according a Poisson Dirichlet distribution. It is tempting to try to play the same game here; however, we are unfortunately not able to with our specific total spin observables. Our total spin observable is trivial in the region where the probabilistic representation holds; we can give more details in the following.
    
    In Theorem 2.3 of \cite{Bj_rnberg_Ueltschi_Frolich_2019}, the authors consider a total spin expression of the form of Theorem \ref{thm:total-spin} for the case $u=1$, the ``Quantum Interchange Model'', and for the matrix $W$ replaced with any $\theta\times\theta$ matrix, with eigenvalues $h_1, \dots, h_\theta$. That model has a probabilistic representation as the Interchange process (without reversals) with configurations re-weighed by $\theta^{\mathrm{\# \ loops}}$, described in Section 3.3 of that paper. A configuration of that process at time $\beta$ is given by a configuration of certain loops; we label the lengths of the loops $l_i$, and the number of loops $l(\sigma)$. The total spin in finite volume is shown to be equal to the expectation of the observable $\prod_{i\ge1} \frac{1}{\theta^{l(\sigma)}}(e^{h_1l_i/n},\dots,e^{h_\theta l_i/n})$. Now using Theorem 4.6 from \cite{ram1995}, one can obtain the same expression for our total spin in Theorem \ref{thm:total-spin}, except the length of a loop, which before was the number of vertices at time $\beta=0$ it visits, is replaced by the modulus of its winding number. The winding number definition comes from the algebraic equivalent in the Brauer algebra of the length of a cycle in the symmetric group, in that it defines conjugacy classes in the Brauer algebra (see Section \ref{section:reptheory} and Theorem 3.1 of \cite{ram1995}). In the case of the interchange process without reversals, (and equivalently in the symmetric group) the length is the same as the modulus of the winding number, so there is no issue; this is not the case when reversals are introduced. Hence this observable does not tend to a function of the rescaled loop lengths, so cannot be compared with the Poisson-Dirichlet distribution. 
\end{rmk}

\section{The Brauer algebra}\label{section:reptheory}
We essentially prove Theorem \ref{thm:main-1.1} by identifying the eigenspaces of the Hamiltonian, their dimensions, and their corresponding eigenvalues. We first observe that the Hamiltonian is actually the action of an element of the Brauer algebra on $\mathcal{H}_V$. Representation theory gives us information on the irreducible invariant subspaces of this action, which leads us to the eigenspaces of the Hamiltonian. This decomposition into irreducibles is part of a classical theory called Schur-Weyl duality. 

In this section we will define the Brauer algebra, and note how its irreducible representations are enumerated, along with those of the symmetric group and the general linear and orthogonal groups. Most of the representation-theoretic results that we use throughout this paper go back to Weyl, but for the sake of accessibility, we mainly follow Ram \cite{ram1995}, but also take some results from other papers.

Let $\theta \in \mathbb{C}$. The Brauer algebra $\mathbb{B}_{n,\theta}$ 
is the (formal) complex span of the set of pairings of $2n$ vertices. We think of pairings as graphs, which we will call diagrams, with each vertex having degree exactly 1. We arrange the vertices in two horizontal rows, labelling the upper row (the northern vertices) $1^+, 2^+, \dots, n^+$, and the lower (southern) $1^-, \dots, n^-$. We call an edge connecting two northern vertices (or two southern) a bar. 

Multiplication of two diagrams is given by concatenation. If $b_1, b_2$ are two diagrams, we align the northern vertices of $b_1$ with the southern of $b_2$, and the result is obtained by removing these middle vertices (which produces a new diagram), and then multiplying the result by $\theta^{l(b_1,b_2)}$, where $l(b_1,b_2)$ is the number of loops in the concatenation. See Figure \ref{fig:4:multiplicationofdiagramsexample}. This defines $\mathbb{B}_{n,\theta}$ as an algebra.

\begin{figure}[h]
    \centering
    \resizebox{0.95\textwidth}{!}{%
    \begin{tikzpicture}[scale=1]

    \draw[-, rounded corners=10pt, dashed, thick](2, 1) -- (2, 1.5);
    \draw[-, rounded corners=10pt, dashed, thick](0, 1) -- (0, 1.5);
    \draw[-, rounded corners=10pt, dashed, thick](1,1) -- (1, 1.5);
    \draw[-, rounded corners=10pt, dashed, thick](3, 1) -- (3, 1.5);
    \draw[-, rounded corners=10pt, dashed, thick](5,1) -- (5, 1.5);
    \draw[-, rounded corners=10pt, dashed, thick](4,1) -- (4, 1.5);

    \draw[font=\large] (-1, 0.5) node {$b_2$};
    \draw[-, rounded corners=10pt, line width=2pt, orange](1, 1) -- (3, 0);
    \draw[-, rounded corners=10pt, line width=2pt, orange](0, 1) -- (0,0);
    \draw[-, rounded corners=10pt, line width=2pt, orange](2, 1) -- (3, 1);
    \draw[-, rounded corners=10pt, line width=2pt, orange](2, 0) -- (1, 0);
    \draw[-, rounded corners=10pt, line width=2pt, orange](4,0) -- (4, 1);
    \draw[-, rounded corners=10pt, line width=2pt, orange](5,0) -- (5, 1);
    
    \draw[thick, fill=white] (0, 1) circle (0.1cm);
    \draw[thick, fill=white] (1, 1) circle (0.1cm);
    \draw[thick, fill=white] (2, 1) circle (0.1cm);
    \draw[thick, fill=white] (3, 1) circle (0.1cm);
    \draw[thick, fill=white] (4, 1) circle (0.1cm);
    \draw[thick, fill=white] (5, 1) circle (0.1cm);
    \draw[thick, fill=white] (0, 0) circle (0.1cm);
    \draw[thick, fill=white] (1, 0) circle (0.1cm);
    \draw[thick, fill=white] (2, 0) circle (0.1cm);
    \draw[thick, fill=white] (3, 0) circle (0.1cm);
    \draw[thick, fill=white] (4, 0) circle (0.1cm);
    \draw[thick, fill=white] (5, 0) circle (0.1cm);

    \draw[font=\large] (-1, 2) node {$b_1$};
    \draw[-, rounded corners=10pt, line width=2pt, orange](2, 2.5) -- (3, 2.5);
    \draw[-, rounded corners=10pt, line width=2pt, orange](0, 2.5) -- (0, 1.5);
    \draw[-, rounded corners=10pt, line width=2pt, orange](1,2.5) -- (1, 1.5);
    \draw[-, rounded corners=10pt, line width=2pt, orange](2, 1.5) -- (3, 1.5);
    \draw[-, rounded corners=10pt, line width=2pt, orange](5,2.5) -- (5, 1.5);
    \draw[-, rounded corners=10pt, line width=2pt, orange](4,2.5) -- (4, 1.5);
    
    \draw[thick, fill=white] (0, 1.5) circle (0.1cm);
    \draw[thick, fill=white] (1, 1.5) circle (0.1cm);
    \draw[thick, fill=white] (2, 1.5) circle (0.1cm);
    \draw[thick, fill=white] (3, 1.5) circle (0.1cm);
    \draw[thick, fill=white] (4, 1.5) circle (0.1cm);
    \draw[thick, fill=white] (5, 1.5) circle (0.1cm);
    \draw[thick, fill=white] (0, 2.5) circle (0.1cm);
    \draw[thick, fill=white] (1, 2.5) circle (0.1cm);
    \draw[thick, fill=white] (2, 2.5) circle (0.1cm);
    \draw[thick, fill=white] (3, 2.5) circle (0.1cm);
    \draw[thick, fill=white] (4, 2.5) circle (0.1cm);
    \draw[thick, fill=white] (5, 2.5) circle (0.1cm);

    \draw[-, rounded corners=10pt, thick](5.65,1.3) -- (5.95, 1.3);
    \draw[-, rounded corners=10pt, thick](5.65,1.2) -- (5.95, 1.2);

    \draw[font=\large] (13.3, 1.25) node {$=b_1 b_2$};
    \draw[font=\large] (6.5, 1.25) node {$\theta^1$};
    \draw[-, rounded corners=10pt, line width=2pt, orange](10, 1.75) -- (9, 1.75);
    \draw[-, rounded corners=10pt, line width=2pt, orange](7, 1.75) -- (7, 0.75);
    \draw[-, rounded corners=10pt, line width=2pt, orange](8, 1.75) -- (10, 0.75);
    \draw[-, rounded corners=10pt, line width=2pt, orange](8, 0.75) -- (9, 0.75);
    \draw[-, rounded corners=10pt, line width=2pt, orange](11,1.75) -- (11, 0.75);
    \draw[-, rounded corners=10pt, line width=2pt, orange](12,0.75) -- (12, 1.75);
    
    \draw[thick, fill=white] (7, 0.75) circle (0.1cm);
    \draw[thick, fill=white] (8, 0.75) circle (0.1cm);
    \draw[thick, fill=white] (9, 0.75) circle (0.1cm);
    \draw[thick, fill=white] (10, 0.75) circle (0.1cm);
    \draw[thick, fill=white] (11, 0.75) circle (0.1cm);
    \draw[thick, fill=white] (12, 0.75) circle (0.1cm);
    \draw[thick, fill=white] (7, 1.75) circle (0.1cm);
    \draw[thick, fill=white] (8, 1.75) circle (0.1cm);
    \draw[thick, fill=white] (9, 1.75) circle (0.1cm);
    \draw[thick, fill=white] (10, 1.75) circle (0.1cm);
    \draw[thick, fill=white] (11, 1.75) circle (0.1cm);
    \draw[thick, fill=white] (12, 1.75) circle (0.1cm);
    
    \end{tikzpicture}
    }
    \caption{Two diagrams $b_1$ and $b_2$ (left), and their product (right). The concatenation contains one loop, so we multiply the concatenation (with middle vertices removed) by $\theta^1$.}
    \label{fig:4:multiplicationofdiagramsexample}
    
\end{figure}
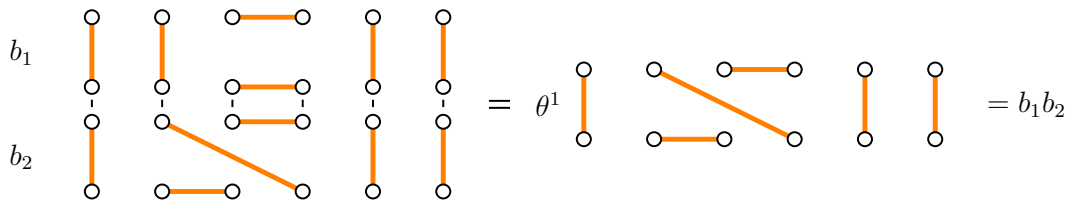

We call the set of diagrams $B_n$. Note that diagrams with no bars are exactly permutations, where $\sigma \in S_n$ is represented by the diagram where $x^-$ is connected to $\sigma(x)^+$, so $S_n\subset B_n$. Moreover the multiplication defined above reduces to multiplication in $S_n$, so $\mathbb{C}S_n$ is a subalgebra of $\mathbb{B}_{n,\theta}$. We write $\mathrm{id}$ for the identity - its diagram has all its edges vertical. We denote the transposition $S_n$ swapping $x$ and $y$ by $(x,y)$, and we write $(\overline{x,y})$ for the diagram with $x^+$ connected to $y^+$, and $x^-$ connected to $y^-$, and all other edges vertical. See Figure \ref{fig:5:idandtranspositions}. Note that just as the transpositions $(x,y)$ generate the symmetric group, the Brauer algebra is generated by the transpositions and the elements $(\overline{x,y})$. 

Let us note that the diagrams and multiplication depicted in Figures \ref{fig:4:multiplicationofdiagramsexample} and \ref{fig:5:idandtranspositions} mirror the paths in the interchange process with reversals (see \cite{aizenman1994}, \cite{Ueltschi_2013} for definitions). In a similar way to the interchange process without reversals being thought of as a continuous time random walk on the symmetric group, this shows that the process with reversals can be thought of as a random walk on the basis $B_n$ of the Brauer algebra. See, for example, Figure 1 from \cite{Ueltschi_2013}. 

\begin{figure}[h]
    \centering
    \resizebox{0.6\textwidth}{!}{%
    \begin{tikzpicture}[scale=1]
    \draw[font=\large] (7.5, 0.5) node {= $(24) \in S_n$};
    \draw[-, rounded corners=10pt, line width=2pt, teal](1, 1) -- (3, 0);
    \draw[-, rounded corners=10pt, line width=2pt, teal](0, 1) -- (0,0);
    \draw[-, rounded corners=10pt, line width=2pt, teal](2,0) --  (1.75, 0.5) -- (2,1);
    \draw[-, rounded corners=10pt, line width=2pt, teal](3, 1) -- (1, 0);
    \draw[-, rounded corners=10pt, line width=2pt, teal](4,0) -- (4, 1);
    \draw[-, rounded corners=10pt, line width=2pt, teal](5,0) -- (5, 1);
    \draw[thick, fill=white] (0, 1) circle (0.1cm);
    \draw[thick, fill=white] (1, 1) circle (0.1cm);
    \draw[thick, fill=white] (2, 1) circle (0.1cm);
    \draw[thick, fill=white] (3, 1) circle (0.1cm);
    \draw[thick, fill=white] (4, 1) circle (0.1cm);
    \draw[thick, fill=white] (5, 1) circle (0.1cm);
    \draw[thick, fill=white] (0, 0) circle (0.1cm);
    \draw[thick, fill=white] (1, 0) circle (0.1cm);
    \draw[thick, fill=white] (2, 0) circle (0.1cm);
    \draw[thick, fill=white] (3, 0) circle (0.1cm);
    \draw[thick, fill=white] (4, 0) circle (0.1cm);
    \draw[thick, fill=white] (5, 0) circle (0.1cm);
    
    \draw[font=\large] (7.5, 2) node {= $(\overline{34})$};
    \draw[-, rounded corners=10pt, line width=2pt, teal](2, 2.5) -- (3, 2.5);
    \draw[-, rounded corners=10pt, line width=2pt, teal](0, 2.5) -- (0, 1.5);
    \draw[-, rounded corners=10pt, line width=2pt, teal](1,2.5) -- (1, 1.5);
    \draw[-, rounded corners=10pt, line width=2pt, teal](2, 1.5) -- (3, 1.5);
    \draw[-, rounded corners=10pt, line width=2pt, teal](5,2.5) -- (5, 1.5);
    \draw[-, rounded corners=10pt, line width=2pt, teal](4,2.5) -- (4, 1.5);
    \draw[thick, fill=white] (0, 1.5) circle (0.1cm);
    \draw[thick, fill=white] (1, 1.5) circle (0.1cm);
    \draw[thick, fill=white] (2, 1.5) circle (0.1cm);
    \draw[thick, fill=white] (3, 1.5) circle (0.1cm);
    \draw[thick, fill=white] (4, 1.5) circle (0.1cm);
    \draw[thick, fill=white] (5, 1.5) circle (0.1cm);
    \draw[thick, fill=white] (0, 2.5) circle (0.1cm);
    \draw[thick, fill=white] (1, 2.5) circle (0.1cm);
    \draw[thick, fill=white] (2, 2.5) circle (0.1cm);
    \draw[thick, fill=white] (3, 2.5) circle (0.1cm);
    \draw[thick, fill=white] (4, 2.5) circle (0.1cm);
    \draw[thick, fill=white] (5, 2.5) circle (0.1cm);
    
    \draw[font=\large] (7.5, 3.5) node {= $\mathrm{id} \in S_6$};
    \draw[-, rounded corners=10pt, line width=2pt, teal](3,3) -- (3, 4);
    \draw[-, rounded corners=10pt, line width=2pt, teal](0,3) -- (0, 4);
    \draw[-, rounded corners=10pt, line width=2pt, teal](1,3) -- (1, 4);
    \draw[-, rounded corners=10pt, line width=2pt, teal](2,3) -- (2, 4);
    \draw[-, rounded corners=10pt, line width=2pt, teal](5,3) -- (5, 4);
    \draw[-, rounded corners=10pt, line width=2pt, teal](4,3) -- (4, 4);
    \draw[thick, fill=white] (0, 3) circle (0.1cm);
    \draw[thick, fill=white] (1, 3) circle (0.1cm);
    \draw[thick, fill=white] (2, 3) circle (0.1cm);
    \draw[thick, fill=white] (3, 3) circle (0.1cm);
    \draw[thick, fill=white] (4, 3) circle (0.1cm);
    \draw[thick, fill=white] (5, 3) circle (0.1cm);
    \draw[thick, fill=white] (0, 4) circle (0.1cm);
    \draw[thick, fill=white] (1, 4) circle (0.1cm);
    \draw[thick, fill=white] (2, 4) circle (0.1cm);
    \draw[thick, fill=white] (3, 4) circle (0.1cm);
    \draw[thick, fill=white] (4, 4) circle (0.1cm);
    \draw[thick, fill=white] (5, 4) circle (0.1cm);
    
    \end{tikzpicture}
    }
    \caption{The identity element, the element $(\overline{34})$, and the transposition $(24) \in S_6$, all lying in $B_6$.}
    \label{fig:5:idandtranspositions}
    
\end{figure}
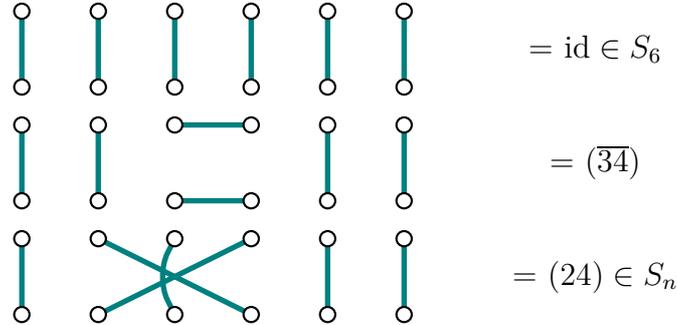

Let us turn to representations. A vector $\rho = (\rho_1, \dots, \rho_t) \in \mathbb{Z}^t$ is a partition of $n$ (we write $\rho \vdash n$) if $\rho_i \ge \rho_{i+1} \ge0$ for all $i$, and $\sum_{i=1}^t \rho_i =n$. Recall that the irreducible representations (and characters) of $\mathbb{C}S_n$ are indexed by partitions of $n$. The Young diagram of $\rho \vdash n$ is the diagram of boxes of $\rho$ with $\rho_j$ boxes in the $j^{th}$ row. When it is unambiguous, will denote the Young diagram of $\rho$ simply by $\rho$. See Figure \ref{fig:youngdiagrams} for an illustration of the Young diagrams of the partitions $(5,5, 3, 1)$, $(4, 1, 1)$ respectively. We label by $c(\rho)$ the sum of contents of the boxes of the Young diagram of $\rho$, where the content of a box in row $i$ and column $j$ is given by $j-i$. For a partition $\rho$, $\rho^T$ is the partition with Young diagram obtained by transposing the diagram of $\rho$ (so $\rho_i^T$ is the length of the $i^{th}$ column of $\rho$). For $\rho \vdash n$ a partition, let $\psi^{S_n}_{\rho}$ be the irreducible representation corresponding to $\rho$, $\chi^{S_n}_{\rho}$ its character, and $d_\rho^{S_n}$ its dimension.  
\begin{figure}[h]
    \centering
    \resizebox{0.5\textwidth}{!}{%
    \begin{tikzpicture}[scale=1]
    \draw[font=\large] (0, 0) node {\yng(5^2,3,1)};
    \draw[font=\large] (4, 0) node {\yng(4,1,1)};
    \end{tikzpicture}
    }
    \caption{The Young diagrams of the partitions $(5,5, 3, 1)$ and $(4, 1, 1)$.}
    \label{fig:youngdiagrams}
\end{figure}
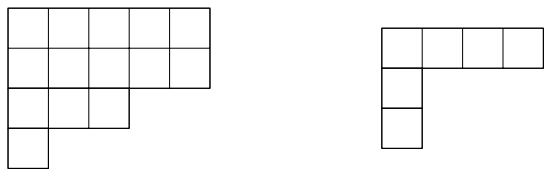
The irreducible representations of the Brauer algebra $\mathbb{B}_{n,\theta}$ are indexed by partitions $\lambda \vdash n-2k$, $0 \le k \le \lfloor \frac{n}{2} \rfloor$ (see, for example, \cite{ram1995} or \cite{cox2007blocks}); let us denote them by $\psi^{\mathbb{B}_{n,\theta}}_{\lambda}$, and their characters by $\chi_{\lambda}^{\mathbb{B}_{n,\theta}}$, and dimensions by $d^{\mathbb{B}_{n,\theta}}_{\lambda}$.

We also recall that irreducible representations of the orthogonal group $O(\theta)$ are given by partitions $\lambda$ of any size such that $\lambda^T_1 + \lambda^T_2 \le \theta$ (see Theorem 1.2 from \cite{ram1995}). The irreducible polynomial representations of $GL(\theta)$ are given by partitions $\rho$ of any size with at most $\theta$ parts. Lastly, we note that the irreducible representations of the special orthogonal group $SO(\theta)$ are indexed by partitions of any size with at most $r=\lfloor\theta/2\rfloor$ parts, with the exception that when $\theta=2r$ even, the $r^{th}$ part can be negative. For each of these three groups $G$, for a given partition $\pi$ we denote the irreducible corresponding to $\pi$, and its character and dimension, by $\psi_\pi^G$, $\chi_\pi^G$ and $d_\pi^G$ respectively.\\

In the following section we prove Theorem \ref{thm:main-1.1}, and in the section after, Theorems \ref{thm:main-higher-spins}, \ref{thm:mag}, \ref{thm:total-spin}, whose proofs are all based on that of \ref{thm:main-1.1}. In Section \ref{section:phasediagrams} we prove Theorems \ref{thm:PD-1/2} and \ref{thm:PD-1}, which follow from Theorem \ref{thm:main-1.1}. In Section \ref{section:restrictiontheorems} we prove the two Propositions \ref{prop:P2} and \ref{prop:P3} which are key technical ingredients for the proof of Theorem \ref{thm:main-1.1}.

\section{Proof of Theorem \ref{thm:main-1.1}}\label{section:proofofmainthm}
In this section we prove our main Theorem \ref{thm:main-1.1}, modulo Propositions \ref{prop:P2} and \ref{prop:P3}, whose proofs are postponed to Section \ref{section:restrictiontheorems}. As noted above, our method is to identify the eigenspaces of our Hamiltonian, their dimensions and associated eigenvalues. We start by viewing the Hamiltonian (\ref{eqn:fullhamiltonian}) as the action of an element of the Brauer algebra $\mathbb{B}_{n,\theta}$ on $\mathcal{H}_V$. 

Let $\theta \ge 2$. The Brauer algebra acts on $\mathcal{H}_V = (\mathbb{C}^{\theta})^{\otimes V}$ by $\mathfrak{p}^{\mathbb{B}_{n,\theta}}(\overline{x,y})=Q_{x,y}$, and $\mathfrak{p}^{\mathbb{B}_{n,\theta}}(x,y)=T_{x,y}$, where recall $T_{x,y}$, $Q_{x,y}$ are given by (\ref{eqn:defnofT,Q}). We therefore have $H = \mathfrak{p}^{\mathbb{B}_{n,\theta}}(\overline{H})$, where 
\begin{equation*}
    \begin{split}
        \overline{H} &= 
        - \sum_{x,y} \left( L_1(x,y) + L_2(\overline{x,y}) \right)\\
        &=  - (L_1+L_2) \sum_{x,y}(x,y) 
        +
        L_2 \sum_{x,y} \left((x,y) - (\overline{x,y}) \right).
    \end{split}
\end{equation*}
Now $\overline{H}$ is a linear combination of two elements in $\mathbb{B}_{n,\theta}$: the sum of all transpositions, which is central in $\mathbb{C}S_n$, and the sum of all transpositions minus all elements $(\overline{x,y})$, which is a central element in $\mathbb{B}_{n,\theta}$. A central element of an algebra acts as a scalar on the irreducible representations of that algebra. Indeed, (see \cite{Doty_2019}) for all $\rho \vdash n$,
\begin{equation}\label{eqn:jucysS_n}
    \psi^{S_n}_{\rho}\left( \sum_{x,y} (x,y) \right)
    =
    c(\rho) \textbf{id},
\end{equation}
and (see Theorem 2.6 from \cite{nazarov}) for all $\lambda\vdash n-2k$, $0 \le k \le \lfloor n/2 \rfloor$,
\begin{equation}\label{eqn:jucysBrauer}
    \psi^{\mathbb{B}_{n,\theta}}_{\lambda} 
    \left(\sum_{x,y}((x,y) - (\overline{x,y})) \right) = (c(\lambda) +k(1-\theta)) \textbf{id}.
\end{equation}
Finding the eigenspaces of the Hamiltonian requires two steps. First we find the irreducible invariant subspaces $\psi^{\mathbb{B}_{n,\theta}}_{\lambda}$ of the action $\mathfrak{p}^{\mathbb{B}_{n,\theta}}$, on each of which the element in (\ref{eqn:jucysBrauer}) acts as a scalar. The element in (\ref{eqn:jucysS_n}) does not act as a scalar on these spaces $\psi^{\mathbb{B}_{n,\theta}}_{\lambda}$. Hence, the second step will be to further decompose these subspaces into smaller spaces (irreducibles $\psi^{S_n}_{\rho}$), on each of which the element in (\ref{eqn:jucysS_n}) does act as a scalar. These smaller spaces are therefore the eigenspaces of the Hamiltonian.

The first step, the decomposition of $\mathfrak{p}^{\mathbb{B}_{n,\theta}}$, is given by a classical theorem called Schur-Weyl duality, which we now describe. The orthogonal group also has a natural action on $\mathcal{H}_V$; for $g \in O(\theta)$, $v_x \in \mathbb{C}^{\theta}_x$ for each $x \in V$, we have $g (v_1 \otimes \cdots \otimes v_n) = g v_1 \otimes \cdots \otimes g v_n$. Schur-Weyl duality states that the actions of the two algebras $\mathbb{B}_{n,\theta}$ and $\mathbb{C}O(\theta)$  on $\mathcal{H}_V$ centralise each other, and $\mathcal{H}_V$ can be viewed as a module of the tensor product $\mathbb{B}_{n,\theta} \otimes \mathbb{C}O(\theta)$, which decomposes as:
\begin{equation}\label{eqn:schur_weyl_duality_brauer}
    \mathcal{H}_V = \bigoplus_{k=0}^{\lfloor\frac{n}{2}\rfloor}
    \bigoplus_{\substack{\lambda \vdash n-2k \\ \lambda_1^T + \lambda_2^T \le \theta}}
    \psi^{O(\theta)}_{\lambda}
    \boxtimes
    \psi^{\mathbb{B}_{n,\theta}}_{\lambda}.
\end{equation}
(Note the square tensor symbol denotes a representation of the tensor product of two algebras, as opposed to the circle tensor which denotes a representation of a single group or algebra). Hence the action $\mathfrak{p}^{\mathbb{B}_{n,\theta}}$ decomposes into irreducibles $\psi^{\mathbb{B}_{n,\theta}}_{\lambda}$ (such that $\lambda_1^T + \lambda_2^T\le\theta$), each with multiplicity $d_{\lambda}^{O(\theta)}$. Note that a similar theorem (the original version of Schur-Weyl duality) holds for the general linear and symmetric groups (see equation (\ref{eqn:SW-Sn}) in Section \ref{section:restrictiontheorems}). Here we only note that those representations of $S_n$ which appear in $\mathcal{H}_V$ are all those with at most $\theta$ parts.

For the second step, we need to restrict $\psi^{\mathbb{B}_{n,\theta}}_{\lambda}$ to the symmetric group and decompose into irreducibles. For $\rho \vdash n$ and $\lambda \vdash n-2k$, $0\le k\le \lfloor n/2\rfloor$, let us define
\begin{equation}\label{eqn:defn_of_b_lamrho}
    \mathrm{res}^{\mathbb{B}_{n,\theta}}_{S_n}
    [\psi_{\lambda}^{\mathbb{B}_{n,\theta}}]
    =
    \bigoplus_{\rho \vdash n }
    (\psi_{\rho}^{S_{n}})^{\oplus b^{n, \theta}_{\lambda, \rho}},
\end{equation}
where $\mathrm{res}$ denotes the restriction of a representation. The coefficients $b^{n, \theta}_{\lambda, \rho}$ are the Brauer algebra - symmetric group branching coefficients. The eigenspaces of the Hamiltonian are therefore indexed by pairs $(\lambda, \rho)$, each appearing with multiplicity $d_{\lambda}^{O(\theta)} b^{n, \theta}_{\lambda, \rho}$; their dimensions are $d_{\rho}^{S_n}$, and their corresponding eigenvalues are $-(L_1+L_2)c(\rho) +L_2[c(\lambda) + k(1-\theta)]$. Taking exponentials and traces, we see that
\begin{equation}\label{eqn:Z-decomposed}
    \begin{split}
        Z_{n,\theta}(L_1,L_2) 
        &=
        \mathrm{tr}\left[
        \mathfrak{p}^{\mathbb{B}_{n,\theta}}(e^{-\frac{1}{n}\overline{H}})
        \right]\\
        &=
        \sum_{k=0}^{\lfloor\frac{n}{2}\rfloor}
        \sum_{\substack{\lambda \vdash n-2k \\ \lambda_1^T + \lambda_2^T \le \theta}} 
        \sum_{\rho \vdash n}
        d_{\lambda}^{O(\theta)}
        b^{n,\theta}_{\lambda, \rho}
        d_{\rho}^{S_n}
        \mathrm{exp}\left[ \frac{1}{n}\left[(L_1 + L_2)c(\rho)  -L_2(c(\lambda)+k(1-\theta)) \right] \right].\\
    \end{split}
\end{equation}
Now we need to take the limit of $\frac{1}{n}\log Z_{n,\theta}(L_1,L_2)$, which will essentially behave like $\frac{1}{n}\log$ of the largest term in the sum above. 
As $n \to \infty$, the behaviour of $c(\lambda)$, $c(\rho)$, and $\frac{1}{n}\log d_{\rho}^{S_n}$ are given by Björnberg \cite{Bj_rnberg_2016}. We will show that $\frac{1}{n}\log d_{\lambda}^{O(\theta)} \to 0$. It remains to analyse the branching coefficients $b^{n,\theta}_{\lambda, \rho}$. In particular, since we are interested in the largest term in the sum above, and the sum is really only over those pairs $(\lambda, \rho)$ for which $b^{n,\theta}_{\lambda, \rho}>0$, it is crucial that we have good knowledge of when these coefficients are non-zero. Obtaining this knowledge is the main technical difficulty of this paper.

Let us introduce some notation. Define $\Lambda_n(\theta)$ to be the set of pairs $(\lambda, \rho)$ of partitions with at most $\theta$ parts, $\lambda \vdash n-2k$ for some $0 \le k \le \lfloor\frac{n}{2} \rfloor$, with $\lambda_1^T + \lambda_2^T \le \theta$, $\rho \vdash n$. Let $P_n(\theta)$ be the set of $(\lambda, \rho) \in \Lambda_n(\theta)$ with the extra condition that $b^{n,\theta}_{\lambda, \rho} >0$. Let $\frac{1}{n}P_n(\theta)$ be the set of pairs $(\frac{\lambda}{n}, \frac{\rho}{n})$, for $(\lambda, \rho)\in P_n(\theta)$.

For $\theta = 2,3$, we give a detailed description of $P_n(\theta)$ in Propositions \ref{prop:P2} and \ref{prop:P3}, proved in Section \ref{section:restrictiontheorems}. Essentially, (ie. apart from a few edge cases which behave well), $(\lambda, \rho) \in \Lambda_n(2)$ lies in $P_n(2)$ iff $0 \le \lambda_1 \le \rho_1 - \rho_2$, and $(\lambda, \rho) \in \Lambda_n(3)$ lies in $P_n(3)$ iff $0 \le \lambda_1 \le \rho_1 - \rho_3$. As noted earlier, we do not know as much detail when $\theta>3$ - we use what we do know to prove Theorem \ref{thm:mag} in Section \ref{section:other-thms}.

We will need to take the limit of the sequence $\frac{1}{n}P_n(\theta)$; let us make clear what we mean by this. Let $\Delta_{\theta} \subset \mathbb{R}^{2\theta}$ be the set of pairs $(x,y) \in ([0,1]^\theta)^2$ such that $\sum_{i=1}^\theta x_i =1, x_i \ge x_{i+1}$ for all $i$, $\sum_{i=1}^{\theta} y_i \in [0,1], y_i \ge y_{i+1}$ for all $i$, and $y_i =0$ for all $i > \lfloor \frac{\theta}{2} \rfloor$. Equip $\mathbb{R}^{2\theta}$ and subsets thereof with $||\cdot||$ the $\infty$-norm, $||z||= \max_{i=1}^{2\theta} |z_i|$, and consider the Hausdorff distance $\mathrm{d}_H(\cdot, \cdot)$ on sets in $\mathbb{R}^{2\theta}$: 
\begin{equation*}
    \mathrm{d}_H(U,W) = \inf \{ \epsilon>0 \ | \ U\subseteq W^\epsilon \ \mathrm{and} \ W \subseteq U^\epsilon \},
\end{equation*}
where $U^\epsilon = \{ x \in \mathbb{R}^{2\theta} \ | \ ||x-u||<\epsilon \ \mathrm{for \ some} \ u \in U \}$. Then Propositions \ref{prop:P2} and \ref{prop:P3} show that $\frac{1}{n}P_n(\theta) \to \Delta^*_\theta$ for $\theta = 2,3$ in this distance, where, recall,
\begin{equation}
    \begin{split}
        \Delta^*_{2} &= \{ (x,y) \in ([0,1]^2)^2 \ | 
                \ x_1 \ge x_2, \ x_1 + x_2 =1, \ y_2 = 0,\ 
                \ 0 \le y_1 \le x_1 - x_2 \},\\
        \Delta^*_{3} &= \{ (x,y) \in ([0,1]^3)^2 \ |
                \ x_1 \ge x_2 \ge x_3, 
                \ x_1 + x_2 + x_3 =1,
                \ y_2, y_3 =0,
                \ 0 \le y_1 \le x_1 - x_3 \}.   
    \end{split}
\end{equation}

The rest of the proof follows very similarly to Section 3 of Björnberg \cite{Bj_rnberg_2016}. As in that paper, we prove a slightly more general convergence result, and then apply it to our setting. 

Let $\Delta$ be any compact subset of $\mathbb{R}^t$, $t\in\mathbb{N}_{>0}$, and let $P_n\subset\Delta$ be a sequence of finite sets with $P_n\to\Delta$ in the Hausdorff distance, and $\frac{1}{n}\log|P_n|\to 0$, as $n\to\infty$. Let $\phi:\Delta\to\mathbb{R}$ continuous, and let $\phi_n:P_n\to\mathbb{R}$ such that $\phi_n\to\phi$ in the sense that there exists $\delta_n\to0$ such that
\begin{equation}\label{eqn:phi_ngoestophi}
    \left| 
    \phi_n(p_n) - \phi(p_n)
    \right|
    \le \delta_n,
\end{equation}
uniformly in $p_n\in P_n$. 
\begin{lem}\label{lem:bjornberglemma1}
    Given the assumptions above, we have that
    \begin{equation*}
        \lim_{n \to \infty}
        \frac{1}{n} \log \left( \sum_{p_n \in P_n} 
        \exp \left[
        n \phi_n(p_n)
        \right] \right)
        =
        \max_{x\in \Delta}  \phi(x).
    \end{equation*}
\end{lem}
\begin{proof}
    Let us first prove an upper bound. We have that
    \begin{equation*}
        \begin{split}
            \frac{1}{n} \log \left( \sum_{p_n \in P_n} 
            \exp [n \phi_n(p_n) ] \right) 
            &\le
            \frac{1}{n} \log \left( |P_n|
            \max_{p_n \in P_n} \{ \exp [n \phi_n(p_n) \} ] \right)\\
            &=
            \max_{p_n \in P_n} [ \phi_n(p_n) ] + o(1)\\
            &\le
            \max_{p_n \in P_n} [ \phi(p_n) ] + \delta_n + o(1)\\
            &\le
            \max_{x \in \Delta} [ \phi(x) ] + \delta_n + o(1),\\
        \end{split}
    \end{equation*}
    where in the second to last inequality we use that $\phi_n$ tends to $\phi$ (\ref{eqn:phi_ngoestophi}), and in the last we use simply that $P_n \subset \Delta$. Hence we have $\limsup_{n \to \infty}
        \frac{1}{n} \log \left( \sum_{p_n \in P_n} 
        \exp \left[
        n \phi_n(p_n)
        \right] \right)
        \le
        \max_{x \in \Delta}  \phi(x)$.
    For the lower bound, we have:
    \begin{equation*}
        \begin{split}
            \frac{1}{n} \log \left( \sum_{(p_n) \in P_n} 
            \exp [n \phi_n(p_n) ] \right) 
            &\ge
            \frac{1}{n} \log \left( \max_{p_n \in P_n} 
            \{ \exp [n \phi_n(p_n) ]  \} \right) \\
            &= \max_{p_n \in P_n} [ \phi_n(p_n) ] \\
            &\ge \max_{p_n \in P_n} [ \phi(p_n) ] + \delta_n.
        \end{split}
    \end{equation*}
    Now it suffices to prove that $\lim_{n \to \infty} \max_{p_n \in P_n} [ \phi(p_n) ] 
        = 
        \max_{x \in \Delta} [ \phi(x) ]$,
    which follows from convergence in the Hausdorff distance. Indeed, since $\Delta$ is compact, the maximum 
    $\max_{x \in \Delta} \phi(x)$
    is attained, say, at $x^*$. Then there exists a sequence of points $p_n \in P_n$ with $p_n \to x^*$. Now 
    $
        \phi \left( p_n \right)
        \le 
        \max_{p_n \in P_n} [ \phi \left( p_n \right) ]
        \le 
        \max_{x \in \Delta} [\phi(x) ]
        =
        \phi(x^*), 
    $
    again in the last inequality using the fact that $P_n \subset \Delta$, which gives the desired limit by continuity of $\phi$. To conclude, $ \liminf_{n \to \infty}
        \frac{1}{n} \log \left( \sum_{p_n \in P_n} 
        \exp \left[
        n \phi_n(p_n)
        \right] \right)
        \ge
        \max_{x\in \Delta}  \phi(x)$,
    which completes the proof.
    \end{proof}


 

Now we set $\Delta=\Delta_\theta^*$, $P_n=\frac{1}{n}P_n(\theta)$, $\phi=\phi_{\theta,L_1,L_2}$ defined below (\ref{eqn:Delta*}), and $\phi_n=\phi_{n,\theta,L_1,L_2}$, where
\begin{equation*}
    \begin{split}
        \phi_{n,\theta, L_1, L_2}(\lambda, \rho)
        =
        \frac{1}{n} \log (d_{\lambda}^{O(\theta)})
        &+ \frac{1}{n} \log (b^{n,\theta}_{\lambda, \rho})
        + \frac{1}{n} \log (d_{\rho}^{S_n}) \\
        &+ \frac{1}{n^2} \left( (L_1 + L_2)c(\rho)  -L_2(c(\lambda)+k(1-\theta)) \right).
    \end{split}
\end{equation*}
Now using Lemma \ref{lem:bjornberglemma1}, in order to prove Theorem \ref{thm:main-1.1}, we note that $\frac{1}{n}P_n(\theta)\to\Delta_\theta^*$ in the Hausdorff distance by Propositions \ref{prop:P2} and \ref{prop:P3}, and it remains to prove $\phi_{n, \theta, L_1, L_2} \to \phi_{\theta, L_1, L_2}$ in the sense of (\ref{eqn:phi_ngoestophi}). Noting that $\frac{1}{n^2}(k(1-\theta)) \to 0$, the final two of the four terms in $\phi_{n,\theta,L_1,L_2}$ give the desired limit; this is proved in Theorem 3.5 of \cite{Bj_rnberg_2016}, the salient points of which are that as $\frac{\rho}{n} \to x$, $\frac{1}{n}\log(d_{\rho}^{S_n}) \to -\sum_{i=1}^\theta x_i \log(x_i)$, and $c(\rho) \to \frac{1}{2}\sum_{i=1}^\theta x_i^2$. So it remains to prove only that $\frac{1}{n} \log (d_{\lambda}^{O(\theta)})
        + \frac{1}{n} \log (b^{n,\theta}_{\lambda, \rho})$
tends to zero as $n \to \infty$, uniformly in $(\lambda,\rho)$. The second of these terms tends to zero by Corollaries \ref{lem:bgoesto0,theta2} and \ref{lem:bgoesto0,theta3} in Section \ref{section:restrictiontheorems}. To show  the first tends to zero, we note that Weyl's formula gives the dimension $d_{\lambda}^{SO(\theta)}$ of the irreducible representation of $SO(\theta)$ corresponding to $\lambda = (\lambda_1, \dots, \lambda_r)$, where $r=\lfloor\theta/2\rfloor$ (see, for example, Section 7 of \cite{goodman-wallach}). For $\theta$ odd, and $\pi_i=n-i+\frac{1}{2}$,

\begin{equation*}
    d_{\lambda}^{SO(\theta)}
    =
    \prod_{1 \le i < j \le r} 
    \frac{(\lambda_i + \pi_i)^2 - (\lambda_j + \pi_j)^2}{\pi_i^2 - \pi_j^2}
    \prod_{1\le i \le r} 
    \frac{\lambda_i + \pi_i}{\pi_i},
\end{equation*}
and for $\pi_i = n-i$, $\theta$ even, we have
\begin{equation*}
    d_{\lambda}^{SO(\theta)}
    =
    \prod_{1 \le i < j \le r} 
    \frac{(\lambda_i + \pi_i)^2 - (\lambda_j + \pi_j)^2}{\pi_i^2 - \pi_j^2}.
\end{equation*}

It's straightforward to see that these dimensions are bounded above by $(2n)^{6r}$. Finally, for $\lambda^T_1\le r$ we note that 
\begin{equation}\label{eqn:res-On-SOn}
        \mathrm{res}^{O(\theta)}_{SO(\theta)}\chi_\lambda^{O(\theta)}=
        \mathrm{res}^{O(\theta)}_{SO(\theta)}\chi_{\lambda'}^{O(\theta)}=
        \chi_\lambda^{SO(\theta)},
\end{equation}
where $\lambda'$ is identical to $\lambda$, except its first column is replaced by $\theta - \lambda_1^T$, except in the case when $\theta=2r$ even, and $\lambda_r>0$, in which case
\begin{equation}\label{eqn:res-On-SOn-even}
        \mathrm{res}^{O(\theta)}_{SO(\theta)}\chi_\lambda^{O(\theta)}=
        \chi_\lambda^{SO(\theta)}+\chi_{\lambda^-}^{SO(\theta)},
\end{equation}
where $\lambda^-$ is the same as $\lambda$ except with $\lambda_r$ replaced with $-\lambda_r$. As a consequence, the dimensions $d_{\lambda}^{O(\theta)}$ are bounded above by $2(2n)^{6r}$. This completes the proof of Theorem \ref{thm:main-1.1}.\\

\section{Proofs of Theorems \ref{thm:main-higher-spins}, \ref{thm:mag} and \ref{thm:total-spin}}\label{section:other-thms}

\begin{proof}[Proof of Theorem \ref{thm:main-higher-spins}]
As noted above, the main technical difficulty in this paper is finding a detailed description for $P_n(\theta)$. For general $\theta$, all of the working from the proof of Theorem \ref{thm:main-1.1} in Section \ref{section:proofofmainthm} holds, apart from the fact that we do not know what the set $P_n(\theta)$ looks like for $\theta >3$. For $L_2 \ge 0$, it turns out that enough information in contained in a theorem of Okada \cite{okada_pieri_type_rules}, which computes the coefficients $b^{n, \theta}_{\lambda, \rho}$ in certain special cases. Note that in \cite{okada_pieri_type_rules}, the coefficients are described in terms of the general linear and orthogonal groups - in Lemma \ref{lem:GltoO_is_BntoSn} we show that this formulation is equivalent to ours.
\begin{rmk}\label{rmk:restrict_one_column_okada}
    Okada's result says: if $\lambda = (1^j)$, $j = 0, \dots, \theta$, then $b^{n,\theta}_{\lambda, \rho} =1$ if $\rho$ has exactly $j$ odd parts, and is zero otherwise (part (2) of Theorem 5.4 of \cite{okada_pieri_type_rules}).
\end{rmk}
 
Now assume $L_2 \ge 0$. Recall the decomposition of $Z_{n,\theta}(L_1,L_2)$ from (\ref{eqn:Z-decomposed}):
\begin{equation}\label{eqn:Z-decomposed-high-spin}
    \begin{split}
        Z_{n,\theta}(L_1,L_2)  
        &=
        \sum_{(\lambda, \rho) \in P_n(\theta)}
        d_{\lambda}^{O(\theta)}
        b^{n,\theta}_{\lambda, \rho}
        d_{\rho}^{S_n}
        \mathrm{exp}\left[ \frac{1}{n}\left[(L_1 + L_2)c(\rho)  -L_2(c(\lambda)+k(1-\theta)) \right] \right].
    \end{split}
\end{equation}
Since the sum behaves like its maximal term, and $L_2\ge0$, it is clear that we would like to minimise $c(\lambda)$. Remark \ref{rmk:restrict_one_column_okada} allows us to do this, since the partitions $(1^j)$ have $c((1^{j}))$ essentially zero. 

Let us make this precise. Given $\rho\vdash n$ with $\rho_1^T\le\theta$, let $j(\rho)$ be the number of parts of $\rho$ of odd length. Then by Remark \ref{rmk:restrict_one_column_okada}, the pair $((1^{j(\rho)}), \rho)$ lies in $P_n(\theta)$. Now take any pair $(\lambda,\rho) \in P_n(\theta)$. It is straightforward to show that $c(\lambda) \ge c((1^{j(\rho)})) - \theta^3$. Indeed, $c((1^{j(\rho)})) \le 0$, and since $\lambda$ has at most $\theta$ parts, it has at most $\theta^2$ boxes with negative content, and those contents must be at least $-\theta$. Substituting into (\ref{eqn:Z-decomposed-high-spin}) gives
\begin{equation}
    \begin{split}
        Z_{n,\theta}&(L_1,L_2) 
        \le\\
        &\sum_{(\lambda, \rho) \in P_n(\theta)}
        d_{\lambda}^{O(\theta)}
        b^{n,\theta}_{\lambda, \rho}
        d_{\rho}^{S_n}
        \mathrm{exp}\left[ \frac{1}{n}\left[(L_1 + L_2)c(\rho)  -L_2(c((1^{j(\rho)})) - \theta^3+k(1-\theta)) \right] \right].
    \end{split}
\end{equation}
The lower bound is trivial, simply take the term $((1^{j(\rho)}), \rho)$ from the sum to achieve
\begin{equation}
    \begin{split}
        Z_{n,\theta}&(L_1,L_2) 
        \ge\\
        &\sum_{\substack{\rho\vdash n, \\ \rho_1^T\le\theta}}
        d_{(1^{j(\rho)})}^{O(\theta)}
        b^{n,\theta}_{(1^{j(\rho)}), \rho}
        d_{\rho}^{S_n}
        \mathrm{exp}\left[ \frac{1}{n}\left[(L_1 + L_2)c(\rho)  -L_2(c((1^{j(\rho)}))+k(1-\theta)) \right] \right].
    \end{split}
\end{equation}
Now we can apply Lemma \ref{lem:bjornberglemma1} to see the result, recalling that $\frac{1}{n}\log|P_n(\theta)|$, $\frac{1}{n}\log d_{\lambda}^{O(\theta)}$, $\frac{1}{n}\log b^{n,\theta}_{\lambda, \rho}$ and $\frac{1}{n^2}(c((1^{j(\rho)})) - \theta^3+k(1-\theta))$ all tend to zero as $n\to\infty$. 
\end{proof}

\begin{proof}[Proof of Theorem \ref{thm:mag}]
    Let $\theta =2,3$, and let $W$ be a skew-symmetric $\theta\times\theta$ matrix with eigenvalues $1$, $-1$ for $\theta=2$, and $1,0,-1$ for $\theta=3$. Consider the model with Hamiltonian $H_h$ given in (\ref{eqn:fullH-mag}), and let $Z_{n,\theta}(L_1,L_2,h) = \mathrm{tr}[e^{-\frac{1}{n}H}]$. The same working as in Section \ref{section:proofofmainthm}, taking traces in (\ref{eqn:schur_weyl_duality_brauer}), gives us:
    \begin{equation}\label{eqn:Z-decomposed-mag}
    \begin{split}
        Z_{n,\theta}&(L_1,L_2,h) 
        =\\
        &\sum_{(\lambda, \rho) \in P_n(\theta)}
        \chi_{\lambda}^{O(\theta)}(e^{hW})
        b^{n,\theta}_{\lambda, \rho}
        d_{\rho}^{S_n}
        \mathrm{exp}\left[ \frac{1}{n}\left[(L_1 + L_2)c(\rho)  -L_2(c(\lambda)+k(1-\theta)) \right] \right].
    \end{split}
    \end{equation}
    Now by Lemma \ref{lem:bjornberglemma1}, to prove the free energy part of the theorem, it suffices to prove that as $\lambda/n \to y$ (as $n\to\infty$), we have
    \begin{equation}\label{eqn:lim-to-|h|y1}
        \frac{1}{n} \log \chi_\lambda^{O(\theta)}(e^{h W}) \to 
        |h|y_1.
    \end{equation}

    We prove a more general lemma, one which holds for all $\theta$. 
    \begin{lem}
        Let $\theta \ge 2$, let $\lambda \vdash n-2k$, $0\le k \le \lfloor n/2\rfloor$ with $\lambda_1^T+\lambda_2^T\le\theta$, and let $\chi_\lambda^{O(\theta)}$ denote the irreducible representation of $O(\theta)$ indexed by $\lambda$. Let $W$ be any $\theta\times\theta$ skew-symmetric matrix with real eigenvalues $w_1\ge\cdots\ge w_\theta$ (note $w_i = -w_{r+1-i}$ for each $i=1,\dots,\theta$). Let $r=\lfloor\theta/2\rfloor$ (note $w_1\ge\cdots\ge w_r\ge0$). Then as $n\to\infty$ and $\lambda/n \to y$, we have
        \begin{equation*}
            \frac{1}{n} \log \chi_\lambda^{O(\theta)}(e^{h W}) \to 
            |h| \sum_{i=1}^r w_iy_i.
        \end{equation*}
    \end{lem}
    \begin{proof}
        Notice that $e^{hW} \in SO(\theta)$. Assume $\lambda_1^T \le \theta/2$, and recall from (\ref{eqn:res-On-SOn}) that in all cases except $\theta$ even and $\lambda_1^T =\theta/2$, we have that 
        $\mathrm{res}^{O(\theta)}_{SO(\theta)}\chi_\lambda^{O(\theta)}=
        \mathrm{res}^{O(\theta)}_{SO(\theta)}\chi_{\lambda'}^{O(\theta)}=
        \chi_\lambda^{SO(\theta)}$,
        where the latter is the irreducible representation of $SO(\theta)$ with highest weight $\lambda$, and where $\lambda'$ is identical to $\lambda$, except its first column is replaced by $\theta - \lambda_1^T$. In this case, a formula due to King (see Theorem 2.5 of \cite{sundaram}) gives 
        \begin{equation}\label{eq:char-tableaux-mag}
            \chi_\lambda^{SO(\theta)}(e^{\beta hw_1}, \dots, e^{\beta hw_r})=
            \sum_{\mathbb{T}}2^{m(\mathbb{T})} e^{\beta h \sum_{i=1}^r w_i(m_i-m_{\overline{i}})},
        \end{equation}
        where the sum is over semistandard Young tableaux of shape $\lambda$ filled with indices $1<\overline{1}<2<\overline{2}<\cdots<r<\overline{r}<\infty$, such that:
        \begin{enumerate}
            \item The entries of row $i$ are all at least $i$,
            \item If $i$ and $\overline{i}$ appear consecutively in a row, then there is an $i$ in the box directly above the $\overline{i}$. 
        \end{enumerate}
        Here $m(\mathbb{T})$ is the number of occurrences of $i$ directly above $\overline{i}$ in the first column of the tableau $\mathbb{T}$, with $\overline{i}$ in row $i$, and $m_i$ is the number of times $i$ appears in the tableau, $m_{\overline{i}}$ similar. We recall also that a semistandard Young tableau of shape $\lambda$ is a Young diagram of shape $\lambda$ with each box filled with one of a set of indices, such that along rows the indices are non-decreasing, and down columns they strictly increase. Let $h>0$. The exponent in (\ref{eq:char-tableaux-mag}) is maximised by the tableau with every box in row $i$ containing $i$. Indeed, taking any tableau $\mathbb{T}$, changing a single box in row $i$ to contain $i$ changes the exponent by either $ h (w_i-w_j)$ (if the box contains $j\ge i$), $ h (w_i+w_j)$ (if the box contains $\overline{j} \ge i$) or $ hw_i$ (if the box contains $\infty$). These are all non-negative, since we ordered $w_1 \ge \cdots \ge w_r \ge 0$. 
        In a very similar way, if $h<0$, the exponent is maximised by the tableau with $\overline{i}$ as each entry of row $i$. In either case, the maximum exponent is $ |h| \sum_{i=1}^s w_i \lambda_i$. Now we have
        \begin{equation}\label{eq:char-compare-mag}
            e^{ |h| \sum_{i=1}^s w_i\lambda_i} \le
            \chi_\lambda^{SO(\theta)}(e^{ hw_1}, \dots, e^{ hw_s})\le 
            e^{ |h| \sum_{i=1}^s w_i\lambda_i} \sum_{\mathbb{T}}2^{m(\mathbb{T})},
        \end{equation}
        and noticing that $2^{m(\mathbb{T})}$ is bounded, and the number of $\mathbb{T}$, which is the dimension of the irreducible representation, satisfies $\frac{1}{n}\log d_{\lambda}^{SO(\theta)} \to 0$, we have the key claim.
    \end{proof}
    This proves (\ref{eqn:lim-to-|h|y1}), and therefore proves the free energy part of the theorem. It remains to prove the second part. Again we prove a more general lemma. Let $r=\lfloor\frac{\theta}{2}\rfloor$, and let $\Delta^\bullet_\theta$ be the set of pairs $(x,y)^2 \in ([0,1]^\theta)^2$, with $x_i \ge x_{i+1}\ge 0$, $\sum_{i=1}^\theta x_i=1$, $y_i \ge y_{i+1}\ge 0$, $y_i=0 \ \mathrm{for} \ i>r$, $\sum_{i=1}^\theta y_i\in[0,1]$.
    \begin{lem}
        Let $h, L_1, L_2$ be real, and let $w_1\ge\cdots\ge w_r\ge0$, where $r=\lfloor\theta/2\rfloor$. Define a function $\Phi$ as 
        \begin{equation*}
        \Phi=\Phi(L_1, L_2, h)
        = 
        \max_{(x,y) \in \Delta^\dagger_{\theta}} (x,y)
        \left[ \phi_{\theta,L_1,L_2} + |h|\sum_{i=1}^rw_iy_i \right],
    \end{equation*}
    where $\Delta^\dagger_{\theta}$ is some compact subset of $\Delta^\bullet_\theta$. Then 
    \begin{equation}
        \frac{\partial \Phi}{\partial h} \Big|_{h\downarrow0}=
        \sum_{i=1}^s w_iy_i^{\uparrow},
        \qquad
        \frac{\partial \Phi}{\partial h} \Big|_{h\uparrow0}=
        \sum_{i=1}^s w_iy_i^{\downarrow},
    \end{equation}
    where $(x^\uparrow, y^\uparrow)$ is the maximiser of $\phi$ in $\Delta^\dagger_{\theta}$ which maximises the inner product $\sum_{i=1}^s w_iy_i$, and $(x^\downarrow, y^\downarrow)$ the one which minimises the inner product.
    \end{lem}
    \begin{proof}
        The proof follows the proof of Theorem 4.1 from \cite{Bj_rnberg_2016} very closely. We prove the case of the right derivative - the left derivative is almost identical. Note that $(x^\uparrow, y^\uparrow)$ (resp. $(x^\downarrow, y^\downarrow)$) may not be unique, but this does not matter for the proof; from hereon in by $(x^\uparrow, y^\uparrow)$ we mean one such maximiser. We have
	    \begin{equation}\label{eq:left-deriv-mag}
	        \frac{\Phi(L_1, L_2,h) - \Phi(L_1, L_2,0)}{h}=
	        \max_{(x,y)\in\Delta_\theta^\dagger} \left[ \sum_{i=1}^sy_iw_i+\frac{\phi(x,y) - \phi(x^\uparrow, y^\uparrow)}{h} \right].
	    \end{equation}
	    We denote the function being maximised on the right hand side by $f(x,y;h)$. Clearly its maximum is bounded below by $\sum_{i=1}^sy^{\uparrow}_iw_i$. For fixed $h>0$, let $(x(h),y(h))$ maximise $f(x,y;h)$ (such a maximiser exists as $x,y$ lie in compact sets, and $f$ is continuous). It suffices to show that as $h\searrow0$, $(x(h),y(h)) \to (x^\uparrow,y^\uparrow)$. Certainly $(x(h),y(h))$ must tend to a maximiser of $\phi(x,y)$; if it did not, then by continuity, $\phi(x,y) - \phi(x^\uparrow, y^\uparrow)$ would stay bounded away from zero (below some negative number), and the right hand side of (\ref{eq:left-deriv-mag}) would tend to $-\infty$. This contradicts the lower bound we noted above. To conclude, $(x(h),y(h))$ must tend to $(x^\uparrow,y^\uparrow)$ (and not a different maximiser), since the sum $\sum_{i=1}^r=y_iw_i$ defining $y^{\uparrow}$ appears in $f(x,y;h)$.
    \end{proof}
    This concludes the proof of Theorem \ref{thm:mag}.
\end{proof}

\begin{proof}[Proof of Theorem \ref{thm:total-spin}]
    We have, taking traces in (\ref{eqn:schur_weyl_duality_brauer}),
    \begin{equation}
        \begin{split}
            \langle e^{(h/n)W} \rangle_{n,\theta}
            &=
            \frac{\sum_\lambda
            \chi_{\lambda}^{O(\theta)}(e^{(h/n)W})
            \chi_\lambda^{\mathbb{B}_{n,\theta}}(e^{(h/n)W})}{\sum_\lambda
            d_{\lambda}^{O(\theta)}
            \chi_\lambda^{\mathbb{B}_{n,\theta}}(e^{(h/n)W})}\\
            &=
            \frac{\sum_\lambda
            d_{\lambda}^{O(\theta)}
            \chi_\lambda^{\mathbb{B}_{n,\theta}}(e^{(h/n)W})
            \frac{\chi_{\lambda}^{O(\theta)}(e^{(h/n)W})}{d_{\lambda}^{O(\theta)}}}{\sum_\lambda
            d_{\lambda}^{O(\theta)}
            \chi_\lambda^{\mathbb{B}_{n,\theta}}(e^{(h/n)W})}.
        \end{split}
    \end{equation}
    Using Lemma B.1 from \cite{Bj_rnberg_Ueltschi_Frolich_2019}, it suffices to show that for $\theta=3$, $ {\chi_{\lambda}^{O(\theta)}(e^{(h/n)W})}/{d_{\lambda}^{O(\theta)}} \to {sinh(hy_1)}/{hy_1}$ as $\lambda/n \to y$, while for $\theta=2$, the limit is $cosh(hy_1)$. We note also that Lemma B.1 from \cite{Bj_rnberg_Ueltschi_Frolich_2019}, holds for complex-valued functions $F$ in the notation of that lemma, so in particular holds for complex $h$ in our case. Let $\theta=3$. Using the determinental formula for the character of the orthogonal group \cite{ram1995}, and the Weyl dimension formula \cite{goodman-wallach}, we have
    \begin{equation*}
        \frac{\chi_{\lambda}^{O(3)}(e^{(h/n)W})}{d_{\lambda}^{O(3)}}=
        \frac{e^{(h/n)(\lambda_1 +1/2)}-e^{-(h/n)(\lambda_1 +1/2)}}{e^{h/2n}-e^{-h/2n}}
        \cdot \frac{1/2}{\lambda_1+1/2},
    \end{equation*}
    which, on expanding the exponentials in the denominator, clearly tends to the desired limit. The $\theta=2$ case is simpler. The dimension $d_{\lambda}^{O(2)}=2$ for all $\lambda$ except $\lambda=\varnothing$ or $(1^2)$ (the trivial and determinant representations), which are both one-dimensional. In the latter two cases, ${\chi_{\lambda}^{O(2)}(e^{(h/n)W})}/{d_{\lambda}^{O(2)}} = 1$, since $e^{(h/n)W}\in SO(2)$, and in the former case, ${\chi_{\lambda}^{O(2)}(e^{(h/n)W})}/{d_{\lambda}^{O(2)}}={(e^{h\lambda_1/n}+e^{-h\lambda_1/n})}/{2}$, which has the desired limit.
\end{proof}

\section{Phase diagrams}\label{section:phasediagrams}
In this section we prove Theorems \ref{thm:PD-1/2}, \ref{thm:PD-1} and \ref{thm:PD-higher} and justify their descriptions of the phase diagrams for their respective systems. We begin with proving Theorem \ref{thm:PD-higher}, and then show that Theorems \ref{thm:PD-1/2} and \ref{thm:PD-1} can be reduced to \ref{thm:PD-higher}.

\subsection{Higher spins; Proof of Theorem \ref{thm:PD-higher}}

\begin{proof}[Proof of Theorem \ref{thm:PD-higher}]
Recall the result Theorem \ref{thm:main-higher-spins}, which gives the free energy of the model with Hamiltonian (\ref{eqn:fullhamiltonian}) in the region $L_2\ge0$:
\begin{equation*}
    \lim_{n \to \infty} \frac{1}{n} \log Z_{n,\theta}(L_1, L_2) 
    = 
    \max_{x \in \Delta_\theta}\left[
    \frac{L_1 + L_2}{2}
    \sum_{i=1}^\theta x_i^2 
    - \sum_{i=1}^\theta x_i \log(x_i)\right],
\end{equation*}
where $\Delta_{\theta} = \{ x \in [0,1]^\theta \ | 
                \ x_i \ge x_{i+1}\ge 0, \ \sum_{i=1}^\theta x_i=1 \}$. 
Let us label the function being maximised $\phi^{\mathrm{int}}$. This function $\phi^{\mathrm{int}}$ is that from Theorem 1.1 of \cite{Bj_rnberg_2016}. In that paper and Lemma C.1 of \cite{Bj_rnberg_Ueltschi_Frolich_2019}, it is proved that the maximisers of $\phi^{\mathrm{int}}$ are always of the form $(x,\frac{1-x}{\theta-1},\dots,\frac{1-x}{\theta-1})$, that for $L_1+L_2\neq\beta_c$ the maximiser is unique, and that at $L_1+L_2=\beta_c$, there are exactly two maximisers, at $x=\frac{1}{\theta}$ and $x=1-\frac{1}{\theta}$ (which become a single unique maximiser when $\theta=2$). Here $\beta_c$ is given by (\ref{eqn:beta_c}). Moreover, it suffices to work with the modified function $\phi^{\mathrm{int}}(x_1,\dots,x_\theta)-\phi^{\mathrm{int}}(\frac{1}{\theta},\dots,\frac{1}{\theta}) = \phi^{\mathrm{int}}(x_1,\dots,x_\theta) -\frac{L_1+L_2}{2\theta}-\log\theta$, since we are subtracting a smooth function of $L_1+L_2$ independent of the variables $x_i$. Combining the above facts, we can consider a function of one variable; let $z=x_1-x_\theta$, $\beta=L_1+L_2$, and $\phi:[0,1]\to\mathbb{R}$ as
\begin{equation*}
    \phi(z)=\phi_{\theta,\beta}(z)=\frac{\beta(\theta-1)}{2\theta}z^2 - \frac{(\theta-1)z+1}{\theta}\log((\theta-1)z+1) +\frac{(\theta-1)(1-z)}{\theta}\log(1-z),
\end{equation*}
and let $\Phi(\beta)=\Phi_\theta(\beta)=\max_{z\in[0,1]}\phi(z)$. Note that in this parameterisation, the two maximisers of $\phi$ at $\beta_c$ are $z=0$ and $z=\frac{\theta-2}{\theta-1}$. To prove Theorem \ref{thm:PD-higher}, it suffices to prove that $\Phi_{\theta}(\beta)$ is smooth for $\beta\neq\beta_c$, and is differentiable but not twice-differentiable at $\beta_c$ for $\theta\ge3$, and is continuous but not differentiable at $\beta_c$ for $\theta=2$.

By \cite{Bj_rnberg_2016}, for all $\theta$, $\Phi(\beta)=0$ for all $\beta\le\beta_c$. Let us denote by $z^*=z^*(\beta)$ the unique maximiser of $\phi$ for all $\beta>\beta_c$. For $\beta>\beta_c$, we can obtain a formula for $\beta$ in terms of the maximiser $z^*$, indeed, setting $\frac{\partial \phi}{\partial z}=0$ gives
\begin{equation}\label{eqn:beta-of_z^*}
    \beta=\frac{1}{z^*}\log(\frac{z^*(\theta-1)+1}{1-z^*}).
\end{equation}
This function is smooth and increasing for $z^*\in(\frac{\theta-2}{\theta-1},1)$, tends to $\beta_c$ as $x^*$ tends to $\frac{\theta-2}{\theta-1}$ and to $+\infty$ as $z^*$ tends to $1$. By the inverse function theorem, $z^*$ is a smooth function of $\beta$ in the region $(\beta_c,\infty)$. Hence $\Phi_\theta(\beta)=\phi(z^*(\beta))$ is a smooth function on the interval $(\beta_c,\infty)$. We now turn to the behaviour at $\beta_c$.

First let us show the useful identity: for $\beta>\beta_c$, $ \Phi'(\beta)=\frac{\partial \phi}{\partial \beta}\Big|_{z^*(\beta)}$. 
Indeed, treating $\phi$ as a function of both $\beta$ and $z$, and using that $z^*(\beta)$ is unique for $\beta>\beta_c$, the chain rule gives:
\begin{equation*}
    \Phi'(\beta)=\frac{\partial}{\partial \beta}\phi(z^*(\beta),\beta)
        =\frac{\partial z^*(\beta)}{\partial \beta}\frac{\partial\phi}{\partial z}(z^*(\beta),\beta) + \frac{\partial\phi}{\partial \beta}(z^*(\beta),\beta),
\end{equation*}
and $\frac{\partial\phi}{\partial z}(z^*(\beta),\beta)=0$ gives the claim. Hence 
\begin{equation}\label{eqn:Phi'_fn_of_z^*}
    \Phi'(\beta)=\frac{(\theta-1)}{2\theta}(z^*(\beta))^2. 
\end{equation}
Now the limit as $\beta\searrow\beta_c$ of $\Phi'(\beta)$ is $\frac{(\theta-2)^2}{2\theta(\theta-1)}$, which is strictly larger than the limit as $\beta\nearrow\beta_c$, which is zero, for all $\theta>2$. Hence the free energy is not differentiable at $\beta_c$ for all $\theta>2$. 

For $\theta=2$, we can find the second derivative. Equations (\ref{eqn:beta-of_z^*}) and (\ref{eqn:Phi'_fn_of_z^*}) combine to give 
\begin{equation*}
    \beta=\frac{\artanh(\Phi'(\beta)^{\frac{1}{2}})}{\Phi'(\beta)^{\frac{1}{2}}}.
\end{equation*}
Taking the Taylor expansion of this function shows that $\beta=2+\frac{3}{8}\Phi'(\beta)+O(\Phi'(\beta)^2)$, so the linear approximation of $\Phi'(\beta)$ at $\beta_c=2$ is $\Phi'(\beta)=\frac{8}{3}(\beta-2)$. So the limit as $\beta\searrow\beta_c$ of $\Phi''(\beta)$ is $\frac{8}{3}$, strictly larger than the limit from the left, which is zero. So indeed $\Phi(\beta)$ is not twice-differentiable at $\beta=\beta_c=2$. This concludes the proof of Theorem \ref{thm:PD-higher}.
\end{proof}

In terms of the actual free energy of the model with Hamiltonian (\ref{eqn:fullhamiltonian}), the above proof gives the following quantities, which we include for completeness. Let $\Phi(\beta L_1,\beta L_2)$ be the free energy (\ref{eqn:FE-higher-spins}) of the model with Hamiltonian (\ref{eqn:fullhamiltonian}) in the region $\beta L_2\ge0$. Here we reparameterise by replacing $(L_1,L_2)$ with $(\beta L_1,\beta L_2)$, $L_1,L_2$ fixed, so that we can vary temperature $\beta$ in the usual way (see the remark after (\ref{eqn:Z}). We have
\begin{equation}\label{eqn:Phi_1st_deriv}
    \begin{split}
        \frac{\partial \Phi}{\partial\beta}\Big|_{\beta\nearrow\beta_c} = \cos(\alpha)\frac{1}{2\theta}, \hspace{1cm}
        \frac{\partial \Phi}{\partial\beta}\Big|_{\beta\searrow\beta_c} = \cos(\alpha)\frac{1}{2\theta}\left(1+\frac{(\theta-2)^2}{\theta-1}\right),
    \end{split}
\end{equation}
and in the case $\theta=2$,
\begin{equation}\label{eqn:Phi_2nd_deriv}
    \begin{split}
        \frac{\partial^2 \Phi}{\partial\beta^2}\Big|_{\beta\nearrow\beta_c} = 0, \hspace{1cm}
        \frac{\partial^2 \Phi}{\partial\beta^2}\Big|_{\beta\searrow\beta_c} = \cos^2(\alpha)\frac{8}{3},
    \end{split}
\end{equation}
where again $\beta_c$ is given by (\ref{eqn:beta_c}), and $\alpha$ is the angle made between the fixed vector $(L_1,L_2)$ and the $45^\circ$ line $L_1=L_2$. One can interpret the first derivative (\ref{eqn:Phi_1st_deriv}) as the average energy of the system, and the second derivative (\ref{eqn:Phi_2nd_deriv}) as the variance of the energy. Indeed, 
\begin{equation*}
    \begin{split}
        \frac{\partial}{\partial\beta}\frac{1}{n}\log\mathrm{tr}[e^{-\frac{\beta}{n}H(n,\theta,L_1,L_2)}]
        &= \left\langle-\frac{1}{n^2}H\right\rangle,\\
        \frac{\partial^2}{\partial\beta^2}\frac{1}{n}\log\mathrm{tr}[e^{-\frac{\beta}{n}H(n,\theta,L_1,L_2)}]
        &= \frac{1}{n}\left(\left\langle\left(\frac{1}{n}H\right)^2\right\rangle-\left\langle\frac{1}{n}H\right\rangle^2\right),
    \end{split}
\end{equation*}
where $\langle A\rangle=\frac{\mathrm{tr}[Ae^{-(\beta/n)H}]}{\mathrm{tr}[e^{-(\beta/n)H}]}$ for matrices $A$. One can show that the first derivative above converges exactly to the first derivative in (\ref{eqn:Phi_1st_deriv}) as $n\to\infty$. The second does not converge, and we include only as a heuristic.\\

Before finishing this subsection, let us prove that for $L_1+L_1>0$, $L_2>0$, the finite (even) volume ground state is the vector given in (\ref{eqn:dimer-states-higher}). In order to give explicit finite volume ground states, we will need a concrete realisation of the eigenspaces of the Hamiltonian. Recall that by our working in Section \ref{section:proofofmainthm}, the eigenspaces of the Hamiltonian (\ref{eqn:fullhamiltonian}) (for any $\theta\ge2$) are indexed by pairs $(\lambda, \rho)$, partitions of $n-2k$ ($0\le k\le\lfloor\frac{n}{2}\rfloor$) and $n$ respectively, with  $\lambda^T_1+\lambda^T_2\le\theta$, $\rho^T_1\le\theta$, and $(\lambda,\rho)\in P_n(\theta)$. Following Theorem 2.33 of \cite{benkart-britten-lemire}, the eigenspace itself can be realised as the union of sets
\begin{equation}\label{eqn:Espaces-general}
    z_\rho z_{\lambda} \cdot \left[\reallywidetilde{\prod_{i=1}^k Q_{m_i,m'_i} \cdot \mathcal{H}_V} \right]_0,
\end{equation}
where $\underline{m}, \underline{m}'$ is a pairing of $2k$ vertices in $V$, $Q_{m_i,m'_i}$ from (\ref{eqn:fullhamiltonian}), $\left[\reallywidetilde{\prod_{i=1}^k Q_{m_i,m'_i} \cdot \mathcal{H}_V} \right]_0$ is the set of vectors in $\prod_{i=1}^k Q_{m_i,m'_i} \cdot \mathcal{H}_V$ which are killed by any $Q_{x,y}$ with $x,y\in V\setminus(\underline{m}\cup\underline{m}')$, $z_{\lambda}\in \mathbb{C} S_{|V\setminus(\underline{m}\cup\underline{m}')|}$ is a Young symmetriser for the partition $\lambda$ acting on $\bigotimes_{x\in V\setminus(\underline{m}\cup\underline{m}')} \mathcal{H}_x$, and $z_\rho \in \mathbb{C} S_n$ is a Young symmetriser for the partition $\rho$ acting on all of $\mathcal{H}_V$. While the formula (\ref{eqn:Espaces-general}) is complicated for general pairs $(\lambda,\rho)$, we will see that for some explicit pairs it simplifies greatly. By our working in Section \ref{section:proofofmainthm}, the dimension of the eigenspace is $d_{\lambda}^{O(\theta)} b_{\lambda,\rho}^{n,\theta} d_{\rho}^{S_n}$.

Let $\theta\ge 2$, $L_1+L_1>0$, $L_2>0$. By our working in Section \ref{section:proofofmainthm}, the eigenvalues of the Hamiltonian are 
\begin{equation}
    -[(L_1+L_2)c(\rho)-L_2(c(\lambda)+k(1-\theta))],
\end{equation}
indexed by pairs $(\lambda, \rho)\in P_n(\theta)$, where $\lambda\vdash n-2k$. While (as noted in Section \ref{section:other-thms}) we do not know the structure of $P_n(\theta)$ for $\theta>3$, calculations yield that for $n$ even, the eigenvalue is minimised in $\Lambda_n(\theta)$ at the pair $(\varnothing,(n))$, and by Remark \ref{rmk:restrict_one_column_okada}, $(\varnothing,(n))\in P_n(\theta)$. Now the dimension of the associated eigenspace is $d_{\varnothing}^{O(\theta)} b_{\varnothing,(n)}^{n,2}d_{(\theta)}^{S_n}=1$, and using (\ref{eqn:Espaces-general}) it is straightforward to check it is spanned by (\ref{eqn:dimer-states-higher}).

\subsection{Spin $\frac{1}{2}$, $\theta=2$}
Let us now prove Theorem \ref{thm:PD-1/2}, and justify the description of the phase diagram of the spin $\frac{1}{2}$ Heisenberg XXZ model illustrated in Figures \ref{fig:GS-1/2} and \ref{fig:PD-1/2}. Let $S=\frac{1}{2}$, so $\theta=2$. Recall the Hamiltonian of the Heisenberg XXZ model is given by
\begin{equation}
    \begin{split}
        H' = - \left( 
        \sum_{x,y} K_1 S_x^1 S_y^1 + K_2 S_x^2 S_y^2 + K_1 S_x^3 S_y^3 
        \right).
    \end{split}
\end{equation}
We use the two identities $4(S_x^1 S_y^1 + S_x^2 S_y^2 + S_x^3 S_y^3) + \mathrm{id} = 2T_{x,y}$
        and $4(S_x^1 S_y^1 - S_x^2 S_y^2 + S_x^3 S_y^3) + \mathrm{id} 
        = 2Q_{x,y}$ (see, for example, Section 7 of \cite{Ueltschi_2013}), which show that the Hamiltonian $H'$ is, up to addition of a constant,
        \begin{equation}\label{eqn:spin1/2hamiltoniantransformed}
           H = H(n, K_1, K_2) = - \frac{1}{4} \left(
        \sum_{x,y} (K_1 + K_2)T_{x,y} + (K_1 - K_2)Q_{x,y}
        \right).
        \end{equation}
        Note that the line $K_1 = K_2 >0$ gives the spin $\frac{1}{2}$ Heisenberg ferromagnet, and the line $K_1 = K_2 <0$ gives the antiferromagnet. Let $Z_n(K_1, K_2) = \mathrm{tr}(e^{-\frac{1}{n} H})$, where $H$ is from (\ref{eqn:spin1/2hamiltoniantransformed}).  Setting $L_1 = \frac{1}{4}(K_1 + K_2)$, $L_2 = \frac{1}{4}(K_1 - K_2)$ from Theorem \ref{thm:main-1.1} we have that the free energy of the system with Hamiltonian $H$ given by (\ref{eqn:spin1/2hamiltoniantransformed}) is 
        \begin{equation*}
        \lim_{n \to \infty} \frac{1}{n} \log Z_n(K_1, K_2) 
        = 
        \max_{(x,y) \in \Delta^*_{2} } \phi_{2, K_1, K_2}(x,y),
        \end{equation*}
        where we have 
        \begin{equation*}
            \phi_{2, K_1, K_2}((x_1, x_2), (y_1, 0))
            =
            \frac{1}{8} \left(
            2K_1(x_1^2 + x_2^2) + (K_2-K_1)y_1^2 \right) 
            - \sum_{i=1}^2 x_i \log(x_i),
        \end{equation*}
        and $\Delta^*_{2} = \{ (x,y) \in ([0,1]^2)^2 \ | 
                \ x_1 \ge x_2, \ x_1 + x_2 =1, \ y_2 = 0,\ 
                \ 0 \le y_1 \le x_1 - x_2 \}$.
\begin{proof}[Proof of Theorem \ref{thm:PD-1/2}]
We analyse this free energy by considering different regions of the $(K_1,K_2)$ plane. If $K_1\ge K_2$ we can set $y_1=0$. This is the region covered by Theorem \ref{thm:main-higher-spins}, and the free energy is exactly that of Theorem 1.1 from \cite{Bj_rnberg_2016}, with $\beta$ from that paper replaced with $\frac{K_1}{2}$. The result of Theorem \ref{thm:PD-higher} shows that in this region, the free energy is smooth apart from at the line $K_1=4$, where it is differentiable, but not twice-differentiable. 

Note that if we insert the condition $x_2 = 1-x_1$ into $\phi_{2, K_1, K_2}(x_1)$ in this region $K_1\ge K_2$ (with $y_1=0$), we can rewrite it as $\phi_{2,K_1, K_2}(x_1) = K_1(2x_1-1)^2 - x_1\log(x_1) - (1-x_1)\log(1-x_1)$. Now consider the region $K_2 \ge K_1$. We have to set $y_1 = x_1 - x_2$ in order to maximise $\phi$. Rearranging, and inserting $x_2 = 1-x_1$ now gives almost the same function as above, but with $K_1$ replaced with $K_2$:
\begin{equation}\label{eqn:phi-K2>K1}
        \phi_{2, K_1, K_2}(x_1)
        =
        K_2(2x_1-1)^2 + K_1 - x_1\log(x_1) - (1-x_1)\log(1-x_1).
\end{equation}
The extra term $K_1$ does not affect the location of the maximiser. So, in the region $K_2 \ge K_1$, the free energy is (up to the addition of $K_1$) that of Theorem \ref{thm:main-higher-spins} and Theorem 1.1 from \cite{Bj_rnberg_2016}. So by our proof of Theorem \ref{thm:PD-higher}, it is smooth everywhere in the region $K_2\ge K_1$ apart from the line $K_2=4$, where it is differentiable but not twice-differentiable. 

It remains to join the two regions $K_1\ge K_2$ and $K_2\ge K_1$ together. Clearly the free energy is continuous on the whole plane. The above working shows that in the region $K_1\le4$, $K_2\le4$, the maximiser of $\phi_{2,K_1, K_2}$ is at $((\frac{1}{2},\frac{1}{2}),(0,0))$, so the free energy is smooth in this region. To conclude, let us consider the free energy on a line $K_1=C\ge4$ as it crosses the half-line $K_1=K_2\ge4$. For $K_2\le4$ on this line it is constant by our working above. If we denote the free energy by $\Phi(K_1,K_2)$, then using (\ref{eqn:phi-K2>K1}), for $K_2>4$, 
\begin{equation*}
    \frac{\partial \Phi}{\partial K_2}=\frac{\partial (\phi_{2,K_1, K_2}(x^*))}{\partial K_2}\ge \frac{\partial \phi_{2,K_1, K_2}}{\partial K_2}\Big|_{x^*(K_2)}=(2x_1^*(K_2)-1)^2>0,
\end{equation*}
the last inequality coming from our working in the proof of Theorem \ref{thm:PD-higher}. Hence the free energy is not differentiable on the half-line $K_1=K_2\ge4$, which completes the proof of Theorem \ref{thm:PD-1/2}.
\end{proof}  

\begin{proof}[Proof of Proposition \ref{prop:proofy-bits-1/2}]
Let us now comment on the phase diagram that Theorem \ref{thm:PD-1/2} indicates, and in the process prove Proposition \ref{prop:proofy-bits-1/2}. We label the region $K_1\le4$, $K_2\le4$ (the region where $(x^*,y^*)=((\frac{1}{2},\frac{1}{2}),(0,0))$ maximises $\phi_{2,K_1, K_2}$) the disordered phase, since it maximises the entropy term (the logarithms) in $\phi_{2,K_1, K_2}$. It is illustrated as the solid pink region in Figure \ref{fig:PD-1/2}. The maximiser $y_1^*=(0,0)$ gives the magnetisation of Theorem \ref{thm:mag} $y_1^\uparrow=0$.

We label the region $K_2> K_1$, $K_2>4$ the Ising phase, illustrated as the dotted yellow region in Figure \ref{fig:PD-1/2}. Proposition \ref{prop:P3} and our working to prove Theorem \ref{thm:PD-1/2} show that the maximiser of $\phi_{2,K_1,K_2}$ is unique in the Ising phase, and of the form $(x^*,y^*)=((x^*_1,x^*_2),(x^*_1-x^*_2,0))$, with $x^*_1>x^*_2$. Then the magnetisation $y_1^\uparrow$ of Theorem \ref{thm:mag} is strictly positive. 

As $||(K_1,K_2)||\to\infty$, the maximiser of $\phi_{2,K_1,K_2}$ tends to $((1,0),(1,0))$. Recall that from our working in Section \ref{section:proofofmainthm} and Proposition \ref{prop:P3}, the eigenvalues of the Hamiltonian (\ref{eqn:spin1/2hamiltoniantransformed}) are given by
\begin{equation}\label{eqn:Evals-1/2}
    -\left[2K_1c(\rho)  -(K_1-K_2)(c(\lambda)+k(1-\theta)) \right],
\end{equation}
where $(\lambda,\rho)$ are partitions of $n-2k$ ($0\le k\le\lfloor\frac{n}{2}\rfloor$) and $n$ respectively, with  $\lambda^T_1+\lambda^T_2\le2$, $\rho^T_1\le2$, and $\lambda_1\le\rho_1-\rho_2$. It is not hard to see that for $K_2> K_1$, $K_2>0$, the finite volume ground states are the eigenspace corresponding to the pair $(\lambda,\rho)=((n),(n))$. Using (\ref{eqn:Espaces-general}), this is the space of vectors invariant under the action of $S_n$, and killed by any $Q_{x,y}$, and has dimension $d_{(n)}^{O(2)} b_{(n),(n)}^{n,2}d_{(n)}^{S_n}=2$. A dimension count shows that it is therefore spanned by the two product states $\bigotimes_{x \in V}(|\frac{1}{2}\rangle\pm i|-\frac{1}{2}\rangle)$. Further, if we consider the Hamiltonian with a magnetisation term $-h\sum_{x\in V}S_x^2$ added, since these product states are eigenvectors of the magnetisation term, for $h$ small and positive $\bigotimes_{x \in V}(|\frac{1}{2}\rangle + i|-\frac{1}{2}\rangle)$ must be the unique ground state, and vice-versa for $h$ small and negative. 

We label the region $K_1>K_2$, $K_1>4$ the $XY$ region, illustrated as the hatched blue region in Figure \ref{fig:PD-1/2}. By our working in the proof of Theorem \ref{thm:PD-1/2}, the maximiser of $\phi_{2,K_1, K_2}$ is unique and of the form $((x^*_1,x^*_2),(0,0))$, so the magnetisation $y_1^\uparrow$ from Theorem \ref{thm:mag} is zero. As $||(K_1,K_2)||\to\infty$, the maximiser of $\phi_{2,K_1, K_2}$ tends to $((1,0),(0,0))$, and as we have already shown in arbitrary spins, the finite volume ground states are given by the eigenspace corresponding to the pair $(\lambda,\rho)=(\varnothing,(n))$. This space is one-dimensional, since $d_{\varnothing}^{O(2)} = b_{\varnothing,(n)}^{n,2}=d_{(n)}^{S_n}=1$, and using (\ref{eqn:Espaces-general}), is spanned by the vector (\ref{eqn:dimer-states-1/2}). 

On the half-line $K_1=K_2>4$ (the supercritical isotropic Heisenberg model), the $y$ term in $\phi_{2,K_1,K_2}$ disappears, so if $(x^*,y^*)$ is a maximiser of $\phi_{2,K_1,K_2}$, then $(x^*,y)$ is too, so long as $(x^*,y)\in\Delta^*_2$. Hence $y_1^\uparrow=x_1^*-x_2^*>0$ by the proof of Theorem \ref{thm:PD-higher}, or \cite{Bj_rnberg_2016}. This concludes the proof of Proposition \ref{prop:proofy-bits-1/2}.    
\end{proof}

From Theorem \ref{thm:PD-1/2}, the transition from the disordered to either of the other two phases is second order, and from $XY$ to Ising is first order. By our working above, the transitions from the Ising to the other phases can also be observed in the quantities from Theorems \ref{thm:mag} and \ref{thm:total-spin}, since $y_1^\uparrow=y_1^*>0$ in the Ising phase, and is zero in the other phases. This transition in $y_1^\uparrow=y_1^*$ is continuous in the Ising-disordered transition, and discontinuous in the Ising-$XY$ transition.

\subsection{Spin 1; $\theta=3$}
\begin{proof}[Proof of Theorem \ref{thm:PD-1}]
Let $S=1$, and recall the Hamiltonian of the bilinear-biquadratic Heisenberg model:
\begin{equation}
        H'' =  - \left( 
        \sum_{x,y} J_1 S_x \cdot S_y + J_2 (S_x \cdot S_y)^2
        \right).
\end{equation}
where $J_1,J_2\in\mathbb{R}$ and $S_x\cdot S_y=\sum_{j=1}^3 S_x^jS_y^j$. Let $P_{x,y}$ be (a scalar multiple of) the spin-singlet operator, given by
\begin{equation*}
    \langle a_x, a_y| P_{x,y} |b_x, b_y \rangle = (-1)^{a_x-b_x}\delta_{a_x, -a_y} \delta_{b_x, -b_y}.
\end{equation*}
Note that the line $J_2=0$ gives the Heisenberg ferromagnet $(J_1>0)$, and antiferromagnet $(J_1<0)$. We use the relations $S_x \cdot S_y = T_{x,y} - P_{x,y}$ and $(S_x \cdot S_y)^2 = P_{x,y} + \mathrm{id}$ (see Lemma 7.1 from \cite{Ueltschi_2013}) to show that Hamiltonian (\ref{eqn:spin1hamiltonian}) is, up to addition of a constant,
\begin{equation}\label{eqn:spin1hamiltoniantransformed}
        H(n, J_1, J_2) = - \left(
        \sum_{x,y} J_1 T_{x,y} + (J_2 - J_1)P_{x,y}
        \right).
\end{equation}
Let $Z_n(J_1, J_2) = \mathrm{tr}(e^{-\frac{1}{n}H})$, where $H$ is given by \ref{eqn:spin1hamiltoniantransformed}. Ueltschi (Theorem 3.2 of \cite{Ueltschi_2013}) shows that for $\theta$ odd, this partition function is the same as when $P_{x,y}$ is replaced with $Q_{x,y}$. For completeness, we show that this equality can be derived from an isomorphism of representations (Lemma \ref{lem:isomorphic_reps}). Now, setting $L_1 = J_1$, $L_2 = J_2 - J_1$, Theorem \ref{thm:main-1.1} shows that the free energy of the model with Hamiltonian (\ref{eqn:spin1hamiltoniantransformed}) is 
\begin{equation*}
        \lim_{n \to \infty} \frac{1}{n} \log Z_n(J_1, J_2) 
        = 
        \max_{(x,y) \in \Delta^*_{3} } \phi_{3,J_1, J_2}(x,y),
    \end{equation*}
    where we have 
    \begin{equation}\label{eqn:phi-J1J2}
        \phi_{3,J_1, J_2}((x_1, x_2, x_3), (y_1, 0, 0))
        =
        \frac{1}{2} \left(
        J_2 \sum_{i=1}^3 x_i^2 + (J_1-J_2)y_1^2 \right) 
        - \sum_{i=1}^3 x_i \log(x_i).
    \end{equation}
    
The proof of Theorem \ref{thm:PD-1} now follows from Theorem \ref{thm:PD-higher} using the change of variables above. \\
\end{proof}

\begin{proof}[Proof of Proposition \ref{prop:proofy-bits-1}]
In the rest of this section we provide the proof of Proposition \ref{prop:proofy-bits-1}, which backs up Remark \ref{rmk:PD-1} and our description of the phases of the bilinear-biquadratic Heisenberg model, illustrated in Figure \ref{fig:PD-1}. 

Let $\phi = \phi_{3, J_1, J_2}$. We define the disordered phase to be the set $\mathcal{A}$ of values of $(J_1,J_2)$ such that $((\frac{1}{3},\frac{1}{3},\frac{1}{3}),(0,0,0))$ is a maximiser of $\phi$; this maximises the entropy term (the logarithms) of $\phi$.

Let us prove Proposition \ref{prop:proofy-bits-1} first in the region $J_2>J_1$. Here, we set $y_1=0$, so $\phi$ reduces to $\phi^{\mathrm{int}}$, and the disordered phase is the region $J_2\le\log16$, by \cite{Bj_rnberg_2016}.
We label the region $J_2> J_1$, $J_2>\log16$ the nematic phase. It is illustrated as the hatched blue region in Figure \ref{fig:PD-1}. As noted above, we must set $y_1=0$, so the magnetisation $y_1^\uparrow$ in Theorem \ref{thm:mag} is zero in this phase. We can say that the transition from disordered to nematic is first order, by Theorem \ref{thm:PD-1}. Lastly, let us show that for $J_2>J_1$, $J_2>0$, the finite (even) volume ground state is the vector (\ref{eqn:dimer-states-1}). By our working in Section \ref{section:proofofmainthm} and Proposition \ref{prop:P3}, the eigenvalues of the Hamiltonian (\ref{eqn:spin1hamiltoniantransformed}) are given by 
\begin{equation}\label{eqn:Evals-1}
    -\left[J_2c(\rho)  +(J_1-J_2)(c(\lambda)+k(1-\theta)) \right],
\end{equation}
where $(\lambda,\rho)$ are partitions of $n-2k$ ($0\le k\le\lfloor\frac{n}{2}\rfloor$) and $n$, respectively, with $\lambda^T_1+\lambda^T_2\le3$, $\rho^T_1\le3$, and $\lambda_1\le\rho_1-\rho_3$. 
For $J_2> J_1$, $J_2>0$, as we have already shown in arbitrary spins, this is minimised by the pair $(\lambda,\rho)=(\varnothing,(n))$, the corresponding eigenspace has dimension $d_{\varnothing}^{O(3)} b_{\varnothing,(n)}^{n,2}d_{(n)}^{S_n}=1$, and using (\ref{eqn:Espaces-general}), the unique ground state of the transformed Hamiltonian (ie. (\ref{eqn:spin1hamiltoniantransformed}) with $P_{x,y}$ replaced with $Q_{x,y}$) is the vector given by the sum over all $\frac{n}{2}$-fold tensor products of the vector $\sum_{a=-1}^1 |a,a\rangle$. Transforming this back to the original Hamiltonian, we have the sum over all possible tensor products of singlet states, which is precisely (\ref{eqn:dimer-states-1}).\\

We can now turn to proving Proposition \ref{prop:proofy-bits-1} in the region $J_2 \le J_1$; this region is more complicated. The function $\phi$ does not reduce to $\phi^{\mathrm{int}}$. We must let $y_1 = x_1 - x_3$. Setting $x_3 = 1 - x_1 - x_2$, we rewrite $\phi$ as a function of $x_1$ and $x_2$:
\begin{equation}\label{eqn:phi-J1>J2}
    \begin{split}
        \phi = \phi_{3, J_1, J_2}(x_1, x_2)
        &=
        \frac{1}{2} \left(
        J_2(-2x_1^2 + x_2^2 - 2x_1x_2 +2x_1) + J_1(2x_1 + x_2-1)^2 \right) \\
        &\ \ \ - x_1 \log(x_1) - x_2 \log(x_2) - (1-x_1-x_2)\log(1-x_1-x_2).
    \end{split}
\end{equation}
Note we are analysing this function in the region $R$ defined by $x_1 \ge x_2$, $1-x_2 \ge x_1 \ge 1-2x_2$ (see Figure \ref{fig:theregionR}).

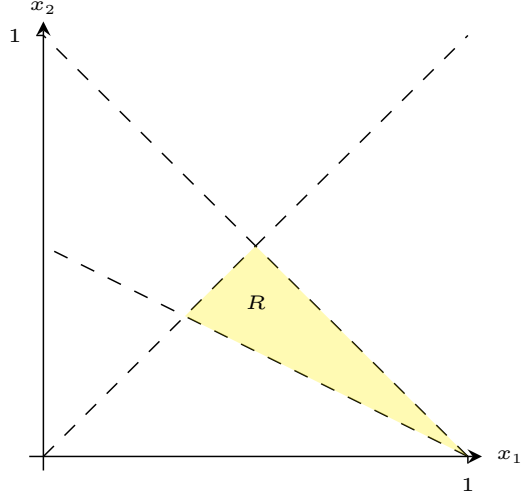
\begin{figure}[h]
    \centering
    \resizebox{0.5\textwidth}{!}{%
    \begin{tikzpicture}[scale=1]

    \draw[->, >=stealth, line width=0.3pt](-0.1, 0) -- (3.1,0);
    \draw[->, >=stealth, line width=0.3pt](0,-0.1) -- (0,3.1);
    \draw[-, >=stealth, dashed, line width=0.3pt](0,0) -- (3,3);
    \draw[-, >=stealth, dashed, line width=0.3pt](3,0) -- (0,3);
    \draw[-, >=stealth, dashed, line width=0.3pt](3,0) -- (0,1.5);
    
    \draw[-, >=stealth, line width=0.3pt](0,3) -- (-0.05, 3);
    \draw[-, >=stealth, line width=0.3pt](3,0) -- (3,-0.05);
    
    \path[fill=yellow, opacity = 0.3](1,1) -- (1.5,1.5) -- (3,0) -- (1,1);
    
    \draw[font=\fontsize{4}{6}] (-0.2, 3) node {$1$};
    \draw[font=\fontsize{4}{6}] (3, -0.2) node {$1$};
    \draw[font=\fontsize{4}{6}] (0, 3.2) node {$x_2$};
    \draw[font=\fontsize{4}{6}] (3.3, 0) node {$x_1$};
    \draw[font=\fontsize{4}{6}] (1.5, 1.1) node {$R$};

    \end{tikzpicture}
    }
    \caption{The region $R$.}
    \label{fig:theregionR}
    
\end{figure}

In this region $J_1\ge J_2$, the boundary of the disordered phase $\mathcal{A}$ is difficult to identify - recall we will show it is a curve $\mathcal{C}$ made up of the half-line $J_2=2J_1-3\le\frac{3}{2}$ and a curve connecting the points $(\frac{9}{4},\frac{3}{2})$ and $(\log16,\log16)$. Outside of the disordered phase $\mathcal{A}$ (within the region $J_1\ge J_2$), we can show that $y_1^\uparrow$ from Theorem \ref{thm:mag} is strictly positive. Indeed, if $(x^*,y^*)$ is a maximiser of $\phi_{3,J_1,J_2}$, then $y^*_1=x^*_1-x^*_3$, meaning the only point with $y^*_1=0$ is $((\frac{1}{3},\frac{1}{3},\frac{1}{3}),(0,0,0))$, and the claim follows from the definition of the disordered phase $\mathcal{A}$. Numerical simulations suggest that $y_1^\uparrow$ is a unique maximiser everywhere in $J_1\ge J_2$ except for the curve between $(\frac{9}{4},\frac{3}{2})$ and $(\log16,\log16)$ which is part of the curve $C$, so the $y_1^*$ form Theorem \ref{thm:total-spin} would exist and be positive; we have not been able to prove this.

Let us consider the ground state behaviour outside the disordered phase.
As $||(J_1, J_2)|| \to \infty$, the logarithm terms in $\phi$ will become negligible. Let $\phi_0$ be $\phi$ with the logarithm terms removed. We maximise $\phi_0$ in the region $R$. Setting $x_3=1-x_1-x_2$, we have
\begin{equation}
    \begin{split}
        \frac{\partial \phi_0}{\partial x_1} 
        &= 
        (2J_1-J_2)(2x_1+x_2-1);\\
        \frac{\partial \phi_0}{\partial x_2} 
        &= 
        J_1(2x_1+x_2-1) +J_2(x_2-x_1).
    \end{split}
\end{equation}
Now since $2J_1 - J_2 >0$, (and as we take our limit we are beyond the conjectured boundary $\mathcal{C}$), so the maximum of $\phi_0$ must lie on the boundary line $x_1 + x_2 =1$ of $R$. Note this implies $x_3=0$, and so $y_1 = x_1$. Substituting $x_2 = 1-x_1$, and rearranging, we have the quadratic
\begin{equation}\label{eqn:quadratic}
    \phi_0(x_1) = \frac{J_1 + J_2}{2}
    \left( \left( 
    x_1 - \frac{J_2}{J_1 + J_2} \right)^2
    +
     \frac{J_1 J_2}{(J_1 + J_2)^2}
    \right),
\end{equation}
where recall we are concerned with the region $x_1 \in [\frac{1}{2},1]$. 

Calculations yield that in the region $J_2 < J_1$, $J_1 \ge 0$, this quadratic has maximum at $x_1=1$. So the maximiser of $\phi_{3,J_1,J_2}$ in this region as $||(J_1,J_2)||\to\infty$ tends to $((1,0,0),(1,0,0))$. Indeed the finite volume ground states (of the transformed Hamiltonian (\ref{eqn:spin1hamiltoniantransformed})) are given by the eigenspace corresponding to the pair $(\lambda,\rho)=((n),(n))$. This space has dimension $d_{(n)}^{O(3)} b_{(n),(n)}^{n,2}d_{(n)}^{S_n}=2n+1$, and using (\ref{eqn:Espaces-general}), is the set of vectors invariant under $S_n$ (equivalently, invariant under any $T_{x,y}$), which are killed by any $Q_{x,y}$. The corresponding eigenspace of the original Hamiltonian is the set invariant under $S_n$ and killed by any $P_{x,y}$. Straightforward analysis of $P_{x,y}$ shows that the product states $\bigotimes_{x\in V}|a\rangle$ with $a_0^2-a_1a_{-1}=0$ lie in this set (although they do not span it), which include the ferromagnetic $|a\rangle= |1\rangle$ and $|-1\rangle$ as well as $|1\rangle+|0\rangle+|-1\rangle$. We label this region, $J_2 < J_1$, $J_1 \ge 0$ and to the right of the curve $\mathcal{C}$, ferromagnetic. It is illustrated as the dotted yellow region in Figure \ref{fig:PD-1}.

Now consider the region $2J_1 - J_2 >0$, and $J_1 < 0$, illustrated by the checkerboard orange region in Figure \ref{fig:PD-1}. In this case, the quadratic (\ref{eqn:quadratic}) has maximum at $x_1 = \alpha:=\frac{J_2}{J_1 + J_2}$, which lies in the range $[\frac{2}{3},1]$. Then the maximiser of $\phi_{3,J_1,J_2}$ is $((\alpha,1-\alpha,0),(\alpha,0,0))$. While we do not label this fourth phase, occupying the region $0>J_1>\frac{1}{2}(3-J_2)$, this phase has some similarity with the ferromagnetic phase. 

In finite volume, calculations from analysing (\ref{eqn:Evals-1}) show that the set of ground states in the fourth phase is the eigenspace corresponding to a pair $(\lambda,\rho)=((\alpha'),(\alpha',1-\alpha'))$, where $\alpha'=\lambda_1/n$ is close to $\alpha=\frac{J_2}{J_1 + J_2}$ (and tends to $\alpha$ as $n\to \infty$). Using (\ref{eqn:Espaces-general}), the eigenspace is spanned by vectors (\ref{eqn:partial-ferro-state}).\\

The rest of this section completes the proof of Proposition \ref{prop:proofy-bits-1} by determining the boundary of the disordered phase $\mathcal{A}$ within the region $J_1\ge J_2$. Recall we will show it is a curve $\mathcal{C}$ made up of the half-line $J_2=2J_1-3\le\frac{3}{2}$ and a curve connecting the points $(\frac{9}{4},\frac{3}{2})$ and $(\log16,\log16)$. From here till the end of the section we work with $\phi$ given in (\ref{eqn:phi-J1>J2}). The partial derivatives of $\phi$ of first and second order are:
\begin{equation}\label{eqn:partialsofphi}
    \begin{split}
        \frac{\partial \phi}{\partial x_1} 
        &= 
        (2J_1-J_2)(2x_1+x_2-1) -\log(x_1) + \log(1-x_1-x_2);\\
        \frac{\partial \phi}{\partial x_2} 
        &= 
        J_1(2x_1+x_2-1) +J_2(x_2-x_1) -\log(x_2) + \log(1-x_1-x_2);\\
        \frac{\partial^2 \phi}{\partial x_1^2}
        &= 
        2(2J_1-J_2) -\frac{1}{x_1} - \frac{1}{1-x_1-x_2};\\
        \frac{\partial^2 \phi}{\partial x_1 \partial x_2} 
        &= 
        (2J_1-J_2) - \frac{1}{1-x_1-x_2};\\
        \frac{\partial^2 \phi}{\partial x_2^2} 
        &= 
        J_1 + J_2 -\frac{1}{x_1} - \frac{1}{1-x_1-x_2}.
    \end{split}
\end{equation}

\begin{lem}\label{lem:(1/3,1/3)isalocalmax}
    The point $(x_1, x_2) = (\frac{1}{3}, \frac{1}{3})$ is always an inflection point of $\phi$, and it is a local maximum point if $2J_1-J_2 < 3$, and if $2J_1-J_2 > 3$ it is not a local maximum point and does not maximise $\phi$ in $R$. 
\end{lem}
\begin{proof}
    Setting $(x_1, x_2) = (\frac{1}{3}, \frac{1}{3})$ in the above shows it is always an inflection point. If $\mathfrak{H}$ is the Hessian matrix of $\phi$, then for any vector $(p,q) \in \mathbb{R}^2$, we have
    \begin{equation*}
        (p,q) \mathfrak{H} (p,q)^T
        =
        (p+q)^2(2J_1-J_2-3) + p^2((2J_1-J_2-3) + b^2(2J_2-J_1-3)),
    \end{equation*}
    which is negative for all $2J_1 -J_2 < 3$ in our region $J_2 \le J_1$, meaning $(\frac{1}{3}, \frac{1}{3})$ is a local maximum. Clearly for $2J_1 -J_2 > 3$, $\frac{\partial^2 \phi}{\partial x^2_1} >0$, so $(\frac{1}{3}, \frac{1}{3})$ is not a local maximum, and it cannot maximise $\phi$ in the region $R$.
\end{proof}

Let $\mathcal{A'}$ be the region within the region $J_2 \le J_1$ where $(\frac{1}{3}, \frac{1}{3})$ is a global maximum of $\phi$ in $R$ (this is the region $\mathcal{A}$ intersected with $J_2 \le J_1$). By the above, all of $\mathcal{A'}$ must lie within the region $2J_1-J_2 \le 3$, (or, not to the right of the line $J_2 = 2J_1-3$). 

\begin{lem}\label{lem:Aconvex}
    The set $\mathcal{A'}$ is convex.
\end{lem}
\begin{proof}
Let $J^{(1)}, J^{(2)}$ be two points in $\mathcal{A'}$. Let $J = sJ^{(1)} + (1-s)J^{(2)}$, $s \in [0,1]$. Since $\phi$ is linear in $J_1, J_2$, we have that for any $(x_1, x_2) \in R$, 
\begin{equation*}
    \begin{split}
        \phi_{2,J}(x_1, x_2) 
        &= 
        s\phi_{2, J^{(1)}}(x_1, x_2) + (1-s)\phi_{2, J^{(2)}}(x_1, x_2) \\
        &\le 
        s\phi_{2, J^{(1)}}\left(\frac{1}{3}, \frac{1}{3}\right) 
        + (1-s)\phi_{2,J^{(2)}}\left(\frac{1}{3}, \frac{1}{3}\right)
        = \phi_{2,J}\left(\frac{1}{3}, \frac{1}{3}\right).
    \end{split}
\end{equation*}
\end{proof}

\begin{lem}\label{lem:movingJ}
    If a point $(J_1^{(0)}, J_2^{(0)})$ within the region $J_1\ge J_2$ lies outside of $\mathcal{A'}$, then the point $(J_1^{(0)}, J_2^{(0)}) + v$, when it lies within the region $J_1\ge J_2$, also lies outside of $\mathcal{A'}$, for any $v = \mu(1,2) + \nu (-1,-3)$, $\mu, \nu>0$.
\end{lem}
\begin{proof}
We have that 
\begin{equation*}
    \begin{split}
        \frac{\partial}{\partial J_1}
        \left(\phi(x_1, x_2) - \phi\left(\frac{1}{3}, \frac{1}{3}\right)\right)
        &=
        \frac{1}{2} (2x_1 + x_2 -1)^2\\
        \frac{\partial}{\partial J_2}
        \left(\phi(x_1, x_2) - \phi\left(\frac{1}{3}, \frac{1}{3}\right)\right)
        &=
        \frac{1}{2} \left(1-2x_1^2+x_2^2-2x_1x_2 +2x_1 - \frac{4}{3}\right).
    \end{split}
\end{equation*}
Firstly, we consider:
\begin{equation*}
    \begin{split}
        \left(\frac{\partial}{\partial J_1} + 2\frac{\partial}{\partial J_2}\right)
        \left(\phi(x_1, x_2) - \phi\left(\frac{1}{3}, \frac{1}{3}\right)\right)
        &=
        \frac{3}{2}x_2^2 -x_2 +\frac{1}{6} 
        =
        \frac{1}{6}(3x_2-1)^2,
    \end{split}
\end{equation*}
which has single root and minimum at $x_2 = \frac{1}{3}$. Let $(J_1^{(0)}, J_2^{(0)})$ lie outside of $\mathcal{A'}$, so there exists some global maximiser $(x_1^*, x_2^*) \neq (\frac{1}{3}, \frac{1}{3})$ in $R$, $\phi(x_1^*, x_2^*) > \phi(\frac{1}{3}, \frac{1}{3})$. Now the above shows that moving $(J_1^{(0)}, J_2^{(0)})$ in the direction $(1,2)$ does not increase $\phi$ at $(\frac{1}{3}, \frac{1}{3})$ any faster than at any other point of $R$, so for all $\mu >0$, $(J_1^{(0)}, J_2^{(0)}) + \mu (1,2)$ cannot lie in $\mathcal{A'}$. 

Secondly, 
\begin{equation*}
    \begin{split}
        \left(-\frac{\partial}{\partial J_1} - 3\frac{\partial}{\partial J_2}\right)
        &=
        -\frac{1}{2} \left(1-2x_1^2+x_2^2-2x_1x_2 +2x_1 - \frac{4}{3}\right)\\
        &=
        \left(\left(x_1-\frac{1}{3}\right)+\frac{1}{2} \left(x_2-\frac{1}{3}\right)\right)^{2}
        -
        \frac{9}{4}\left(x_2-\frac{1}{3}\right)^{2},
    \end{split}
\end{equation*}
which takes the value zero exactly on the lines $x_1 = x_2$ and $x_1 = 1-2x_2$, two of the boundary lines of $R$, and is positive in the rest of $R$. By the same argument as above, if $(J_1^{(0)}, J_2^{(0)}) \notin A$, then $(J_1^{(0)}, J_2^{(0)}) + \nu (-1,-3) \notin A$, for all $\nu>0$. The lemma follows.
\end{proof}

\begin{lem}\label{lem:partialregionofA}
    The region bounded by and including the line $J_1=J_2$, $J_1 \le \log(16)$, the line $2J_1-J_2 = 3, J_2 \le \frac{3}{2}$, and the straight line from the point $(\log(16), \log(16))$ to the point $(\frac{9}{4}, \frac{3}{2})$, lies within $\mathcal{A'}$.
\end{lem}

\begin{proof}
To begin with, note that on the line $J_1=J_2$, we can use our results from the case $J_2 \ge J_1$. This means all $J_1=J_2$, $J_1 \le \log(16)$ lie in $\mathcal{A'}$. Now the previous lemma implies that all $2J_1 - J_2 \le \log(16)$ lies in $\mathcal{A'}$, since if it were not true, we would be able to move from a point not in $\mathcal{A'}$ in the direction $(1,2)$ and arrive at a point in $\mathcal{A'}$. Now by the same logic, and the fact that $\mathcal{A'}$ is convex, it suffices to show that the point $(J_1, J_2) = (\frac{9}{4}, \frac{3}{2})$ lies in $\mathcal{A'}$. We show that at this point, there are no inflection points of $\phi$ besides $(\frac{1}{3}, \frac{1}{3})$, and $(\frac{1}{3}, \frac{1}{3})$ maximises $\phi$ on the boundary of $R$.

Set $(J_1, J_2) = (\frac{9}{4}, \frac{3}{2})$. Substituting $z = 1-x_1- x_2, w = x_1$ into (\ref{eqn:partialsofphi}) gives $\frac{\partial \phi}{\partial x_1} =0$ if and only if $z=1$ or 
\begin{equation}\label{eqn:dphibydx_1=0}
    w = \frac{-\log(z)}{3(1-z)}.
\end{equation}
Note that the region $R$ is transformed into $R'$, given by $\frac{1}{2}(\frac{1}{w} -1) \ge z \ge \frac{1}{w} -2$, $zw \ge 0$. The line $z=1$ intersects $R'$ at the single point $(w,z) = (\frac{1}{3}, 1)$, which corresponds to $(x_1, x_2) = (\frac{1}{3}, \frac{1}{3})$. Substituting (\ref{eqn:dphibydx_1=0}) into (\ref{eqn:partialsofphi}) gives that, on the line where $\frac{\partial \phi}{\partial x_1} =0$, the value of $\frac{\partial \phi}{\partial x_2}$ is:
\begin{equation*}
    \frac{\partial \phi}{\partial x_1}(z)
    =
    \frac{3}{2}
    + 
    \frac{\log(z)(1+5z)}{4(1-z)}
    +
    \log \left(
    \frac{-z\log(z)}{3(1-z) + (1+z)\log(z)}
    \right).
\end{equation*}
\begin{rmk}\label{rmk:theremark}
    Let $r$ be the unique zero of $3(1-z)+(1+z)\log(z)$. This function $\frac{\partial \phi}{\partial x_1}(z)$ is positive in the range $(r,1)$, except at $z=1$, where it is zero. (It is not defined in $(0,r]$). 
\end{rmk}
Hence either there are no points of inflection in $R$, or $(\frac{1}{3}, \frac{1}{3})$ is the only one. Proving Remark \ref{rmk:theremark} by hand is difficult. However, a rigorous computer-assisted argument is available, which is due to Dave Platt. See Appendix \ref{section:appendix}. 

It remains to analyse $\phi$ on the boundary of $R$. Substituting $x_1 = 1-2x_2$ into $\phi$, we have
\begin{equation*}
    \phi(x_2) = \frac{15}{8} 
    -
    \frac{21}{4}x_2
    + 
    \frac{63}{8} x_2^2
    -
    (1-2x_2)\log(1-2x_2)
    -
    2x_2\log(x_2),
\end{equation*}
which it is not hard to prove is maximised at $x_2 = \frac{1}{3}$ in the region $x_2 \in [0,\frac{1}{2}]$. Indeed, its first derivative
$
    \frac{1}{4}(-21 +63x_2 +8 \log(1-2x_2) -8\log(x_2))
$
is zero at $x_2 =0$, and its second derivative
$
    \frac{126 x_2^2 - 62x_2 +8}{8x^2 - 4x}
$
is negative in the range. Substituting $x_1 =x_2$ into $\phi$ gives exactly the same function as above. As $x_1 + x_2 \to 1$, the first order derivatives of $\phi$ tend to $-\infty$. Hence $\phi$ on the boundary of $R$ must be maximised at $(\frac{1}{3}, \frac{1}{3})$, so the same holds over all $R$, and so we can conclude that $(\frac{9}{4}, \frac{3}{2}) \in \mathcal{A'}$, which is what we wanted to prove; this completes the proof of Lemma \ref{lem:partialregionofA}.

\end{proof}


Combining the above lemmas give us the information we need about the boundary of $\mathcal{A'}$. Lemma \ref{lem:Aconvex} implies that its boundary exists, and adding Lemmas \ref{lem:movingJ}, \ref{lem:(1/3,1/3)isalocalmax} and \ref{lem:partialregionofA} shows that its boundary is made up of the line $J_1=J_2 \le \log(16)$, and a curve $\mathcal{C}$ which (as a function of $J_1$) is a continuous, convex line, which is the line $2J_1 - J_2 =3$ for $J_2 \le \frac{3}{2}$, and that its gradient lies in $[2,3]$. This curve must meet the line $J_1=J_2$ at the point $(\log(16), \log(16))$. Indeed, recall $\mathcal{A}$ is the region of the whole plane where $((\frac{1}{3}, \frac{1}{3}, \frac{1}{3}), (0,0,0))$ maximises our original $\phi_{3, J_1, J_2}(x,y)$; we have $\mathcal{A}' = \mathcal{A} \cap \{ J_2 \le J_1 \}$. Then the same proof as above can be employed to show that $\mathcal{A}$ is convex, which is what we need. This completes the description of the boundary of $\mathcal{A}$, which in turn completes the proof of Proposition \ref{prop:proofy-bits-1}.
\end{proof}

\section{Branching Coefficients}\label{section:restrictiontheorems}

As noted in Section \ref{section:proofofmainthm}, the aim of this section is to prove Propositions \ref{prop:P2} and \ref{prop:P3}, which determine the sets $P_n(\theta)$ and the limits $\frac{1}{n}P_n(\theta)$, $\theta =2,3$. Recall $\Lambda_n(\theta)$ is the set of pairs of partitions $(\lambda, \rho)$, $\lambda \vdash n-2k$, $0 \le k \le \lfloor \frac{n}{2} \rfloor$, $\rho \vdash n$, such that $\lambda_1^T + \lambda_2^T \le \theta$ and $\rho_1^T \le \theta$. Recall that $b^{n,\theta}_{\lambda, \rho}$ is the coefficient of the irreducible $\rho$ in the restriction of the irreducible $\lambda$ from $\mathbb{B}_{n,\theta}$ to $\mathbb{C}S_n$. Then $P_n(\theta)$ is the set of $(\lambda, \rho) \in \Lambda_n(\theta)$ such that $b^{n,\theta}_{\lambda, \rho}>0$. Most of the work in proving Propositions \ref{prop:P2} and \ref{prop:P3} is contained in three lemmas which we begin this section with. The first shows that the coefficients $b^{n,\theta}_{\lambda, \rho}$ are also the branching coefficients of the orthogonal and general linear groups. The second is a useful recurrence relation, and the third determines $b^{n,\theta}_{\lambda, \rho}$ for certain values of $\rho$, in terms of the Littlewood-Richardson coefficients.

\subsection{Useful lemmas for all $\theta$}
Fix $\theta \ge 2$. First we rephrase the coefficients $b^{n,\theta}_{\lambda, \rho}$ in terms of the general linear and orthogonal groups, using Schur-Weyl duality. Recall that the irreducible polynomial representations of $GL(\theta)$ are indexed by $\rho$, partitions of any non-negative integer with at most $\theta$ parts. Similarly, those of $O(\theta)$ are indexed by $\lambda$, partitions of any non-negative integer whose first two columns sum to at most $\theta$. Let $\rho \vdash n$, and let $g^{n,\theta}_{\lambda, \rho}$ denote the coefficient of $\psi_\lambda^{O(\theta)}$ in the restriction of $\psi_\rho^{GL(\theta)}$ from $GL(\theta)$ to $O(\theta)$. 

\begin{lem}\label{lem:GltoO_is_BntoSn}
    The symmetric group-Brauer algebra and orthogonal group-general linear group branching coefficients are the same. That is, for all $(\lambda, \rho) \in \Lambda_n(\theta)$, we have that $g^{n,\theta}_{\lambda, \rho} = b^{n,\theta}_{\lambda, \rho}$. 
\end{lem}
\begin{proof}
    Recall that Schur-Weyl duality (\ref{eqn:schur_weyl_duality_brauer}) states that as a module of $\mathbb{B}_{n,\theta} \otimes \mathbb{C}O_n(\mathbb{C})$, 
    \begin{equation}
        \mathcal{H}_V = \bigoplus_{\substack{\lambda \vdash n-2k \\ \lambda_1^T + \lambda_2^T \le \theta}}
        \psi^{\mathbb{B}_{n,\theta}}_{\lambda}
        \boxtimes
        \psi^{O(\theta)}_{\lambda}.
    \end{equation}
    The equivalent statement for the symmetric and general linear groups says that as a module of $\mathbb{C}S_n \otimes \mathbb{C}GL_n(\mathbb{C})$,
    \begin{equation}\label{eqn:SW-Sn}
        \mathcal{H}_V = \bigoplus_{\substack{\rho \vdash n \\ \rho_1^T\le \theta}}
        \psi^{S_n}_{\rho}
        \boxtimes
        \psi^{GL(\theta)}_{\rho}.
    \end{equation}
    Restricting each $\psi^{\mathbb{B}_{n,\theta}}_{\lambda}$ in the first equation to $\mathbb{C}S_n$, and each $\psi^{GL(\theta)}_{\rho}$ in the second to $O_n(\mathbb{C})$, we have, as a module of $\mathbb{C}S_n \otimes \mathbb{C}O_n(\mathbb{C})$,
    \begin{equation}
        \bigoplus_{(\lambda,\rho) \in \Lambda_n(\theta)}
        b^{n,\theta}_{\lambda, \rho}
        \psi^{S_n}_{\rho}
        \boxtimes
        \psi^{O(\theta)}_{\lambda}
        =
        \bigoplus_{(\lambda,\rho) \in \Lambda_n(\theta)}
        g^{n,\theta}_{\lambda, \rho}
        \psi^{S_n}_{\rho}
        \boxtimes
        \psi^{O(\theta)}_{\lambda},
    \end{equation}
    and hence the result.
\end{proof}

From hereon in we simply use $b^{n,\theta}_{\lambda, \rho}$ to denote either itself or $g^{n,\theta}_{\lambda, \rho}$. For $\lambda$ a partition with $\lambda_1^T + \lambda_2^T \le \theta$, recall $\lambda'$ is the partition such that $\lambda'^T_1 = \theta - \lambda^T_1$, and $\lambda'^T_j = \lambda_j^T$ for all $j>1$. Note that $\lambda''=\lambda$. We next prove a useful recurrence relation. Let 

\begin{lem}\label{lem:recurrence_relation}
    The symmetric group-Brauer algebra branching coefficients satisfy the following recurrence relation. Let $(\lambda, \rho) \in \Lambda_n(\theta)$, such that $\rho_\theta >0$. Then $b^{n,\theta}_{\lambda, \rho} = b^{n-\theta, \theta}_{\lambda', \rho-\underline{1}}$, where $\underline{1}$ is the partition with all parts equal to $1$. 
\end{lem}
\begin{rmk}\label{rmk:recurrence_relation}
    Note that as a consequence, if $(\lambda, \rho) \in \Lambda_n(\theta)$, $\rho_\theta >0$, then
    \begin{equation*}
        b^{n,\theta}_{\lambda, \rho} 
        =
        \begin{cases}
          b^{n,\theta}_{\lambda, \rho - \underline{\rho_\theta}} \ \ \ \rho_\theta \ \mathrm{even}\\  
          b^{n,\theta}_{\lambda', \rho - \underline{\rho_\theta}} \ \ \ \rho_\theta \ \mathrm{odd},\\ 
        \end{cases}
    \end{equation*}
    where $\underline{\rho_\theta}$ is the partition with all parts equal to $\rho_\theta$. 
\end{rmk} 
\begin{proof}
    We use the fact that $b^{n,\theta}_{\lambda, \rho}$ is a coefficient in the restriction of the irreducible $\psi_\rho^{GL(\theta)}$ of $GL(\theta)$ to $O(\theta)$. By character orthogonality, we have
    \begin{equation}\label{eqn:integrate_on_orthogonal_group}
        b^{n,\theta}_{\lambda, \rho}
        =
        \int_{O(\theta)}
        \chi_\rho^{GL(\theta)}(g)
        \chi_\lambda^{O(\theta)}(g)
        \ dg,
    \end{equation}
    where $dg$ denotes the Haar measure on the orthogonal group. By the Pieri rule (or, for example, the remarks after equation (1) of \cite{stembridge}), $\chi_\rho^{GL(\theta)} = \chi_{\rho-\underline{1}}^{GL(\theta)}  \chi^{GL(\theta)}_{\underline{1}}$. Then we note that $\chi^{GL(\theta)}_{\underline{1}}$ is the determinant character of $GL(\theta)$ (or $O(\theta)$, when restricted), and that $\chi^{O(\theta)}_{\underline{1}} \chi_\lambda^{O(\theta)} = \chi_{\lambda'}^{O(\theta)}$ (see, for example, the remark after Proposition 2.6 in \cite{okada_pieri_type_rules}). Substituting into (\ref{eqn:integrate_on_orthogonal_group}) completes the proof.
\end{proof}

The last lemma in this subsection gives us control of the coefficients $b^{n,\theta}_{\lambda, \rho}$ for certain values of $\rho$. In order to prove it, we need to introduce some more representation theory of the Brauer algebra.

The Brauer algebra's semisimplicity is dependent on the parameter $\theta$. When $\theta$ is a positive integer, $\mathbb{B}_{n,\theta}$ is semisimple if and only if $\theta \ge n-1$ and when $\theta \notin \mathbb{Z}$, it is always semisimple. See \cite{wenzl}, \cite{rui}. The Brauer algebra has indecomposable representations, known as the cell modules (see \cite{cox2007blocks}), indexed by partitions $\lambda \vdash n-2k$, $0 \le k \le \lfloor \frac{n}{2} \rfloor$. When $\mathbb{B}_{n,\theta}$ is semisimple, these are exactly the irreducibles. Their characters are described by Ram \cite{ram1995}. Note that Ram's results on the cell characters are stated for when the algebra is semisimple, but they extend to the case when it is not. 

When $\mathbb{B}_{n,\theta}$ is not semisimple, the cell modules are not necessarily irreducible (in fact they are not even necessarily semisimple). The irreducible representation corresponding to $\lambda$ is then a quotient of the cell module corresponding to $\lambda$. Let us denote the character of the cell module corresponding to $\lambda$ by $\gamma_{\lambda}^{\mathbb{B}_{n,\theta}}$.

The restrictions of representations of $\mathbb{C}S_n$ to $\mathbb{C}S_{n-1}$ and $\mathbb{B}_{n,\theta}$ to $\mathbb{B}_{n-1,\theta}$ are well studied. Let $\rho \vdash n$, $\lambda \vdash n-2k$, $0 \le k \le \lfloor \frac{n}{2} \rfloor$. We have the following, in terms of characters (see, for example, Sections 4 and 5 (and Figures 1 and 2) from \cite{Doty_2019}, and Proposition 1.3 from \cite{nazarov}):
\begin{equation}\label{eqn:branchingrules1}
    \begin{split}
        \mathrm{res}^{S_n}_{S_{n-1}} [\chi_{\rho}^{S_n}]
        &=
        \sum_{\overline{\rho} = \rho - \square} \chi^{S_{n-1}}_{\overline{\rho}};\\
        \mathrm{res}^{\mathbb{B}_{n,\theta}}_{\mathbb{B}_{n-1,\theta}}
        [\gamma_{\lambda}^{\mathbb{B}_{n,\theta}}]
        &=
        \sum_{\overline{\lambda} = \lambda \pm \square} \gamma^{\mathbb{B}_{n-1,\theta}}_{\overline{\lambda}};
    \end{split}
\end{equation}
and if $\theta\ge2$ is an integer and $\lambda$ further satisfies $\lambda^T_1 + \lambda^T_2 \le \theta$,
\begin{equation}\label{eqn:branchingrules2}
        \mathrm{res}^{\mathbb{B}_{n,\theta}}_{\mathbb{B}_{n-1,\theta}}
        [\chi_{\lambda}^{\mathbb{B}_{n,\theta}}]
        =
        \sum_{\substack{\overline{\lambda} = \lambda \pm \square \\ \overline{\lambda}^T_1 + \overline{\lambda}^T_2 \le \theta} }
        \chi^{\mathbb{B}_{n-1,\theta}}_{\overline{\lambda}},
\end{equation}
where in the first equality the sum is over all $\overline{\rho} \vdash n-1$ whose Young diagram can be obtained from that of $\rho$ by removing a box; in the second the sum is over $\overline{\lambda} \vdash n-1-2r$, $0 \le r \le \lfloor\frac{n-1}{2} \rfloor$, whose Young diagram can be obtained from that of $\lambda$ by removing or adding a box; and in the third the sum is the same as the second, except we are restricted to those $\overline{\lambda}$ with $\overline{\lambda}^T_1 + \overline{\lambda}^T_2 \le \theta$.

We now describe how cell modules of $\mathbb{B}_{n,\theta}$ decompose when restricted to $\mathbb{C}S_n$.
This result is a special case of Theorem 5.1 from \cite{ram1995}. We call a partition $\pi$ even if all its parts $\pi_i$ are even. Let $\lambda \vdash n-2k$, $0 \le k \le \lfloor\frac{n}{2}\rfloor$. Then
\begin{equation}\label{eqn:ramsrestrictionofcells}
    \mathrm{res}^{\mathbb{B}_{n,\theta}}_{S_n}
    [\gamma_{\lambda}^{\mathbb{B}_{n,\theta}}]
    =
    \sum_{\rho \vdash n }
    \tilde{b}^{n, \theta}_{\lambda, \rho}
    \chi_{\rho}^{S_{n}}
    =
    \chi_{\lambda}^{S_{n-2k}} 
    \times
    \sum_{\substack{\pi \vdash 2k \\ \pi \ \mathrm{even}}}
    \chi_{\pi}^{S_{2k}},
\end{equation}
or,
\begin{equation*}
    \tilde{b}^{n, \theta}_{\lambda, \rho}
    =
    \sum_{\substack{\pi \vdash 2k \\ \pi \ \mathrm{even}}}
    c_{\lambda, \pi}^{\rho}.
\end{equation*}
Let us make a few useful remarks.
\begin{rmk}\label{rmk:usefuls}
    \begin{enumerate}
        \item Since the irreducible representation of $\mathbb{B}_{n,\theta}$ corresponding to $\lambda$ is a quotient of the cell module corresponding to $\lambda$, we have $b^{n, \theta}_{\lambda, \rho} \le \tilde{b}^{n, \theta}_{\lambda, \rho}$ for all $\lambda, \rho$. 
        \item Since $c^\rho_{\lambda,\pi}$ is determined by $\pi$ and the skew-diagram $\rho\setminus\lambda$, we have that $\tilde{b}^{n, \theta}_{\lambda, \rho}$ is fully determined by the skew-diagram $\rho \setminus \lambda$.
        \item If $\lambda\nleq\rho$ then $c^\rho_{\lambda,\pi}=0$, so as a consequence, $\tilde{b}^{n, \theta}_{\lambda, \rho}=0$ (and therefore $b^{n, \theta}_{\lambda, \rho}=0$) if $\lambda \nleq \rho$. 
        \item Combining the above with Remark \ref{rmk:recurrence_relation}, we have that $b^{n, \theta}_{\lambda, \rho} = 0$ in the following cases: if $\lambda_j > \rho_j - \rho_\theta$ for $j \le \lfloor \theta/2 \rfloor$, or if $\rho_j = \rho_\theta$, $j> \lfloor \theta/2 \rfloor$ with either $\rho_\theta$ odd, $\lambda_j =0$, or $\rho_\theta$ even, $\lambda_j =1$.
    \end{enumerate}
    
\end{rmk}

\begin{lem}\label{lem:rho_columns_<=_theta+1}
    Let $(\lambda, \rho) \in \Lambda_n(\theta)$, such that $\rho_1^T + \rho_2^T \le \theta +1$. Then $b^{n,\theta}_{\lambda, \rho} = \tilde{b}^{n,\theta}_{\lambda, \rho}$.
\end{lem}
\begin{proof}
    We work by induction on $n$. The base case, $n=1$, is straightforward, since $\mathbb{B}_{1, \theta} = \mathbb{C}S_1$. Assume the theorem is proved for $n-1, n-2, \dots$. 
    Since $\rho_1^T + \rho_2^T \le \theta +1$, in almost all cases $\rho = \pi + \square$, (meaning the Young diagram of $\rho$ can be obtained from a valid Young diagram $\pi \vdash n-1$ by adding a box), with $\pi_1^T + \pi_2^T \le \theta$; the exception is the case where $\theta$ is odd, $\rho_1^T = \rho_2^T = (\theta+1)/2$, and $\rho_{(\theta+1)/2} \ge 3$. We will deal with this exceptional case second, and the former case now. Let 
\begin{equation*}
    \mathrm{res}^{\mathbb{B}_{n,\theta}}_{S_{n-1}} [\chi^{\mathbb{B}_{n,\theta}}_{\lambda}]
    =
    \sum_{\pi \vdash n-1} \alpha^{n, \theta}_{\lambda, \pi} \chi_\pi^{S_{n-1}},
    \hspace{1cm} 
    \mathrm{res}^{\mathbb{B}_{n,\theta}}_{S_{n-1}} [\gamma^{\mathbb{B}_{n,\theta}}_{\lambda}]
    =
    \sum_{\pi \vdash n} \tilde{\alpha}^{n, \theta}_{\lambda, \pi} \chi_\pi^{S_{n-1}}.
\end{equation*}
    Note that in a similar way to part 1 of Remark \ref{rmk:usefuls}, $\alpha^{n, \theta}_{\lambda, \pi} \le \tilde{\alpha}^{n, \theta}_{\lambda, \pi}$ for all $\lambda, \pi$. Now fix a $\pi\vdash n-1$ with $\rho = \pi + \square$, with $\pi_1^T + \pi_2^T \le \theta$. We will exploit the fact that there are two ways to restrict from the Brauer algebra $\mathbb{B}_{n,\theta}$ to $\mathbb{C}S_{n-1}$; either by restricting first to $\mathbb{B}_{n-1,\theta}$, or first to $\mathbb{C}S_{n}$. Formulaically, using (\ref{eqn:branchingrules1}) and (\ref{eqn:branchingrules2}), this reads:
    \begin{equation}\label{eqn:alpha}
        \alpha^{n, \theta}_{\lambda, \pi}
        =
        \sum_{\substack{\overline{\lambda} = \lambda \pm \square \\ \overline{\lambda}_1^T + \overline{\lambda}_2^T \le \theta}}
        b^{n-1, \theta}_{\overline{\lambda}, \pi}
        =
        \sum_{\overline{\pi} = \pi + \square}
        b^{n, \theta}_{\lambda, \overline{\pi}},
    \end{equation}
    and
    \begin{equation}\label{eqn:alpha_tilde}
        \tilde{\alpha}^{n, \theta}_{\lambda, \pi}
        =
        \sum_{\overline{\lambda} = \lambda \pm \square}
        \tilde{b}^{n-1, \theta}_{\overline{\lambda}, \pi}
        =
        \sum_{\overline{\pi} = \pi + \square}
        \tilde{b}^{n, \theta}_{\lambda, \overline{\pi}}.
    \end{equation}
    Since $\pi_1^T + \pi_2^T \le \theta$, by part 3 of Remark \ref{rmk:usefuls}, each $\overline{\lambda}$ with $\tilde{b}^{n-1, \theta}_{\overline{\lambda}, \pi} >0$ must also have $\overline{\lambda}_1^T + \overline{\lambda}_2^T \le \theta$. Now the central sums in equations (\ref{eqn:alpha}) and (\ref{eqn:alpha_tilde}) are sums over the same set of partitions $\overline{\lambda}$. Now by the inductive assumption, each $b^{n-1, \theta}_{\overline{\lambda}, \pi} = \tilde{b}^{n-1, \theta}_{\overline{\lambda}, \pi}$, which gives $\alpha^{n, \theta}_{\lambda, \pi} = \tilde{\alpha}^{n, \theta}_{\lambda, \pi}$. Equating the right hand terms in the equations (\ref{eqn:alpha}) and (\ref{eqn:alpha_tilde}), and recalling that $b^{n, \theta}_{\lambda, \overline{\pi}} \le \tilde{b}^{n, \theta}_{\lambda, \overline{\pi}}$, we must have the equality $b^{n, \theta}_{\lambda, \overline{\pi}} = \tilde{b}^{n, \theta}_{\lambda, \overline{\pi}}$, for each $\overline{\pi} = \pi + \square$. Since $\rho = \pi + \square$, we are done.
    
    It remains to prove the lemma for the special case where $\theta$ is odd, $\rho_1^T = \rho_2^T = (\theta+1)/2$, and $\rho_{(\theta+1)/2} \ge 3$. Here, we let $\rho = \pi + \square$, where the differing square lies on row $(\theta+1)/2$. Now $\pi_1^T = \pi_2^T = (\theta+1)/2$. The equations (\ref{eqn:alpha}) and (\ref{eqn:alpha_tilde}) still hold, but now there exists one possible summand of the central sum in (\ref{eqn:alpha_tilde}) where $\overline{\lambda}_1^T + \overline{\lambda}_2^T > \theta$. This summand appears in the case when $\lambda_1^T = (\theta+1)/2$, $\lambda_2^T = (\theta-1)/2$, and the summand itself is $\overline{\lambda}$, obtained by adding a box in row $(\theta+1)/2$ (column 2). In other instances of $\lambda$, we use the same method as the first part of the proof. 
    
    Now, again employing the inductive assumption on the terms in the central sums of (\ref{eqn:alpha}) and (\ref{eqn:alpha_tilde}), we have that $\alpha^{n, \theta}_{\lambda, \pi} + \tilde{b}^{n-1, \theta}_{\overline{\lambda}, \pi}= \tilde{\alpha}^{n, \theta}_{\lambda, \pi}$, where $\lambda$ and $\overline{\lambda}$ are the specific partitions described above. Plugging this into the right hand sides of (\ref{eqn:alpha}) and (\ref{eqn:alpha_tilde}), we have
    \begin{equation}\label{eqn:midstep_in_useful_lemma_2}
        \tilde{b}^{n-1, \theta}_{\overline{\lambda}, \pi}
        +
        \sum_{\overline{\pi} = \pi + \square}
        b^{n, \theta}_{\lambda, \overline{\pi}}
        =
        \sum_{\overline{\pi} = \pi + \square}
        \tilde{b}^{n, \theta}_{\lambda, \overline{\pi}}.
    \end{equation}
    Let ${\pi^*}$ be $\pi$ with one box added in row $(\theta+1)/2 +1$ (column 1). Note that $\pi^*$ is the only $\overline{\pi} = \pi + \square$ satisfying $\overline{\pi}_1^T + \overline{\pi}_2^T > \theta +1$. We will prove that $b^{n, \theta}_{\lambda, \pi^*} + \tilde{b}^{n-1, \theta}_{\overline{\lambda}, \pi} = \tilde{b}^{n, \theta}_{\lambda, \pi^*}$. Then (\ref{eqn:midstep_in_useful_lemma_2}) becomes
    \begin{equation*}
        \sum_{\substack{\overline{\pi} = \pi + \square \\ \overline{\pi} \neq \pi^*}}
        b^{n, \theta}_{\lambda, \overline{\pi}}
        =
        \sum_{\substack{\overline{\pi} = \pi + \square \\ \overline{\pi} \neq \pi^*}}
        \tilde{b}^{n, \theta}_{\lambda, \overline{\pi}},
    \end{equation*}
    and similar to the first part of the proof, recalling $b^{n, \theta}_{\lambda, \overline{\pi}} \le \tilde{b}^{n, \theta}_{\lambda, \overline{\pi}}$ gives $b^{n, \theta}_{\lambda, \overline{\pi}} = \tilde{b}^{n, \theta}_{\lambda, \overline{\pi}}$ for all $\overline{\pi} = \pi + \square$, with $\overline{\pi}_1^T + \overline{\pi}_2^T \le \theta +1$. This covers $\rho$. So, it remains to prove $b^{n, \theta}_{\lambda, \pi^*} + \tilde{b}^{n-1, \theta}_{\overline{\lambda}, \pi} = \tilde{b}^{n, \theta}_{\lambda, \pi^*}$.
    
    Now, Okada \cite{okada_pieri_type_rules} gives an explicit algorithm for calculating $b^{n, \theta}_{\lambda, \pi^*}$. Working through that algorithm, we find that $b^{n, \theta}_{\lambda, \pi^*}  = \tilde{b}^{n, \theta}_{\lambda, \pi^*} - \tilde{b}^{n, \theta}_{\hat{\lambda}, \pi^*}$, where $\hat{\lambda}$ is obtained from $\lambda$ by adding two boxes, one in each of the first two columns. Now it is straightforward to see that $\tilde{b}^{n, \theta}_{\hat{\lambda}, \pi^*} = \tilde{b}^{n-1, \theta}_{\overline{\lambda}, \pi}$, since $\pi^* \setminus \hat{\lambda}$ and $\pi \setminus \overline{\lambda}$ are identical skew-diagrams (remark \ref{rmk:usefuls}). This completes the proof.
\end{proof}

We can now determine the sets $P_n(\theta)$ for $\theta=2,3$.

\subsection{Spin $\frac{1}{2}$; $\theta=2$}

Recall $\Lambda_n(\theta)$ is the set of pairs of partitions $(\lambda, \rho)$, $\lambda \vdash n-2k$, $0 \le k \le \lfloor \frac{n}{2} \rfloor$, $\rho \vdash n$, such that $\lambda_1^T + \lambda_2^T \le \theta$ and $\rho_1^T \le \theta$. Recall the set $P_n(\theta)$ is given by $(\lambda, \rho) \in \Lambda_n(\theta)$ such that $b^{n,\theta}_{\lambda, \rho}>0$, where $b^{n,\theta}_{\lambda, \rho}$ is the coefficient of the irreducible $\chi_{\rho}^{S_n}$ in the restriction of $\chi_{\lambda}^{\mathbb{B}_{n,\theta}}$ from $\mathbb{B}_{n,\theta}$ to $\mathbb{C}S_n$. 
\begin{prop}\label{prop:P2}
        For $\theta=2$, the $\mathbb{C}S_n$-$\mathbb{B}_{n,2}$ branching coefficient $b_{\lambda,\rho}^{n,2}$ is strictly positive if and only if $\lambda_1 \le \rho_1 - \rho_2$, with the exceptions of $\lambda = \varnothing$ or $\lambda = (1,1)$, in which case both rows of $\rho$ must be even or odd, respectively. Hence $\frac{1}{n}P_n(2)\to\Delta_2^*$ in the Hausdorff distance, where
        \begin{equation*}
            \Delta^*_{2} = \{ (x,y) \in ([0,1]^2)^2 \ | 
                \ x_1 \ge x_2, \ x_1 + x_2 =1, \ y_2 = 0,\ 
                \ 0 \le y_1 \le x_1 - x_2 \}.
        \end{equation*}
\end{prop}
\begin{proof}
    We prove first that the irreducible representation $\psi_{(n-2k)}^{\mathbb{B}_{n,2}}$ of $\mathbb{B}_{n,2}$ restricts to the symmetric group as:
    \begin{equation}\label{theta=2restrictionthm}
                \mathrm{res}^{\mathbb{B}_{n,2}}_{S_n} [\chi^{\mathbb{B}_{n,2}}_{(n-2k)}] =
            \sum_{i=0}^{k} \chi^{S_n}_{(n-i,i)}.
    \end{equation}
    Indeed, by Remark \ref{rmk:recurrence_relation} (using $(n-2k)' = (n-2k)$) and Lemma \ref{lem:rho_columns_<=_theta+1}, we have $b^{n,2}_{(n-2k), (n-i,i)} = b^{n-2i,2}_{(n-2k), (n-2i)} = \tilde{b}^{n-2i,2}_{(n-2k), (n-2i)} = \mathbbm{1}\{0 \le i \le k \}$, the last equality coming from (\ref{eqn:ramsrestrictionofcells}) and the definition of the Littlewood-Richardson coefficients. Combining (\ref{theta=2restrictionthm}) and Okada's result in Remark \ref{rmk:restrict_one_column_okada} gives the first part of the proof. The second is a straightforward application of the definition of the Hausdorff distance.
\end{proof}

The proof of Proposition \ref{prop:P2} also implies the following corollary.

\begin{cor}\label{lem:bgoesto0,theta2}
    We have that for $\theta =2$, $\lim_{n \to \infty} \frac{1}{n} \log( \max_{(\lambda, \rho) \in \Lambda_n(2)} b^{n,2}_{\lambda, \rho})$ = 0.
\end{cor}

\subsection{Spin $1$; $\theta = 3$} 

\begin{prop}\label{prop:P3}
    For $\theta=3$, the $\mathbb{C}S_n$-$\mathbb{B}_{n,3}$ branching coefficient $b_{\lambda,\rho}^{n,3}$ is strictly positive if and only if $\lambda_1 \le \rho_1 - \rho_3$, with the following exceptions:
    \begin{enumerate}
        \item 
        If $\lambda = (n-2k)$, and $\rho_2 = \rho_3$ odd, or $\rho_1 = \rho_2$ odd, then $b_{\lambda,\rho}^{n,3}=0$;
        \item
        If $\lambda = (n-2k-1,1)$, and $\rho_2 = \rho_3$ even, or $\rho_1 = \rho_2$ even, then $b_{\lambda,\rho}^{n,3}=0$;
        \item
        If $\lambda = (1^j)$, $j=0, \dots, 3$, then $b_{\lambda,\rho}^{n,3}>0$ if and only if $\rho$ has $j$ odd parts. 
    \end{enumerate}
    As a consequence, $\frac{1}{n}P_n(3)\to\Delta_3^*$ in the Hausdorff distance, where
        \begin{equation*}
            \Delta^*_{3} = \{ (x,y) \in ([0,1]^3)^2 \ | 
                \ x_1 \ge x_2 \ge x_3, \ x_1 + x_2 + x_3=1, \ y_2 = y_3 = 0,\ 
                \ 0 \le y_1 \le x_1 - x_3 \}.
        \end{equation*}
\end{prop}
\begin{proof}
From Remark \ref{rmk:recurrence_relation}, we see that if $\lambda_1>\rho_1-\rho_3$ then $b_{\lambda,\rho}^{n,3}=0$. For the rest of the first part of the Proposition, let $\lambda = (n-2k)$ or $(n-2k-1,1)$, and let $\lambda_1 \le \rho_1 - \rho_3$. Then, using Remark \ref{rmk:recurrence_relation} and Lemma \ref{lem:rho_columns_<=_theta+1}, $b^{n,3}_{\lambda, (\rho_1, \rho_2, \rho_3)} 
= b^{n-3\rho_3,3}_{\lambda^*, (\rho_1 - \rho_3, \rho_2 - \rho_3)} 
= \tilde{b}^{n-3\rho_3,3}_{\lambda^*, (\rho_1 - \rho_3, \rho_2 - \rho_3)}$, where $\lambda^* = \lambda$ if $\rho_3$ even, $\lambda^* = \lambda'$ if $\rho_3$ odd. The cases where $\lambda^* \nleq (\rho_1 - \rho_3, \rho_2 - \rho_3)$ (which give $b^{n,3}_{\lambda, (\rho_1, \rho_2, \rho_3)} =0$) are the cases: $\lambda = (n-2k)$, $\rho_2 = \rho_3$ odd, and $\lambda = (n-2k-1,1)$, $\rho_2 = \rho_3$ even. 

It remains to determine when $\tilde{b}^{n,3}_{\lambda, (\rho_1, \rho_2)}$ is non-zero. We need to show that if $\lambda \le \rho$, it is non-zero unless $\lambda = (n-2k)$, and $\rho_1 = \rho_2$ odd, or $\lambda = (n-2k-1,1)$, and $\rho_1 = \rho_2$ even. Recall the coeffiecient (from (\ref{eqn:ramsrestrictionofcells})) is given by 
\begin{equation}\label{eqn:last_lemma_step}
    \tilde{b}^{n,3}_{\lambda, (\rho_1, \rho_2)}
    =
    \sum_{\substack{\tau \vdash2k \\ \tau \ \mathrm{even}}} c_{\lambda, \tau}^{(\rho_1, \rho_2)}.
\end{equation}
Let us prove the $\lambda = (n-2k)$ case. By the Littlewood-Richardson rule (or its special case the Pieri rule - see Section I.9 of Macdonald \cite{macdonaldbook}), $\tilde{b}^{n,3}_{(n-2k), (\rho_1, \rho_2)}$ is equal to $|A|$, where $A$ is the set of even partitions $\tau \vdash 2k$, $\tau \le \rho$, such that $\rho \setminus \tau$ is a skew diagram with no two boxes the same column. Wlog $\tau = (2k-2m, 2m)$. For $\tau \in A$, we must have $0 \le 2m \le \rho_2$, and $\rho_2 \le 2k-2m \le \rho_1$. (Note that we certainly have $2k \ge \rho_2$, which follows from $n-2k \le \rho_1$). Now the only case where no such $\tau$ exists is when $\rho_1 = \rho_2$ odd, since in this case, any even $\tau$ must give $\rho \setminus \tau$ with two boxes in the last column. 
The case $\lambda = (n-2k-1,1)$ is obtained in a similar way.


The third special case $\lambda = (1^j)$ is given by Okada \cite{okada_pieri_type_rules} - see Remark \ref{rmk:restrict_one_column_okada}.

The final part of the theorem now follows by applying the first part, and the definition of the Hausdorff distance. 
\end{proof}

We also have the following corollary. 

\begin{cor}\label{lem:bgoesto0,theta3}
    We have that for $\theta =3$, $\lim_{n \to \infty} \frac{1}{n} \log( \max_{(\lambda, \rho) \in \Lambda_n(3)} b^{n,3}_{\lambda, \rho})$ = 0.
\end{cor}
\begin{proof}
By Lemma \ref{lem:recurrence_relation} and \ref{lem:rho_columns_<=_theta+1}, each non-zero $b^{n,3}_{\lambda,\rho}$ is equal to some $\tilde{b}^{m,3}_{\lambda',\rho'}$, where $m \le n$, $(\lambda', \rho') \in \Lambda_m(3)$, $\rho' \le \rho$. Now $\tilde{b}^{m,3}_{\lambda',\rho'}$ is (from (\ref{eqn:ramsrestrictionofcells})):
\begin{equation*}
    \tilde{b}^{m,3}_{\lambda',\rho'} 
    = 
    \sum_{\substack{\tau \vdash 2j \\ \tau \ \mathrm{even}}} c_{\lambda', \tau}^{\rho'}.
\end{equation*}
Since ${\rho'}_1^T \le 3$, the number of $\tau \vdash 2j$ with $\tau \le \rho'$ is bounded by $n^3$. Then the Littlewood-Richardson coefficient $c^{\rho'}_{\lambda', \tau}$ is bounded by $n^2$, since $\lambda'_2, \lambda'_3 \le 1$, and $\lambda_j =0$ for $j \ge 4$. Hence $\tilde{b}^{m,3}_{\lambda',\rho'}$ is bounded by $n^{5 +2}$, which gives the result.
\end{proof}

\appendix
\section{Numerical proof of Remark \ref{rmk:theremark}}\label{section:appendix}
Recall the function
\begin{equation*}
    w(z):=
    \frac{\partial \phi}{\partial x_1}(z)
    = 
    \frac{3}{2}
    + 
    \frac{\log(z)(1+5z)}{4(1-z)}
    +
    \log \left(
    \frac{-z\log(z)}{3(1-z) + (1+z)\log(z)}
    \right).
\end{equation*}
We need to prove that this function $w(z)$ is positive in the range $(r,1)$, where $r$ is the unique root of $3(1-z) + (1+z)\log(z)$ in $(0,1)$. This proof is due to Dave Platt.

Away from $r$ and $1$, this can be done straightforwardly using ARB, a C library for rigorous real and complex arithmetic (see \url{https://arblib.org/index.html}). We split the interval into small pieces, and use the program to show positivity on each piece. This works on the interval $[81714053/2^{30},1013243800/2^{30}]$. Near $r$, the function is large, and we can show by hand that it is positive. Indeed, $\log(z)(1+5z)/(4(1-z))$ and $\log(-z\log(z))$ are both increasing on the interval $[r, 81714053/2^{30}]$, and their sum, plus $3/2$, is easily bounded on the interval by $1$ in magnitude. Then $-\log(3(1-z) +(1-z)\log(z))$ is decreasing on the interval, and its value at $81714053/2^{30}$ is far larger than $1$. 

It remains to show that $w(z)$ is positive on $[1013243800/2^{30},1]$. The function's first three derivatives are zero at $1$, and the fourth is positive at $1$. We use the argument principle, and compute the integral $w'(z)/(2\pi i w(z))$ along a circle centre $1$ and radius $1/16$. There are no poles within this circle, and there are four zeros at $1$, so computing the integral to be $4$ implies there are no more zeros within the circle. We use a double exponential quadrature technique due to Pascal Molin. This approximates the integral to a sum with an explicit error term. We use Theorem 3.10 from \cite{molin:hal-00491561}, with $D=1$, $h=0.15$ and $n=91$, which, using ARB, gives the sum to be $[4.00000 \pm 5.24e-6] + [\pm 5.10e-6]*I$. The integral must be an integer by the argument principle, and $D=1$ means the explicit error term is at most $e^{-1}$, hence the integral must equal 4.

\section{Equivalence of $Q_{x,y}$ and $P_{x,y}$}
In this second appendix we study a second representation of $\mathbb{B}_{n,\theta}$, which we'll prove is isomorphic to the representation $\mathfrak{p}^{\mathbb{B}_{n,\theta}}$, for $\theta$ odd, and not isomorphic for $\theta$ even. Recall 
\begin{equation}
    \mathfrak{p}^{\mathbb{B}_{n,\theta}}(\overline{x,y})
    =
    Q_{x,y}, 
    \hspace{1cm}
    \mathfrak{p}^{\mathbb{B}_{n,\theta}}(x,y)
    =
    T_{x,y}.
\end{equation}
This will give the equivalence, in spin $1$, between our model with Hamiltonian (\ref{eqn:fullhamiltonian}), and the bilinear-biquadratic Heisenberg model with Hamiltonian (\ref{eqn:spin1hamiltonian}); equality of their partition functions was proved by Ueltschi (\cite{Ueltschi_2013}, Theorem 3.2).

Recall $\langle a_x, a_y| P_{x,y} |b_x, b_y \rangle = (-1)^{a_x-b_x}\delta_{a_x, -a_y} \delta_{b_x, -b_y}$. Define $\tilde{\mathfrak{p}}^{\mathbb{B}_{n,\theta}}: \mathbb{B}_{n,\theta} \to \mathrm{End}(\mathcal{H}_V)$, given by
\begin{equation}\label{eqn:defn_of_B_n's_action}
    \tilde{\mathfrak{p}}^{\mathbb{B}_{n,\theta}}(\overline{x,y})
    =
    P_{x,y}, 
    \hspace{1cm}
    \tilde{\mathfrak{p}}^{\mathbb{B}_{n,\theta}}(x,y)
    =
    T_{x,y}.
\end{equation}
\begin{lem}\label{lem:isomorphic_reps}
    For $\theta$ odd, and all $n$, the representations $\tilde{\mathfrak{p}}^{\mathbb{B}_{n,\theta}}$ and $\mathfrak{p}^{\mathbb{B}_{n,\theta}}$ of $\mathbb{B}_{n,\theta}$ are isomorphic via a unitary transformation, and for $\theta$ even, the two are not isomorphic.
\end{lem}
\begin{proof}
    Since the elements $(x,y)$ and $(\overline{x,y})$ generate the algebra $\mathbb{B}_{n,\theta}$, it suffices to find an invertible linear function $\psi_n: \mathcal{H}_V \to \mathcal{H}_V$ such that 
    \begin{equation}\label{eqn:isomorphic_reps_condition}
        \psi_nT_{x,y}\psi_n^{-1} = T_{x,y}, \ \  \psi_nQ_{x,y}\psi_n^{-1} = P_{x,y},
    \end{equation}
    for all $x,y$. By the Schur-Weyl duality for the general linear and symmetric groups (\ref{eqn:SW-Sn}), the first condition holds if and only if $\psi_n = \psi^{\otimes n}$ for some $\psi \in GL(\theta)$. Then the second condition also holds if and only if $\psi^{\otimes 2}Q_{x,y}(\psi^{\otimes 2})^{-1} = P_{x,y}$ for all $x,y$, which holds if and only if:
    \begin{equation*}
        \begin{split}
            (-1)^{a_x-b_x}\delta_{a_x, -a_y} \delta_{b_x, -b_y}
            &=
            \sum_{r_x,r_y,s_x,s_y} \psi_{a_x,r_x} \psi_{a_y,r_y} 
            \delta_{r_x, r_y} \delta_{s_x, s_y}
            (\psi^{-1})_{s_x,b_x} (\psi^{-1})_{s_y,b_y}\\
            &=
            \sum_{r,s} \psi_{a_x, r} \psi_{a_y, r}
            (\psi^{-1})_{s,b_x} (\psi^{-1})_{s,b_y}\\
            &=
            (\psi \psi^T)_{a_x, a_y} ((\psi^{-1})^T (\psi^{-1}))_{b_x, b_y}.
        \end{split}
    \end{equation*}
    Hence the two representations are isomorphic if and only if there exists an invertible $\theta\times\theta$ matrix $\psi$ such that 
    \begin{equation*}
        \psi \psi^T 
        =
        \begin{bmatrix}
            & & & & (-1)^S\\
            & & & (-1)^{S-1} & \\
            & & \reflectbox{$\ddots$} & & \\
            & (-1)^{1-S} & & & \\
            (-1)^{-S} & & & &
        \end{bmatrix},
    \end{equation*}
    and 
    \begin{equation*}
        (\psi^{-1})^T \psi^{-1} 
        =
        \begin{bmatrix}
            & & & & (-1)^S\\
            & & & (-1)^{S-1} & \\
            & & \reflectbox{$\ddots$} & & \\
            & (-1)^{1-S} & & & \\
            (-1)^{-S} & & & &
        \end{bmatrix}^T,
    \end{equation*}
    where recall $\theta=2S+1$.
    For $\theta$ odd the two matrices on the right hand sides above are the same, so it suffices to note that we can set the central entry in $\psi$ to be $1$, and the rest to be made up of nested invertible 2x2 matrices $g_1$, $g_2$, given by, for example,
    \begin{equation*}
        g_1
        =
        \frac{1}{\sqrt{2}}
        \begin{bmatrix}
            -1 & i\\
            -1 & -i
        \end{bmatrix},
        \ \ 
        g_2
        =
        \frac{1}{\sqrt{2}}
        \begin{bmatrix}
            -1 & i\\
             1 & i
        \end{bmatrix},
    \end{equation*}
    since
    \begin{equation*}
        g_1 g_1^T
        =
        \begin{bmatrix}
            0 & 1\\
            1 & 0
        \end{bmatrix},
        \ \ 
        g_2 g_2^T
        =
        \begin{bmatrix}
            0 & -1\\
            -1 & 0
        \end{bmatrix}.
    \end{equation*}
    This shows that for $\theta$ odd, the representations $\tilde{\mathfrak{p}}^{\mathbb{B}_{n,\theta}}$ and $\mathfrak{p}^{\mathbb{B}_{n,\theta}}$ are indeed isomorphic, and since $g_1$, $g_1$ are unitary, so is $\psi_n$. For $\theta$ even, there are fractional powers of $(-1)$ appearing, so we have to make a choice, say, of $(-1)^{\frac{1}{2}}=\pm1$, and then the rest of the entries are determined by $(-1)^a=(-1)^{a-\frac{1}{2}}(-1)^{\frac{1}{2}}$. Whichever we choose though, $\psi \psi^T$ will always be a symmetric matrix, and 
    \begin{equation*}
        \begin{bmatrix}
            & & & & (-1)^S\\
            & & & (-1)^{S-1} & \\
            & & \reflectbox{$\ddots$} & & \\
            & (-1)^{1-S} & & & \\
            (-1)^{-S} & & & &
        \end{bmatrix}
    \end{equation*}
    will always be anti-symmetric (and non-zero), so the two cannot be equal. This concludes the proof.\\
    \end{proof}
    
    \begin{rmk}\label{rmk:appendix-spin-conj}
        When $S=1$ ($\theta=3$), if $\psi_n=\psi^{\otimes n}$ is that of Lemma \ref{lem:isomorphic_reps}, then we have that $\psi^{-1}S^k\psi=W$ is antisymmetric (its transpose is its negative), for each $k=1,2,3$. This means that we have the following conjugation of the Hamiltonian (\ref{eqn:spin1hamiltonian}) of $S=1$ bilinear-biquadratic model with external magnetisation in the $S^k$ direction, to get the Hamiltonian (\ref{eqn:fullH-mag}):
        \begin{equation*}\begin{split}
            &\psi_n^{-1}\left(- \left( 
            \sum_{x,y} J_1 S_x \cdot S_y + J_2 (S_x \cdot S_y)^2 \right) 
            - h\sum_{x}S_x^k\right)\psi_n\\
            =&\psi_n^{-1}\left(- \left( 
            \sum_{x,y} J_1 T_{x,y} + (J_2-J_1) P_{x,y} \right) 
            - h\sum_{x}S_x^k\right)\psi_n\\
            =&- \sum_{x,y} \left( L_1 T_{x,y} + L_2 Q_{x,y}\right)- h \sum_{x}W_x,
        \end{split}\end{equation*}
        where $L_1=J_1$, $L_2=J_2-J_1$.
        Since the spin matrices have eigenvalues $\{1,0,-1\}$, the resulting antisymmetric $W$ of (\ref{eqn:fullH-mag}) does too. Hence Theorem \ref{thm:mag} indeed gives the magnetisation in the $S^k$ direction in the model (\ref{eqn:spin1hamiltonian}), and Theorem \ref{thm:total-spin} gives (the exponential of) the total spin in the $S^k$ direction in the same model.
    \end{rmk}

\section*{Acknowledgements}
I would like to thank Davide Macera, Sasha Sodin, Daniel Ueltschi, Jakob Björnberg, and Hjalmar Rosengren for many useful discussions, Dave Platt for the above numerical proof, Arum Ram for sharing some unpublished work, and Tom Mutimer for advice on using numerical software. My thanks to the referee who made several useful suggestions. This research was supported by the EPSRC Studentship 1936327.

\bibliographystyle{plain}
\bibliography{bibx}

\end{document}